\documentclass[10pt,draftcls,onecolumn]{IEEEtran}
\ifCLASSINFOpdf
\else
   \usepackage[dvips]{graphicx}
\fi
\usepackage{url}
\usepackage{graphicx}
\usepackage{amsmath}
\usepackage{comment}
\usepackage[utf8]{inputenc}
\usepackage[english]{babel}
\usepackage{color}
\usepackage{framed}
\usepackage{cite}
\usepackage{amssymb,amsfonts}
\usepackage{textcomp}
\usepackage{cite}

\usepackage{bbold}
\usepackage{ifpdf}
\usepackage{times}
\usepackage{layout}
\usepackage{float}
\usepackage{afterpage}
\usepackage{amstext}
\usepackage{amssymb,bm}
\usepackage{latexsym}
\usepackage{color}
\usepackage{flexisym}
\usepackage{kantlipsum}
\usepackage{amsthm}
\usepackage{upgreek}

\usepackage[caption=false,font=footnotesize]{subfig}

\usepackage{booktabs}
\usepackage{multicol}

\usepackage{enumitem}
\usepackage{amsmath}
\usepackage{algpseudocode}
\usepackage{amsthm}
\usepackage{dsfont}
\usepackage{algorithm}
\usepackage{thmtools}
\usepackage{amssymb}
\usepackage[dvipsnames]{xcolor}

\def\bGamma{\boldsymbol{\Gamma}}

\def\bSigma{\boldsymbol{\Sigma}}

\def\bOmega{\boldsymbol{\Omega}}

\def\mba{\mathbf{a}}
\def\mbb{\mathbf{b}}
\def\mbc{\mathbf{c}}

\def\mbe{\mathbf{e}}

\def\mbg{\mathbf{g}}

\def\mbn{\mathbf{n}}

\def\mbp{\mathbf{p}}
\def\mbq{\mathbf{q}}
\def\mbr{\mathbf{r}}
\def\mbs{\mathbf{s}}
\def\mbt{\mathbf{t}}
\def\mbu{\mathbf{u}}
\def\mbv{\mathbf{v}}

\def\mbx{\mathbf{x}}
\def\mby{\mathbf{y}}
\def\mbz{\mathbf{z}}

\def\mbA{\mathbf{A}}
\def\mbB{\mathbf{B}}
\def\mbC{\mathbf{C}}

\def\mbG{\mathbf{G}}
\def\mbH{\mathbf{H}}
\def\mbI{\mathbf{I}}

\def\mbL{\mathbf{L}}
\def\mbM{\mathbf{M}}
\def\mbN{\mathbf{N}}
\def\mbO{\mathbf{O}}
\def\mbP{\mathbf{P}}
\def\mbQ{\mathbf{Q}}
\def\mbR{\mathbf{R}}
\def\mbS{\mathbf{S}}

\def\mbU{\mathbf{U}}
\def\mbV{\mathbf{V}}
\def\mbW{\mathbf{W}}
\def\mbX{\mathbf{X}}

\def\mbZ{\mathbf{Z}}

\newtheorem{theorem}{Theorem}
\newtheorem{proposition}{Proposition}
\newtheorem{lemma}{Lemma}
\newtheorem{corollary}{Corollary}

\theoremstyle{definition}
\newtheorem{definition}{Definition}

\algnewcommand\algorithmicinput{\textbf{Input:}}
\algnewcommand\Input{\item[\algorithmicinput]}
\algnewcommand\algorithmicoutput{\textbf{Output:}}
\algnewcommand\Output{\item[\algorithmicoutput]}
\algnewcommand\algorithmicinit{\textbf{Initialize:}}
\algnewcommand\Init{\item[\algorithmicinit]}

\makeatletter
\newcommand*{\rom}[1]{\expandafter\@slowromancap\romannumeral #1@}
\makeatother

\begin{document}

\title{Harnessing the Power of Sample Abundance: Theoretical Guarantees and Algorithms for Accelerated One-Bit Sensing}
\author{Arian Eamaz, \IEEEmembership{Student Member, IEEE}, Farhang Yeganegi, Deanna Needell, \IEEEmembership{Member, IEEE}, and \\ Mojtaba Soltanalian, \IEEEmembership{Senior Member, IEEE}
\thanks{This work relates in part to Department of Navy Award N00014-22-1-2666 issued by the Office of Naval Research. Any opinions, findings, conclusions, or recommendations expressed in
this material are those of the authors and do not necessarily reflect the views of the
Office of Naval Research. A conference precursor to this work was presented at the 2023 IEEE International Symposium on Information Theory (ISIT). The first two authors contributed equally to this work.}
\thanks{A. Eamaz, F. Yeganegi, and M. Soltanalian are with the Department of Electrical and Computer Engineering, University of Illinois Chicago, Chicago, IL 60607, USA (\emph{Corresponding author: Arian Eamaz}).
}
\thanks{D. Needell is with the Department of Mathematics, University of California Los Angeles, Los
Angeles, CA 90095, USA.}
}

\markboth{Submitted to the IEEE TRANSACTIONS ON INFORMATION THEORY, 2023
}
{Shell \MakeLowercase{\textit{et al.}}: Bare Demo of IEEEtran.cls for IEEE Journals}
\maketitle

\begin{abstract}
One-bit quantization with time-varying sampling thresholds (also known as random dithering)
has recently found significant utilization potential in statistical signal processing applications due to its relatively low
power consumption and low implementation cost. In addition
to such advantages, an attractive feature of one-bit analog-to-digital converters (ADCs) is their superior sampling rates as compared to their conventional multi-bit counterparts. This characteristic endows one-bit signal processing frameworks with what one may refer to as \emph{sample abundance}.  We show that sample abundance plays a pivotal role in many signal recovery and optimization problems that are formulated as (possibly non-convex) quadratic programs with linear feasibility constraints. Of particular interest to our work are low-rank matrix recovery and compressed sensing applications that take advantage of one-bit quantization. We demonstrate that the sample abundance paradigm allows for the transformation of such problems to merely linear feasibility problems by forming large-scale overdetermined linear systems---thus removing the need for handling costly optimization constraints and objectives. To make the proposed computational cost savings achievable, we offer enhanced randomized Kaczmarz algorithms to solve these highly overdetermined feasibility problems and provide theoretical guarantees in terms of their convergence, sample size requirements, and overall performance. Several numerical results are presented to illustrate the effectiveness of the proposed
methodologies.  
\end{abstract}


\IEEEpeerreviewmaketitle

\section{Introduction}
\IEEEPARstart{W}{e} consider an optimization problem of the form:
\begin{equation}
\label{eq:1nnnn}
\begin{aligned}
\underset{\mbX \in \Omega_{c}}{\textrm{minimize}}\quad f(\mbX)\quad
\text{subject to} \quad &\mathcal{A}\left(\mbX\right)=\mathbf{y},
\end{aligned}
\end{equation}
where $f(.)$ is a cost function, $\Omega_{c}$ is a feasible set, $\mbX\in \mathbb{C}^{n_{1}\times n_{2}}$ is the matrix of unknowns, $\mathbf{y}\in \mathbb{R}^{n}$ is the measurement vector, and $\mathcal{A}$ is a linear transformation mapping $\mathbb{C}^{n_{1}\times n_{2}}$ into $\mathbb{R}^{n}$. This problem emerges in a wide variety of applications, particularly as the relaxed variant of some well-known NP-hard problems. Although many problems can be expressed in
the form in the program (\ref{eq:1nnnn}), the applications we will focus on in this paper
include some specific problems of interest which can take advantage of low-resolution (and particularly one-bit) sampling and processing:
\begin{itemize}
\item The task of \emph{recovering a low-rank matrix from its linear measurements} plays a central role in computational science. The problem occurs in many areas of applied mathematics, such as signal processing \cite{candes2015phase,davenport2016overview,candes2011tight}. In this scenario, the cost function of the program
(\ref{eq:1nnnn}), $f(.)$, is typically the nuclear norm or the Frobenius norm, and the constraint set $\Omega_{c}$ would be an amplitude restriction on the elements of matrix $\mbX$; see \cite{recht2011null}.
    \item \emph{Compressed sensing} (CS) offers a framework for the simultaneous sensing and compression
of finite-dimensional vectors by relying on linear dimensionality reduction. Through a CS formulation, sparse signals may be recovered from highly incomplete measurements \cite{candes2006robust}. The problem (1) can be adopted in the CS context when $f\left(\mbX\right)=\left\|\operatorname{vec}\left(\mbX\right)\right\|_{1}$.
\end{itemize}
This paper will primarily concentrate on relaxation problems in which the objective function $f(\cdot)$ is relaxed to address the NP-hard cost caused due to the associated  feasibility constraints. Such relaxation techniques prove significantly helpful in various applications, such as nuclear norm minimization for low-rank matrix recovery and $\ell_1$ minimization for compressed sensing problems.

Sampling the signals of interest at high data rates with high-resolution ADCs would dramatically increase the overall manufacturing cost and power consumption of such ADCs. In multi-bit sampling scenarios, a very large number of quantization levels is necessary in order to represent the original continuous signal in with high accuracy, which in turn leads to a considerable reduction in sampling rate \cite{boufounos2015quantization,eamaz2021modified}. This attribute of multi-bit sampling is the key reason for the general emergence of underdetermined systems $n_{1}n_{2}\geq n$ in the program~(\ref{eq:1nnnn}) \cite{candes2013phaselift,eamaz2022phase}. An alternative solution to such challenges is to deploy \emph{one-bit quantization} which is a coarse sensing scenario, where the signals are merely compared with predefined threshold levels at the ADCs, producing sign data ($\pm1$). This enables signal processing equipments to sample at a very high rate, with a considerably lower cost and energy consumption, compared to their counterparts which employ multi-bit ADCs \cite{mezghani2018blind,eamaz2021modified}. Several applications abound of one-bit ADCs, such as multiple-input multiple-output wireless communications \cite{mezghani2018blind}. Various methods for one-bit quantization, including our own, offer distinct trade-offs between computational and implementation complexity and reconstruction accuracy. Take the Sigma-Delta approach as an example, which necessitates memory elements to store state-variables associated with thresholds. This sequential quantization of incoming measurements demands significant memory and computational power, causing delays in acquiring measurements while the thresholds are being updated. Therefore, sensors used to collect measurements must possess substantial computational capabilities to accommodate such requirements \cite{knudson2016one}. 

The effectiveness of incorporating random thresholds (dithers) within the framework of one-bit quantization has been extensively established in various contexts \cite{dabeer2006signal}.
Research findings suggest that a carefully chosen dither signal can significantly improve the resolution and performance of digital instrumentation \cite{carbone1997quantitative}.
The implementation of a dithered generator in quantization systems within an ADC system, was described in \cite{robinson2019analog}.

Recent works  \cite{eamaz2021modified,eamaz2023covariance,eamaz2022covariance} considered Gaussian dithering for recovering covariance from one-bit measurements. In \cite{dirksen2022covariance}, the authors provide a comprehensive investigation of estimating the covariance matrix from one-bit measurements using uniform dithering. Dithered One-bit quantization has also been applied to various problems, such as sparse parameter estimation \cite{dabeer2006signal}, phase retrieval \cite{eamaz2022phase}, and sampling theory \cite{eamaz2022uno} as discussed in contemporary literature. The study presented in \cite{eamaz2022covariance} demonstrates that one-bit signal parameter estimation with time-varying thresholds yields notably superior results compared to using constant thresholds. 

Randomized iterative algorithms have become a widely employed technique for dimension reduction, encompassing various iterative methods for solving linear systems and their variations, as well as extensions to non-linear optimization problems \cite{martinsson2020randomized}. These algorithms include well-known methods such as randomized Kaczmarz algorithm (RKA) \cite{strohmer2009randomized}, and coordinate descent \cite{leventhal2010randomized, ma2015convergence}.
\emph{In this approach, instead of utilizing all available information at each iteration, a sketched subset of the information is used, selected through a sketching process.} By multiplying the data with a random matrix, randomized sketching effectively reduces the size of the problem while preserving its essential characteristics.
Common examples of sketching matrices include Gaussian matrices with independent entries and other matrices that satisfy the Johnson-Lindenstrauss property. Another example is block-identity matrices, which effectively subsample the data matrix.
The Kaczmarz method \cite{kaczmarz1937angenaherte} is an iterative projection algorithm for solving linear systems of equations and inequalities. It is usually applied to highly overdetermined systems because of its simplicity.
Many variants of this iterative method and their convergence rates have been proposed and studied in recent decades for both consistent and inconsistent systems including the randomized Kaczmarz algorithm, the randomized block Kaczmarz algorithm and most recently, the sampling Kaczmarz-Motzkin method (SKM) \cite{strohmer2009randomized,leventhal2010randomized,needell2014paved,de2017sampling}.

\subsection{Prior Art}
Our work aligns seamlessly with the recent advancements in the exploration of coarse quantization for parameter estimation from random measurements. The emergence of one-bit compressed sensing (one-bit CS) has led to a plethora of papers that have greatly enhanced our understanding of employing dithering for one-bit signal reconstruction.

The domain of one-bit low-rank matrix sensing has garnered considerable interest, yet it remains relatively unexplored. Several noteworthy papers, such as \cite{foucart2019recovering}, have shed light on this intriguing subject. In \cite{foucart2019recovering}, the authors delve into the theoretical guarantees of one-bit sensing from low-rank matrices, employing two distinct algorithms: one based on hard singular value thresholding and the other utilizing semidefinite programming.
The initial investigation in this study considers a scenario \emph{without} thresholds and subsequently progresses to the Gaussian sampling matrix and Gaussian dithering scenario, where both adaptive and non-adaptive thresholds are employed. However, the recovery algorithms and theoretical guarantees presented in this study require the availability of side information regarding the signal of interest, specifically an upper bound on the Frobenius norm of input signal. This requirement contradicts the essence of employing dithering in one-bit sensing, as the literature suggests that the primary advantage of dithering lies in the ability to recover the signal magnitude without additional information.

One-bit CS was first introduced in \cite{boufounos20081}. Later, in \cite{jacques2013robust}, this method was extensively investigated, and they presented a lower bound on the reconstruction error, alongside proposing heuristic algorithms for recovering the underlying signals. Extensive research has been dedicated to the problem of one-bit CS \emph{without} dithering, resulting in numerous proposed approaches for its solution. Notably, the binary iterative hard thresholding (BIHT) and its normalized version, normalized BIHT (NBIHT), have emerged as prominent methods in this field \cite{boufounos2015quantization,plan2013one,jacques2013quantized,jacques2013robust,friedlander2021nbiht}.
In some earlier works, the ditherless setup was taken into account, leading to the derivation of various theoretical guarantees for reconstruction algorithms \cite{jacques2013robust,ai2014one,friedlander2021nbiht,liu2019one}. For instance,
the authors of \cite{jacques2013robust} presented the required number of measurements for robust sparse signal reconstruction within a ball 
around the original signal, specifically focusing on Gaussian sampling matrices.  
In \cite{plan2014dimension}, authors proposed the \emph{random hyperplane tessellations} to find an embedding between the Hamming distance and direction error recovery in a probability. In \cite{oymak2015near}, the authors characterized the tradeoff between the distortion of random hyperplane tessellations and sample complexity in terms of Gaussian complexity, addressing both arbitrary and structured sets.  Nevertheless, their random hyperplane tessellations suffer from two primary limitations: they focus on the ditherless scenario and only-direction estimation, and they exclusively consider Gaussian sampling matrices. However,
recent studies have shown that under appropriate conditions, complete signal reconstruction can be accomplished by employing nonzero random thresholds \cite{knudson2016one,baraniuk2017exponential,xu2020quantized}. In \cite{jacques2017small}, quasi-isometric embeddings were introduced with high probability through scalar (dithered) quantization following linear random projection. In such embeddings, both multiplicative and additive distortions coexist when measuring distances between mapped vectors using the $\ell_1$-norm.
The authors of \cite{knudson2016one} introduced a sensing method that partitions the measurements into two sets and utilizes a constant threshold in one-bit data acquisition, 
thereby estimating the signal norm. Simultaneously, this method employs a zero threshold for the remaining one-bit data acquisition part,
allowing for the estimation of the direction of the signal. In the context of Gaussian random sampling matrices, the suggestion in \cite{baraniuk2017exponential} involves introducing either adaptive or random dither to the compressive measurements of a signal before binarization. In a recent investigation \cite{xu2020quantized}, a scalar quantization scheme coupled with a uniformly distributed dithering signal was found to extend theoretical frameworks beyond Gaussian measurements. The study provided theoretical guarantees for sampling matrices meeting the restricted isometry property (RIP). Notably, their approach operates without requiring the high-resolution limit (HRA), making it applicable across various bit-rates. In \cite{derezinski2022sharp}, they developed the random hyperplane tessellations to find an
embedding between the Hamming distance and recovery error for uniform
dithering scenario. However, they only considered sun-gaussian matrix and could not provide any theoretical guarantees for the deterministic measurements.
In \cite{thrampoulidis2020generalized}, the authors tackled the problem of parameter estimation with dithered quantization by formulating it as a generalized LASSO problem. This approach allowed the authors to leverage the existing theoretical guarantees available in the literature for generalized LASSO.

\subsection{Motivations and Contributions of the Paper}
Our work was primarily motivated by two key references, \cite{xu2020quantized} and \cite{dirksen2022sharp}. Similar to these works, we explore the theoretical guarantees of one-bit sensing without being confined to a specific reconstruction algorithm. In this paper, we seek to include various sampling matrices, including \emph{deterministic} ones like the DCT. 
Specifically, our aim is to broaden the theoretical assurances to sampling matrices assumed to be solely isotropic, without the necessity of satisfying any matrix deviation inequality, such as RIP. In all the established theoretical assurances for the one-bit sensing problem, the focus has been on the scalar parameter of uniform dithers dominating the dynamic range of measurements. With our proposed scheme, we aim to illustrate the consequences when the scalar parameter fails to meet the dynamic range requirements and how this impacts the upper bound of recovery error.
Additionally, we delve into the scenario of \emph{sample abundance}, a common situation in one-bit sensing, where we reformulate the relaxed optimization problem \eqref{eq:1nnnn} into a linear feasibility problem.
Using our theorems, we aim to address the question of \emph{how many samples are necessary for a linear feasibility problem solver to obtain a solution within the ball with a specific radius, around the original signal}. 

Moreover, we endeavor to harness the advantages of \emph{randomized algorithms} for one-bit sensing, a novel endeavor where we explore their theoretical guarantees and convergence analysis in this context for the first time. 
It is important to note that all the theoretical guarantees derived in this paper are \emph{uniform reconstruction} guarantees. This means that they hold true for all desired signals in the given space, ensuring consistent and reliable reconstruction across the entire signal set.

This paper principally contributes to the following areas:\\
\textbf{1) RKA-based recovery for the dithered one-bit sensing.} 
In this paper, we consider the deployment of dithered one-bit quantization,leading to an increased sample size and a \emph{highly overdetermined system}. The proposed \emph{O}ne-bit aided \emph{R}andomized \emph{K}aczmarz \emph{Algorithm}, which we refer to as ORKA, can find the signal $\mbX$ in the program (\ref{eq:1nnnn}) by (i) generating abundant one-bit samples, in order to define a large scale overdetermined system where a finite volume feasible set is created for (\ref{eq:1nnnn}), and (ii) solving this linear feasibility problem by leveraging one of the efficient solver families, \emph{Kaczmarz algorithms}.
We propose \emph{two novel} variants of RKA which will be compared with the existing RKA variants. We conduct a thorough investigation into the robustness of RKA when dealing with a noisy linear inequality system. Our findings demonstrate that RKA for the linear inequality system, remains robust in the presence of noise, and we are able to obtain the upper recovery bound under this condition.
\\
\textbf{2) Theoretical Guarantees of dithered one-bit sensing.} 
We introduce the concept of the \emph{finite volume property} (FVP) to explore an upper bound on the probability of converging the average of distances between the signal and hyperplanes, to its mean. By leveraging the FVP, we can obtain the required number of samples needed for one-bit sensing to ensure that the solution lies within a ball around the original signal.
This required number of samples is investigated for both arbitrary and structured sets. For the first time in the literature, we are able to offer theoretical guarantees for measurements obtained from discrete cosine transform (DCT).\\
\textbf{3) Projected Kaczmarz algorithm-based recovery.}
Various sampling schemes present tradeoffs among the number of samples, computational cost of the reconstruction, and storage requirements. While the abundance of samples is desirable in scenarios involving dithered one-bit sensing, practical applications such as matrix completion and phase retrieval from coded diffraction patterns \cite{candes2015phase} impose restrictions on the number of available measurements. In such cases, it is not feasible to increase the number of measurements.
An alternative approach is to develop algorithms that require fewer samples, albeit at the cost of increased computational complexity. In our pursuit of a storage-friendly algorithm, we strive to avoid the burden of recording the full matrix per iteration. To achieve this, we leverage the power of two well-studied techniques: sketch-and-project and low-rank matrix factorization. 
We unveil a novel approach that combines low-rank matrix factorization with RKA.
A key advantage of our proposed algorithms is that they do not rely on any additional side information, such as known upper bounds on the signal norm. This makes them more practical for real-world applications.
Through comprehensive numerical evaluations, we demonstrate the superior performance of our proposed algorithms compared to state-of-the-art approaches.

\subsection{Organization and Notations}
In Section~\rom{2}, we will introduce our algorithm to solve (\ref{eq:1nnnn}), ORKA, which tackles the problem as a large-scale overdetermined system and finds the original signal in the one-bit polyhedron by an accelerated Kaczmarz approach. Moreover, two new variants of the Kaczmarz algorithms are proposed that enhance the convergence rate and the computational complexity of these solvers. In Section~\rom{3}, we delve into the robustness of the proposed algorithms in the presence of noise. We thoroughly discuss the impact of noise on the performance of the algorithms and address the need for novel algorithms specifically designed for noisy scenarios. As a representative application, in Section~\rom{4}, ORKA and other proposed algorithms will be applied in the context of low-rank matrix recovery and compressed sensing.

To investigate the theoretical guarantees  of dithered one-bit sensing, we will first introduce the FVP theorem in Section~\rom{5} based on the general Hoeffding's bound. In Section~~\rom{6}, the convergence of proposed randomized algorithms is discussed. 
Section~\rom{7} focuses on investigating the performance of the proposed algorithms in the tasks of one-bit low-rank matrix sensing and one-bit CS. We consider scenarios involving both noisy and noiseless conditions. Through comprehensive comparisons with state-of-the-art methods, we aim to demonstrate the superior performance and efficiency of our proposed algorithms. 
Finally, Section~\rom{8} concludes the paper.

\underline{\emph{Notation:}}
We use bold lowercase letters for vectors and bold uppercase letters for matrices. $\mathbb{C}$ and $\mathbb{R}$ represent the set of complex and real numbers, respectively. $(\cdot)^{\top}$ and $(\cdot)^{\mathrm{H}}$ denote the vector/matrix transpose, and the Hermitian transpose, respectively. $\mbI_{N}\in \mathbb{R}^{N\times N}$ is the identity matrix of size $N$. $\operatorname{Tr}(.)$ denotes the trace of the matrix argument. $\left\langle \mbB_{1},\mbB_{2}\right\rangle=\operatorname{Tr}(\mbB_{1}^{\mathrm{H}}\mbB_{2})$ is the standard inner product between two
matrices. The nuclear norm of a matrix $\mbB\in \mathbb{C}^{N_{1}\times N_{2}}$ is denoted $\left\|\mbB\right\|_{\star}=\sum^{M}_{i=1}\sigma_{i}$ where $M$ and $\left\{\sigma_{i}\right\}$ are the rank and singular values of $\mbB$, respectively. The Frobenius norm of a matrix $\mbB$ is defined as $\|\mbB\|_{\mathrm{F}}=\sqrt{\sum^{N_{1}}_{r=1}\sum^{N_{2}}_{s=1}\left|b_{rs}\right|^{2}}$ where $\{b_{rs}\}$ are elements of $\mbB$. The $\ell_{k}$-norm of a vector $\mathbf{b}$ is defined as $\|\mbb\|^{k}_{k}=\sum_{i}|b|^{k}_{i}$. The Hadamard (element-wise) product of two matrices $\mbB_{1}$ and $\mbB_{2}$ is denoted as $\mbB_{1}\odot \mbB_{2}$. Additionally, the  Kronecker product is denoted as $\mbB_{1}\otimes\mbB_{2}$. The vectorized form of a matrix $\mbB$ is written as $\operatorname{vec}(\mbB)$. $\mathbf{1}_{s}$ is the $s$-dimensional all-one vector. If there exists a $c>0$ such that $a\leq c b$ (resp. $a\geq c b$) for two quantities
$a$ and $b$, we have $a \gtrsim b$ (resp. $a \lesssim b$). Given a scalar $x$, we define $(x)^{+}$ as $\max\left\{x,0\right\}$. The set $[n]$ is defined as $[n]=\left\{1,\cdots,n\right\}$. $\operatorname{Diag}\left\{\mathbf{b}\right\}$ denotes a diagonal matrix with $\{b_{i}\}$ as its diagonal elements. A ball with radius $r$ centered at a point $\mby\in\mathbb{R}^{n}$ is defined as $\mathcal{B}_r\left(\mby\right)=\left\{\mby_1\in\mathbb{R}^{n}|\left\|\mby-\mby_1\right\|_2\leq r\right\}$. The function $\operatorname{sgn}(\cdot)$ yields the sign of its argument. The function $\log(\cdot)$ denotes the natural logarithm, unless its base is otherwise stated. For an event $\mathcal{E}$, $\mathbb{I}_{(\mathcal{E})}$ is the indicator function for that event meaning that $\mathbb{I}_{(\mathcal{E})}$ is $1$ if $\mathcal{E}$ occurs; otherwise, it is zero. The notation $x \sim \mathcal{U}_{[a,b]}$ means a random variable drawn from the uniform distribution over the interval $[a,b]$ and $x \sim \mathcal{N}(\mu,\sigma^2)$ represents the normal distribution with mean $\mu$ and variance $\sigma^2$. The sub-Gaussian norm of a random variable~$X$ characterized by
\begin{equation}
\label{STARIAN}
\|X\|_{\psi_2}=\inf \left\{t>0: \mathbb{E}\left\{e^{X^2/t^2}\right\} \leq 2\right\}.
\end{equation}
For an arbitrary set $\mathcal{K}$, the Gaussian complexity is defined as
\begin{equation}
\gamma\left(\mathcal{K}\right) = \mathbb{E}\left\{ \sup _{\mbx \in \mathcal{K}}\left|\langle \mbg, \mbx\rangle\right|\right\}, \quad \mbg \sim \mathcal{N}\left(0, \mbI_n\right).
\end{equation}
A covering number is the number of $r$-balls of a given size needed to completely cover a given set $\mathcal{K}$, i.e., $\mathcal{N}\left(\mathcal{K},\|\cdot\|_2, r\right)$. The Kolmogorov $r$-entropy of a set $\mathcal{K}$ is denoted by $\mathcal{H}\left(\mathcal{K},r\right)$ defined as
the logarithm of the size of the smallest $r$-net of $\mathcal{K}$ \cite{kolmogorov1959varepsilon}. The Hamming distance between $\operatorname{sgn}(\mbx),\operatorname{sgn}(\mby)\in\{-1,1\}^n$ 
is defined as 
\begin{equation}
\label{hamming}
d_{\mathrm{H}}(\operatorname{sgn}(\mbx),\operatorname{sgn}(\mby))=\frac{1}{n}\sum_{i=1}^{n}\mathbb{I}_{(\operatorname{sgn}(x_i)\neq \operatorname{sgn}(y_i))}.
\end{equation}

\section{ORKA: Towards Circumventing Costly Constraints }
\label{sec_3}
An interesting alternative to enforcing the feasible set $\Omega_{c}$ in (\ref{eq:1nnnn}) emerges when one increases the number of samples $n$, and solves the overdetermined linear system of equations with $n\geq n_{1}n_{2}$. In this sample abundance regimen, the linear constraint $\mathcal{A}\left( \mbX\right)=\mathbf{y}$ may actually yield the optimum inside $\Omega_{c}$. As a result of increasing the number of samples, it is possible that the intersection of these hyperplanes will achieve the desired point without the need to consider costly constraints. However, this idea may face practical limitations in the case of multi-bit quantization systems since ADCs capable of ultra-high rate sampling are difficult and expensive to produce.

In this section, by deploying the idea of one-bit quantization with time-varying thresholds, linear equality constraints are superseded by a massive array of linear inequalities in forming the feasible polyhedron. Therefore, by increasing the number of samples, a finite-volume space may be created inside $\Omega_{c}$ with shrinking size; making projections on $\Omega_{c}$ redundant. From a practical point of view, one-bit sampling is done efficiently at a very high rate with a significantly lower cost compared to its high-resolution counterpart. It has been examined in \cite{eamaz2022phase} that even though only partial information is made available to one-bit signal processing algorithms, they can achieve acceptable recovery performance with less complexity compared to the high-resolution scenario.

We first formulate the one-bit quantization with multiple dithering scheme. In Sections~\ref{sec_RKA} and \ref{sec_SKM}, we present a 
summarized review of RKA and Sampling Kaczmarz-Motzkin algorithm (SKM), respectively.
Then in Section~\ref{sec_preSKM}, we propose a novel Kaczmarz method variant formulated based on the 
\emph{SKM} 
and a \emph{preconditioning approach}. One-bit sampling via time-varying thresholds will be combined with the proposed randomized Kaczmarz method to create highly overdetermined linear inequalities. This paves the way for the recovery of the desired signal $\mbX$ in the program (\ref{eq:1nnnn}) without solving the original optimization problem; merely by tacking accounts of its linear constraints. 
Due to the block structure of the linear feasibility in ORKA, we will propose a block-based Kaczmarz algorithm accordingly in Section~\ref{ORKA}.

\subsection{One-Bit Quantization With Multiple Thresholding}
\label{one-bit}
Let $y_{k}=y(k\mathrm{T})$ denote the uniform samples of signal $y(t)$ with the sampling rate $1/\mathrm{T}$. In practice, the discrete-time samples occupy pre-determined quantized values. We denote the quantization operation on $y_{k}$ by the function $Q(\cdot)$. This yields the scalar quantized signal as $r_{k} = Q(y_{k})$.
In one-bit quantization, compared to zero or constant thresholds, time-varying sampling thresholds yield a better recovery performance \cite{eamaz2023covariance}. These thresholds may be chosen from any distribution.  
In the case of one-bit quantization with such time-varying sampling thresholds, we have
$r_{k} = \operatorname{sgn}\left(y_{k}-\tau_{k}\right)$.
The information gathered through the one-bit sampling with time-varying thresholds presented here may be formulated in terms of an overdetermined linear system of inequalities. 
We have $r_{k}=+1$ when $y_{k}>\tau_{k}$ and $r_{k}=-1$ when $y_{k}<\tau_{k}$. Therefore, one can formulate the geometric location of the signal as
$r_{k}\left(y_{k}-\tau_{k}\right) \geq 0$.
Collecting all the elements in the vectors as $\mathbf{y}=[y_{k}] \in \mathbb{R}^{n}$ and $\mathbf{r}=[r_{k}] \in \left\{-1,1\right\}^{n}$, we have
$\mbr\odot\left(\mathbf{y}-\boldsymbol{\uptau}\right) \succeq \mathbf{0}$,
or equivalently
\begin{equation}
\label{eq:6}
\begin{aligned}
\bOmega_{\mby}  \mathbf{y} &\succeq \mbr \odot \boldsymbol{\uptau},
\end{aligned}
\end{equation}
where $\bOmega_{\mby}  \triangleq \operatorname{diag}\left(\mbr\right)$. Denote the time-varying sampling threshold in $\ell$-th signal sequence by $\boldsymbol{\uptau}^{(\ell)}$, where  $\ell\in [m]$.
It follows from \eqref{eq:6} that
\begin{equation}
\label{eq:7}
\begin{aligned}
\bOmega^{(\ell)}_{\mby} \mathbf{y} &\succeq \mathbf{r}^{(\ell)} \odot \boldsymbol{\uptau}^{(\ell)}, \quad \ell \in [m],
\end{aligned}
\end{equation}
where $\bOmega^{(\ell)}_{\mby} =\operatorname{diag}\left(\mbr^{(\ell)}\right)$. 
Denote the concatenation of all $m$ sign matrices as
\begin{equation}
\label{eq:9}
\Tilde{\bOmega}_{\mby} =\left[\begin{array}{c|c|c}
\bOmega^{(1)}_{\mby}  &\cdots &\bOmega^{(m)}_{\mby} 
\end{array}\right]^{\top}, \quad \Tilde{\bOmega}_{\mby} \in \left\{-1,0,1\right\}^{m n\times n}.
\end{equation}
Rewrite the $m$ linear 
inequalities in \eqref{eq:7} as
\begin{equation}
\label{eq:8}
\Tilde{\bOmega}_{\mby}  \mathbf{y} \succeq \operatorname{vec}\left(\mbR_{\mby} \right)\odot \operatorname{vec}\left(\bGamma\right),
\end{equation}
where $\mathbf{R}$ and $\bGamma$ are matrices, whose columns are the sequences $\left\{\mathbf{r}^{(\ell)}\right\}_{\ell=1}^{m}$ and $\left\{\boldsymbol{\uptau}^{(\ell)}\right\}_{\ell=1}^{m}$, respectively.

Assuming a large number of samples --- a common situation in one-bit sampling scenarios --- hereafter we treat (\ref{eq:8}) as an overdetermined linear system of inequalities associated with the one-bit sensing scheme.
The inequality (\ref{eq:8}) can be recast as a polyhedron,
\begin{equation}
\label{eq:8n}
\begin{aligned}
\mathcal{P}_{\mby}=\left\{\mby^{\prime} \in \mathbb{R}^n \mid \tilde{\boldsymbol{\Omega}}_{\mby} \mby^{\prime} \succeq \operatorname{vec}\left(\mbR_{\mby}\right) \odot \operatorname{vec}(\boldsymbol{\Gamma})\right\} \subset \mathbb{R}^n,
\end{aligned}
\end{equation}
which we refer to as the \emph{one-bit polyhedron}. Generally, it can be assumed that the signal $\mbx\in\mathbb{R}^{d}$ is observed linearly through the sampling matrix $\mbA\in\mathbb{R}^{n\times d}$, creating the measurements as $\mby=\mbA\mbx$. Based on \eqref{eq:8}, the one-bit polyhedron for this type of problem is given by
\begin{equation}
\label{eq:80n}
\begin{aligned}
\mathcal{P}_{\mathbf{x}} = \left\{\mathbf{x}^{\prime}\in\mathbb{R}^{d} \mid \mbP_{\mby}  \mathbf{x}^{\prime} \succeq \operatorname{vec}\left(\mbR_{\mby} \right)\odot \operatorname{vec}\left(\bGamma\right)\right\}\subset \mathbb{R}^d,
\end{aligned}
\end{equation}
where $\mbP_{\mby} =\Tilde{\bOmega}_{\mby} \mbA$ or equivalently
\begin{equation}
\label{eq:90}
\mbP_{\mby} =\left[\begin{array}{c|c|c}
\mbA^{\top}\bOmega^{(1)}_{\mby}&\cdots &\mbA^{\top}\bOmega^{(m)}_{\mby} 
\end{array}\right]^{\top}, \quad \mbP_{\mby} \in\mathbb{R}^{m n\times d}.
\end{equation}
By taking advantage of one-bit sampling in the asymptotic sample size scenario (sample abundance), the space constrained by the one-bit polyhedron $\mathcal{P}_{\mathbf{x}}$ \emph{shrinks} to become contained within the feasible set~$\Omega_c$. Note that this shrinking space always contains the global minima, with a volume that is diminished with an increased sample size. To find a solution inside the one-bit polyhedron, the ORKA employs the variants of RKA introduced in the following subsections.

\subsection{Randomized Kaczmarz Algorithm}
\label{sec_RKA}
The RKA is a \emph{sub-conjugate gradient method} to solve a linear feasibility problem, i.e, $\mbC\mathbf{x}\succeq\mathbf{b}$ where $\mbC$ is a ${m\times n}$ matrix with $m>n$ \cite{leventhal2010randomized,strohmer2009randomized}. Conjugate-gradient methods immediately turn the mentioned inequality to an equality in the following form $\left(\mbb-\mbC\mathbf{x}\right)^{+}=0$,
and then, approach the solution by the same process as used for systems of equations. The projection coefficient $\beta_{i}$ of the RKA is \cite{leventhal2010randomized}
\begin{equation}
\label{eq:22}
\beta_{i}= \begin{cases}
\left(b_{j}-\langle\mathbf{c}_{j},\mathbf{x}_{i}\rangle\right)^{+} & \left(j \in \mathcal{I}_{\geq}\right), \\ b_{j}-\langle\mathbf{c}_{j},\mathbf{x}_{i}\rangle & \left(j \in \mathcal{I}_{=}\right),
\end{cases}
\end{equation}
where the disjoint index sets $\mathcal{I}_{\geq}$ and $\mathcal{I}_{=}$ partition $[m]$
and $\{\mathbf{c}_{j}\}$ are the rows of $\mathbf{C}$.
Also, the unknown column vector $\mathbf{x}$ is iteratively updated as
\begin{equation}
\label{eq:23}
\mathbf{x}_{i+1}=\mathbf{x}_{i}+\frac{\beta_{i}}{\left\|\mbc_{j}\right\|^{2}_{2}} \mbc^{\star}_{j},
\end{equation}
where, at each iteration $i$, the index $j$ is drawn from the set $[m]$ independently at random following the distribution
$\operatorname{Pr}\{j=k\}=\frac{\left\|\mbc_{k}\right\|^{2}_{2}}{\|\mbC\|_{\mathrm{F}}^{2}}$.
Assuming that the linear system is consistent with nonempty feasible set $\mathcal{P}_{\mbx}$ created by the intersection
of hyperplanes around the desired point $\mbx$,
RKA converges linearly in expectation to the solution $\widehat{\mbx}\in\mathcal{P}_{\mbx}$\cite{strohmer2009randomized,leventhal2010randomized}:
\begin{equation}
\label{eq:15}
\mathbb{E}\left\{\hbar\left(\mathbf{x}_{i},\widehat{\mbx}\right)\right\} \leq \left(1-q_{_{\text{RKA}}}\right)^{i}~ \hbar\left(\mathbf{x}_{0},\widehat{\mbx}\right),
\end{equation}
where $\hbar\left(\mathbf{x}_{i},\widehat{\mbx}\right)=\left\|\mathbf{x}_{i}-\widehat{\mbx}\right\|_{2}^{2}$, is the euclidean distance between two points in the space, $i$ is the number of required iterations for RKA, and $q_{_{\text{RKA}}} \in \left(0,1\right)$ is given by $q_{_{\text{RKA}}}=\frac{1}{\kappa^{2}\left(\mbC\right)}$, with $\kappa\left(\mbC\right)=\|\mbC\|_{\mathrm{F}}\|\mbC^{\dagger}\|_{2}$ denoting the scaled condition number of a matrix $\mbC$. To reach an error bounded by $\varepsilon$, \eqref{eq:15} shows that this
is achieved in $I$ iterations if we impose $\left(1-q_{\mathrm{RKA}}\right)^I \simeq  e^{-I / \kappa^2} \leqslant \varepsilon$, that is $I=\kappa^2 \log (1 / \varepsilon)$. Since, $\kappa\geq \sqrt{n}$, RKA thus needs at least $\mathcal{O}\left(n\log (1 / \varepsilon)\right)$ multiplication to reach an error $\varepsilon$.

\subsection{Sampling Kaczmarz Motzkin Method}
\label{sec_SKM}
The SKM combines the ideas of both the RKA and the Motzkin method. The generalized convergence analysis of the SKM with sketch matrix which has been formulated based on the convergence analysis of RKA, and sampling Motzkin method for solving linear feasibility problem has been comprehensively explored in \cite{de2017sampling}.
The central contribution of SKM lies in its innovative way of projection plane selection. The hyperplane selection is done as follows: At iteration $i$, the SKM algorithm selects a collection of $\gamma$ (denoted by the set $\mathcal{T}_{i}$) rows, uniformly at random out of $m$ rows of the constraint matrix $\mbC$. Then, out of these $\gamma$ rows, the row $j_{\star}$ with the maximum positive residual is selected; i.e.
\begin{equation}
j_{\star}=\operatorname{argmax}~\left\{ \left(b_{j}-\langle\mbc_{j},\mathbf{x}_{i}\rangle\right)^{+}\right\},~j\in\mathcal{T}_{i}.    
\end{equation}
Finally, the solution is updated as \cite{de2017sampling} $\mathbf{x}_{i+1}=\mathbf{x}_{i}+\lambda_{i}\frac{\beta_{i}}{\left\|\mbc_{j_\star}\right\|^{2}_{2}} \mbc^{\star}_{j_{\star}}$,
where $\lambda_{i}$ is a relaxation parameter which for consistent systems must satisfy $0\leq \lim_{i\to\infty} \inf \lambda_{i}\leq \lim_{i\to\infty} \sup \lambda_{i}<2$ \cite{strohmer2009randomized},
to ensure convergence. The SKM algorithm enjoys the linear convergence as shown in \cite{de2017sampling} same as RKA. 

\subsection{Preconditioned SKM}
\label{sec_preSKM}
Assume $\mathcal{P}_{\mbx}$ as the space created by the intersection of hyperplanes in a linear feasibility problem. According to the convergence rate of RKA, reducing the value of the scaled condition number $\kappa\left(\mbC\right)$ or equivalently increasing the value of $q_{_{\text{RKA}}}$ in \eqref{eq:15} leads to an accelerated convergence of the RKA to $\mathcal{P}_{\mbx}$.
From another perspective, this property (lower value of $\kappa\left(\mbC\right)$) provides the RKA or its variant, SKM, enjoying a lower number of iterations required to reach an specific error, usually considered to be the algorithm's termination criterion. Consequently, assuming $I$ as a number of iterations, the computational cost of RKA which behaves as $\mathcal{O}\left(I n\right)$, is diminished as well.
The following theorem states how one can achieve the infimum of the scaled condition number $\kappa\left(\mbC\right)$:
\begin{theorem}
\label{scaled_number}
The infimum scaled condition number of a matrix $\mbC\in\mathbb{R}^{m\times n}$ is given by
\begin{equation}
\label{eq:Arian}
\inf_{\mbC} \kappa\left(\mbC\right) = \sqrt{n},
\end{equation}
which is achieved if and only if $\mbC$ is of the form $\mbC=\alpha \mbU$, where $\mbU$ is an orthonormal column matrix and $\alpha\in\mathbb{R}$ is a scalar.
\end{theorem}
\begin{proof}
The condition number of the matrix $\mbC$ is defined as $\varrho\left(\mbC\right)=\frac{\sigma_{\mathrm{max}}}{\sigma_{\mathrm{min}}}$, where $\sigma_{\mathrm{max}}$ and $\sigma_{\mathrm{min}}$ are the maximum and minimum singular values of the matrix $\mbC$, respectively. The scaled condition number can be written as $\kappa\left(\mbC\right)=\frac{\|\mbC\|_{\mathrm{F}}}{\sigma_{\mathrm{min}}}$. Therefore, the scaled condition number $\kappa\left(\mbC\right)$ has the following relation with the condition number $\varrho\left(\mbC\right)$,
\begin{equation}
\label{gorang}
\kappa\left(\mbC\right)=\frac{\|\mbC\|_{\mathrm{F}}}{\sigma_{\mathrm{max}}}\varrho\left(\mbC\right). 
\end{equation}
Based on the relation between the norm-$2$ and the Frobenius norm of a matrix $\mbC$, $\|\mbC\|_{\mathrm{F}}\leq \sqrt{n}  \|\mbC\|_{\mathrm{2}}$ or equivalently, $\frac{\|\mbC\|_{\mathrm{F}}}{\sigma_{\mathrm{min}}} \leq \sqrt{n} \frac{\|\mbC\|_{\mathrm{2}}}{\sigma_{\mathrm{min}}}$,
the condition number $\varrho\left(\mbC\right)$ can be considered to be an upper bound for the scaled condition number $\kappa\left(\mbC\right)$ as follows,
$\kappa\left(\mbC\right)\leq \sqrt{n} \varrho\left(\mbC\right)$.
Additionally,
the lowest possible value of $\varrho$ is $1$ which is achieved for scaled unitary matrices $\mbU$ as if we let $\mbS=\alpha \mbU$, and $\mbO=\mbS^{\top}\mbS=\alpha^{2}\mbI_{n}$, then, $\sigma_{\mathrm{min}}=\sigma_{\mathrm{max}}=\alpha$, and $\varrho=1$. Vice versa, if $\varrho=1$, it means $\sigma_{\mathrm{min}}=\sigma_{\mathrm{max}}$ which leads to a diagonal matrix $\mbO=\alpha^{2}\mbI_{n}$. 
\end{proof}
Following Theorem~\ref{scaled_number}, it would be enough to make the matrix $\mbC$ unitary by a process which is referred to as the \emph{preconditioning method}. In preconditioning, the linear feasibility is rewritten as $\mbC \mbM \mathbf{z}\succeq\mathbf{b}$,
where $\mbM$ is the preconditioner and $\mathbf{x}=\mbM \mathbf{z}$. The straightforward way to approach this task is to use QR decomposition where the constraint matrix is decomposed as $\mbC=\mbQ_{c}\mbR_{c}$, where $\mbQ_{c}\in\mathbb{R}^{m\times n}$ is a unitary matrix, and $\mbR_{c}\in\mathbb{R}^{n\times n}$ is an upper triangular matrix, leading to $\mbQ_{c}=\mbC\mbR^{-1}_{c}$. For a matrix $\mbC\in\mathbb{R}^{m\times n}$, since we have assumed $m>n$, we can obtain a unitary matrix $\mbQ_{c}\in\mathbb{R}^{m\times n}$ and an upper triangular matrix $\mbR_{c}\in\mathbb{R}^{n\times n}$ such that the QR-decomposition holds; i.e. $\mbC=\mbQ_{c}\mbR_{c}$.
Thus, based on Theorem~\ref{scaled_number}, a good choice for the preconditioner is $\mbM=\mbR^{-1}_{c}$. To find a
point $\bar{\mathbf{z}}$ in a nonempty feasible set,
the SKM method described in Section~\ref{sec_SKM} is employed.
Finally, one may approach the solution of the original linear feasibility $\mbC\mbx\succeq\mathbf{b}$ by computing $\bar{\mathbf{x}}=\mbR^{-1}_{c}\bar{\mathbf{z}}$.
We refer to this method \emph{Pr}econditioned \emph{SKM} (PrSKM) which is summarized in Algorithm~\ref{algorithm_1}.
\begin{algorithm}[t]
\caption{PrSKM Algorithm}
\label{algorithm_1}
\begin{algorithmic}[1]
\Statex \emph{Input:} A matrix $\mbC$ and the measurement vector $\mbb$.
\Statex \emph{Output:} A solution $\bar{\mbx}$ in a nonempty feasible set of $\mbC\mbx\succeq\mbb$.
\State $\left[\mbQ_{c},\mbR_{c}\right]=\operatorname{QR}\left(\mbC\right)$ $\triangleright$ $\operatorname{QR}(\cdot)$ computes the QR-decomposition of a matrix $\mbC$.
\State $\mbQ_{c}\mbz\succeq\mbb$ $\triangleright$ The new problem that we should solve respect to $\mbz$.
\State Choose a sample set of $\gamma$ constraints (denoted as $\mathcal{T}_{i}$) uniformly at random from the rows of $\mbQ_{c}$.
\State 
$j_{\star}\gets\operatorname{argmax}~\left\{ \left(b_{j}-\langle\mbq_{j},\mathbf{z}_{i}\rangle\right)^{+}\right\},~j\in\mathcal{T}_{i}$ $\triangleright$ $\mbq_{j}$ is the $j$-th row of $\mbQ_{c}$.
\State $\mathbf{z}_{i+1}\gets\mathbf{z}_{i}+\lambda_{i}\frac{\left(b_{j_{\star}}-\langle\mbq_{j_{\star}},\mathbf{z}_{i}\rangle\right)^{+}}{\left\|\mbq_{j_\star}\right\|^{2}_{2}} \mbq^{\star}_{j_{\star}}$.
\State Repeat steps (4)-(6) until convergence and obtain $\bar{\mbz}$.
\State $\mbR_{c}\bar{\mbx}=\bar{\mathbf{z}}$ $\triangleright$ Obtain $\bar{\mbx}$ via the Gaussian elimination algorithm.
\State \Return $\bar{\mbx}$
\end{algorithmic}
\end{algorithm}
As shown in Theorem~\ref{scaled_number}, the scaled condition number of the matrix $\mbQ_{c}$ is $\kappa\left(\mbQ_{c}\right)=\sqrt{n}$. Therefore, step~$5$ of Algorithm~\ref{algorithm_1} converges linearly in expectation to a nonempty feasible set of $\mbQ_{c}\mbz\succeq\mbb$ as follows,
\begin{equation}
\label{convergence}
\mathbb{E}\left\{\hbar\left(\mathbf{x}_{i},\widehat{\mbx}\right)\right\} \leq \left(1-\frac{1}{n}\right)^{i}~ \hbar\left(\mathbf{x}_{0},\widehat{\mbx}\right).
\end{equation}
Note that to run the Algorithm~\ref{algorithm_1} for solving the one-bit polyhedron \eqref{eq:80n}, it is enough to set $\mbC=\mbP_{\mby}$ and $\mbb=\operatorname{vec}\left(\mbR_{\mby}\right)\odot \operatorname{vec}\left(\bGamma\right)$. Later, in Theorem~\ref{theorem_2}, we will show that it is only required to set $\mbC=\mbA$ in Algorithm~\ref{algorithm_1} which is more computationally and storagelly efficient compared to that of setting $\mbC=\mbP_{\mby}$.

\subsection{Storage-Friendly PrSKM}
The PrSKM algorithm uses QR decomposition as a preconditioning process, which can be computationally challenging for high-dimensional matrices. To address this issue, the literature has proposed a technique called \emph{sketch-and-precondition}, as outlined in \cite{rokhlin2008fast}. This approach is designed to mitigate the curse of dimensionality associated with QR-decomposition-based preconditioning. 

To explain how this method works, consider for simplicity an overdetermined linear system, i.e., $\mbU\mbx=\mbb$ with $\mbU\in\mathbb{R}^{m\times n}$, $m\gg n$. 
Define a Gaussian matrix $\mbN\in \mathbb{R}^{m\times s}$. Sketch-and-precondition computes a sketch $\mbN^{\top}\mbU$ using a random test matrix $\mbN$ with sketch size $s=\mathcal{O}\left(n/\epsilon^{2}\right)\ll m$ where $\epsilon\in (0,1)$. Next, it performs a QR-decomposition $\mbN^{\top}\mbU=\mbQ_p\mbR_p$ of the resulting $s\times n$ matrix and then uses $\mbR_{p}^{-1}$ as a preconditioner for $\mbU$. This approach significantly reduces the computational cost of QR-decomposition from $\mathcal{O}\left(mn^2\right)$ to $\mathcal{O}\left(sn^2\right)$, regardless of the number of rows. This reduction in computational cost is particularly advantageous for highly-overdetermined linear systems, demonstrating the significance of this approach. For more information, we encourage a reader to see
\cite{martinsson2020randomized,woodruff2014sketching}. This paper introduces the first instance of utilizing the sketch-and-precondition technique on a linear inequality system within the literature. 
To explicitly state the algorithmic implementation of storage-friendly PrSKM, the first two steps of Algorithm~\ref{algorithm_1} must be replaced by the following steps:
\begin{enumerate}
    \item Generate a Gaussian test matrix $\mbN\in\mathbb{R}^{m\times s}$ where $s$ follows $\mathcal{O}\left(n/\epsilon^{2}\right)$ with $\epsilon=\frac{1}{2}$, sketch the matrix $\mbC\in\mathbb{R}^{m\times n}$ from $\mbC\mbx\succeq\mbb$ as $\mbS=\mbN^{\top}\mbC$.
    \item $\left[\mbQ_s,\mbR_s\right]=\operatorname{QR}\left(\mbS\right)$ $\triangleright$ $\operatorname{QR}(\cdot)$ computes the QR-decomposition of the sketched matrix $\mbS$, where $\mbQ_s\in \mathbb{R}^{s\times n}$ and $\mbR_{s}\in\mathbb{R}^{n\times n}$.
    \item $\mbM=\mbR^{-1}_{s}$ $\triangleright$ $\mbM$ is the preconditioner.
    \item $\mbC \mbM \mathbf{z}\succeq\mathbf{b}$ $\triangleright$ New problem that we should solve respect to $\mbz$.
\end{enumerate}

We provide the convergence of storage-friendly PrSKM to the nonempty feasible set of $\mbC \mbM \mathbf{z}\succeq\mathbf{b}$ in the following theorem:
\begin{theorem}
\label{stprskm}
Suppose $\mbN\in\mathbb{R}^{m\times s}$ is a Gaussian matrix with the sketch size $s=\mathcal{O}\left(n/\epsilon^{2}\right)$ with $\epsilon\in (0,1)$ (e.g., $\epsilon=\frac{1}{2}$). Consider the sketch matrix as $\mbS=\mbN^{\top}\mbC$ where $\mbC\in\mathbb{R}^{m\times n}$, $m\gg n$. Assume $\left[\mbQ_s,\mbR_s\right]=\operatorname{QR}\left(\mbS\right)$ where $\operatorname{QR}(\cdot)$ denotes the QR-decomposition operator. Then, the storage-friendly PrSKM algorithm converges linearly in expectation to the nonempty feasible set of $\mbC\mbR_{s}^{-1} \mathbf{z}\succeq\mathbf{b}$ as follows:
\begin{equation}
\label{stprskm1}
\mathbb{E}\left\{\hbar\left(\mathbf{x}_{i},\widehat{\mbx}\right)\right\}\leq\left(1-\frac{1}{3 n}\right)^{i}~ \hbar\left(\mathbf{x}_{0},\widehat{\mbx}\right).
\end{equation}
\end{theorem}
The proof of Theorem~\ref{stprskm} is presented in Appendix~\ref{ak}.

\subsection{Block SKM}
\label{ORKA}
In two previous sections, we have proposed PrSKM and storage-friendly PrSKM algorithms in order to solve the one-bit polyhedron \eqref{eq:80n} in an asymptotic sample size scenario. It is worth noting that both methods are
\emph{row-based} approaches, where at each iteration, the row index of the matrix $\mbP_{\mby}$ is chosen independently at random. However, the matrix $\mbP_{\mby}$ in (\ref{eq:80n}) has a block structure as formulated in \eqref{eq:90}.
This fact motivates us to investigate the block-based RKA methods to find the desired signal in the one-bit polyhedron $\mathcal{P}_{\mathbf{x}}$ for further efficiency enhancement. Our proposed algorithm for block systems, \emph{Block SKM}, is motivated by (i) random selection of one block at each iteration, (ii) choosing a subset of rows using the idea of Motzkin sampling, and (iii) updating the solution using the randomized block Kaczmarz method, which takes advantage of the efficient matrix-vector multiplication, thus giving the method a significant reduction in computational cost. Algorithm~\ref{algorithm_2} shows the implementation of our proposed Block SKM method to solve the linear feasibility $\mbB\mbx\succeq\mbb$ with $\mbB=\left[\begin{array}{c|c|c}\mbB^{\top}_{1} &\cdots&\mbB^{\top}_{m}\end{array}\right]^{\top}$ and $\mathbf{b}=\left[\begin{array}{c|c|c}\mathbf{b}^{\top}_{1} &\cdots&\mathbf{b}^{\top}_{m}\end{array}\right]^{\top}$ where $\mbB_{j}\in\mathbb{R}^{n\times d}$ and $\mbb_{j}\in\mathbb{R}^{n}$ for all $j\in[m]$. Note that in step~$4$ of Algorithm~\ref{algorithm_2}, the reason behind choosing $k^{\prime}<d$ is due to the computation of $\left(\mbB_{j}^{\prime}\mbB_{j}^{\prime\top}\right)^{-1}$ in the next step (step $5$). For $k^{\prime}>d$, the matrix $\mbB_{j}^{\prime}\mbB_{j}^{\prime\top}$ is rank-deficient and its inverse is not available. The Block SKM algorithm can be considered to be a special case of the more general \emph{sketch-and-project} method with a sparse block sketch matrix as defined in \cite{derezinski2022sharp}. The convergence of Block SKM algorithm
with the sparse Gaussian sketch in the case of $k^{\prime}=1$ is presented in the following lemma:
\begin{algorithm}[t]
\caption{Block SKM Algorithm}
\label{algorithm_2}
\begin{algorithmic}[1]
\Statex \emph{Input:} A block matrix $\mbB=\left[\begin{array}{c|c|c}\mbB^{\top}_{1} &\cdots&\mbB^{\top}_{m}\end{array}\right]^{\top}$, and the measurement vector $\mathbf{b}=\left[\begin{array}{c|c|c}\mathbf{b}^{\top}_{1} &\cdots&\mathbf{b}^{\top}_{m}\end{array}\right]^{\top}$.
\Statex \emph{Output:} A solution $\bar{\mbx}$ in a nonempty feasible set of $\mbB\mbx\succeq\mbb$.
\State Choose a block $\mbB_{j}$ uniformly at random with the probability $\operatorname{Pr}\{j=k\}=\frac{\left\|\mbB_{k}\right\|^{2}_{\mathrm{F}}}{\|\mbB\|_{\mathrm{F}}^{2}}$.
\State $\mathbf{e}\gets\mbB_{j}\mathbf{x}-\mathbf{b}_{j}$.
\State $\mbe^{\prime}\gets\operatorname{sort}\left(\mbe\right)$ $\triangleright$ $\operatorname{sort}(\cdot)$ is the operator that sorts the elements of the vector $\mbe$ from $e_{\mathrm{max}}$ (the maximum element of $\mathbf{e}$) to $e_{\mathrm{min}}$ (the minimum element of $\mathbf{e}$). This step is inspired by the idea of the Motzkin sampling, presented in Section~\ref{sec_SKM}, to have an accelerated convergence.
\State $[\mbB_{j}^{\prime},\mbb_{j}^{\prime}]=\operatorname{select}_{k^{\prime}}(\mbB_{j},\mbb_{j},\mbe^{\prime})$ $\triangleright$ $\operatorname{select}_{k^{\prime}}(\cdot)$ is the operator which selects the first $k^{\prime}<d$ element of $\mathbf{e}^{\prime}$ and construct the sub-problem $\mbB_{j}^{\prime}\mathbf{x}	\succeq\mathbf{b}_{j}^{\prime}$, where $\mbB_{j}^{\prime}\in\mathbb{R}^{k^{\prime}\times d}$ and $\mathbf{b}_{j}^{\prime}\in\mathbb{R}^{k^{\prime}}$.
\State $\mbB_{j}^{\prime\dagger}\gets\mbB_{j}^{\prime\top}\left(\mbB_{j}^{\prime}\mbB_{j}^{\prime\top}\right)^{-1}$ $\triangleright$ $\mbB_{j}^{\prime\dagger}$ is the Moore-Penrose pseudoinverse of $\mbB_{j}^{\prime}$.
\State $\mathbf{x}_{i+1}\gets\mathbf{x}_{i}+\lambda_{i}\mbB_{j}^{\prime\dagger}\left(\mathbf{b}_{j}^{\prime}-\mbB_{j}^{\prime}\mathbf{x}\right)^{+}$.
\State Repeat steps (1)-(6) until convergence and obtain $\bar{\mbx}$.
\State \Return $\bar{\mbx}$
\end{algorithmic}
\end{algorithm}
\begin{lemma}
\label{kvm}
The Block SKM algorithm with the sparse Gaussian sketch in the case of $k^{\prime}=1$, converges linearly in expectation to the nonempty feasible set of $\mbB\mathbf{x}\succeq \mathbf{b}$, $\mbB\in\mathbb{R}^{m n\times d}$, as follows,
\begin{equation}
\begin{aligned}
\mathbb{E}\left\{\hbar\left(\mathbf{x}_{i},\widehat{\mbx}\right)\right\}\leq \left(1-\frac{c \sigma_{\min }^2(\widehat{\mbB})\log n}{\|\widehat{\mbB}\|_{\mathrm{F}}^2}\right)^{i}\hbar\left(\mathbf{x}_{0},\widehat{\mbx}\right),
\end{aligned}
\end{equation}
where $c$ is a positive constant value, $\widehat{\mbB}$ is the $n\times d$ submatrix of $\mbB$ (one of the candidates of $\mbB_{j}$ in Algorithm~\ref{algorithm_2}).
\end{lemma}
The proof of Lemma~\ref{kvm} is provided in Appendix~\ref{A}. In the rest of the paper, we will derive all theoretical guarantees for the RKA. Similar guarantees can also be derived for the variants of the RKA; i.e. SKM, Block SKM and etc.
Throughout the paper, we represent each iterate of the update process for RKA, PrSKM, and Block SKM as $\mathrm{KA}_{r}(\mbx)$, $\mathrm{KA}_{p}(\mbx)$, and $\mathrm{KA}_{b}(\mbx)$, respectively. To examine the performance of the Block SKM, we will compare it with the PrSKM, SKM and RKA.

\subsection{Comparing RKA, SKM, PrSKM and Block SKM}
\label{NUM_COM}
In this section, we numerically compare the RKA, SKM, PrSKM, and Block SKM in linear system of
inequalities. Accordingly, we utilize ORKA to make a linear equation $\mbB\mathbf{x}=\mathbf{y}$ a linear inequality system, where the number of time-varying sampling threshold sequences is $m=40$, $\mbB\in\mathbb{R}^{100\times 10}$, $\mathbf{x}\in\mathbb{R}^{10}$, and $\mathbf{y}\in\mathbb{R}^{100}$. Each row of $\mbB$ is generated as $\mathbf{b}_{j}\sim\mathcal{N}\left(\mathbf{0},\mbI_{10}\right)$. Also, the desired signal $\mathbf{x}$ is generated as $\mathbf{x}\sim\mathcal{N}\left(\mathbf{0},\mbI_{10}\right)$. All
time-varying sampling threshold sequences are generated according to 
$\left\{\boldsymbol{\uptau}^{(\ell)}\sim\mathcal{N}\left(\mathbf{0},\mbI_{100}\right)\right\}_{\ell=1}^{m}$. The performance of the RKA, SKM, PrSKM, and Block SKM is illustrated in Fig.~\ref{figure_1}. It can be seen that the Block SKM has a better accuracy in recovering the desired signal $\mathbf{x}$ in the one-bit polyhedron (\ref{eq:80n}) compared to the other three approaches. The NMSE results in Fig.~\ref{figure_1} are averaged over $1000$ experiments.

\begin{figure}[t]
	\centering
		{\includegraphics[width=0.47\columnwidth]{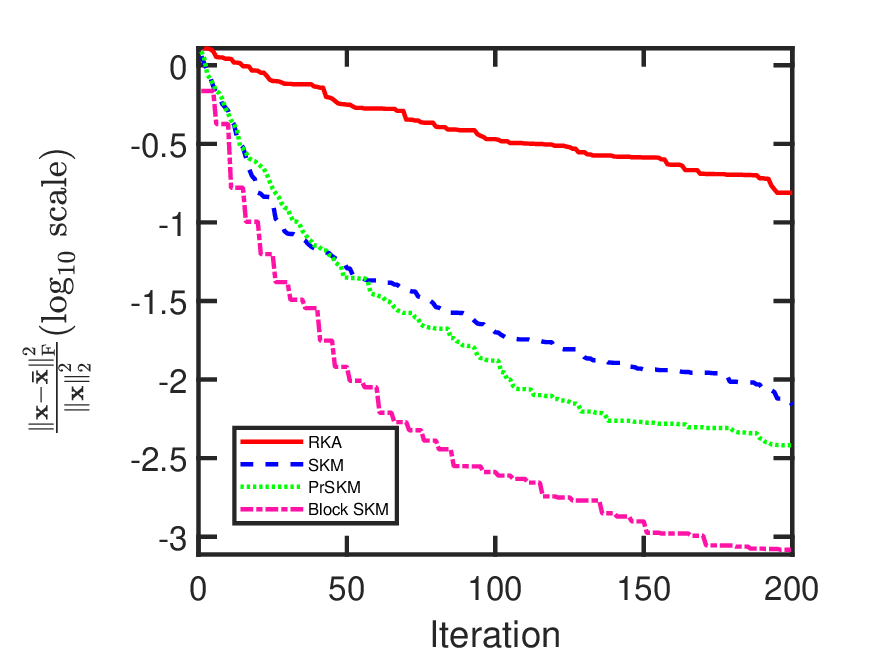}}
	\caption{Comparing the 
 recovery performance of the two proposed Kaczmarz algorithms, namely the PrSKM and the Block SKM, with that of SKM and RKA for
 a linear inequality system. 
 \vspace{-15pt}
 }
	\label{figure_1}
\end{figure}
\section{ORKA With Noisy Measurements}
\label{NOISE}
Herein, we formulate the noisy version of one-bit sampling with time-varying thresholds. Denote $\mbz=[z_{j}]\in\mathbb{R}^{n}$ as a noise vector which has been added to the linear system of equations $\mby=\mbA\mbx$ leading to $\mby_z=\mbA\mbx+\mbz$. Then, the corresponding noisy one-bit samples are generated as
\begin{equation}
\label{St_8}
\begin{aligned}
r_{j}^{(\ell)} &= \begin{cases} +1 & \langle\mba_j,\mbx\rangle+z_j>\tau_{j}^{(\ell)},\\ -1 & \langle\mba_j,\mbx\rangle+z_j<\tau_{j}^{(\ell)},
\end{cases}
\quad j\in[n],~\ell\in[m],
\end{aligned}
\end{equation}
where $\mba_{j}$ denotes the $j$-th row of a sampling matrix $\mbA$.
Consequently, the one-bit polyhedron associated with \eqref{St_8} is rewritten as
\begin{equation}
\label{St_10}
\mbP_{\mby_z} \mbx+\boldsymbol{\upnu}\succeq \operatorname{vec}\left(\mbR_{\mby_z}\right)\odot \operatorname{vec}\left(\bGamma\right),
\end{equation}
where $\mbP_{\mby_z}$ is defined similar to \eqref{eq:90} and $\boldsymbol{\upnu}=\Tilde{\bOmega}_{\mby_z}\mbz$ is the noise of our system with $\Tilde{\bOmega}_{\mby_z}$ defined similar to \eqref{eq:9}. For instance, 
assuming a zero-mean Gaussian noise vector 
$\mbz\sim\mathcal{N}\left(\mathbf{0},\bSigma_z\right)$ with 
the covariance matrix $\bSigma_z$, the distribution of $\boldsymbol{\upnu}$ will be $\mathcal{N}\left(\mathbf{0}, \Tilde{\bOmega}_{\mby_z}\bSigma_z\Tilde{\bOmega}_{\mby_z}^{\mathrm{H}}\right)$.

The robustness of the RKA
against noise has been demonstrated in \cite{needell2010randomized}. Furthermore, the authors of \cite{huang2022linear} specifically explored the performance of the RKA 
in the presence of \emph{Gaussian} and \emph{Poisson} noise, highlighting its robustness even when dealing with Poisson noisy measurements. In our discussion, in Section~\ref{sec_noisy_a} we will explore how the inconsistency of a linear system in a noisy scenario manifests itself in the recovery error of the RKA. Next, in Section~\ref{QRKA} we will propose a novel algorithm to have a robust recovery performance in the presence of impulsive noise.

\subsection{Robustness of ORKA Against Noise}
\label{sec_noisy_a}
Given a linear system  of equations $\mbU\mbx=\mbb$ that is highly over-determined and subject to a noise vector $\mbn=[n_{j}]$ resulting in a corrupted system of equations $\mbU\mbx\approx\mbb+\mbn$.
The convergence rate of the noisy RKA was comprehensively discussed in \cite[Theorem~2.1]{needell2010randomized} for the case of $\mbU\mbx\approx\mbb+\mbn$.
The primary contrast between the convergence rates of RKA
and noisy RKA, as demonstrated in \cite[Theorem~2.1]{needell2010randomized}, lies in the second term of convergence rate $\kappa^2\max_{j}\frac{n^{2}_j}{\left\|\mbu_j\right\|^2_2}$. This term indicates the degree to which the error in the corrupted system $\mbU\mbx\approx\mbb+\mbn$ deviates from the main solution. 

Drawing inspiration from the convergence rate of the noisy RKA, 
we can similarly derive the convergence rate of noisy RKA in the case of noisy linear system of inequalities
$\mbC\mbx+\mbn\succeq\mbb$ using the following proposition:
\begin{proposition}
\label{prop_noisyRKA}
Let $\mbC\in\mathbb{R}^{m\times n}$ have full column rank and assume $\widehat{\mbx}$ is the solution of unperturbed feasibility problem $\mbC\mbx\succeq\mbb$. Let $\bar{\mbx}_{i}$ be the $i$-th iterate of the noisy RKA run with $\mbC\mbx+\mbn\succeq\mbb$, and let $n_j$ denotes the $j$-th element of $\mbn$. Then we have
\begin{equation}
\label{St_15000}
\mathbb{E}\left\{\left\|\bar{\mbx}_{i}-\widehat{\mbx}\right\|_2\right\} \leq \left(1-\frac{1}{\kappa^{2}\left(\mbC\right)}\right)^{\frac{i}{2}}~ \left\|\mathbf{x}_{0}-\widehat{\mbx}\right\|_2+\kappa\max_{j}\gamma_j,
\end{equation}
where $\gamma_j=\frac{|n_j|}{\left\|\mbc_j\right\|_2}$, with $\mbc_j$ denotes the $j$-th row of the matrix $\mbC$.
\end{proposition}
The proof of Proposition~\ref{prop_noisyRKA} is presented in Appendix~\ref{A5}.
Based on the noisy one-bit polyhedron \eqref{St_10}, the parameters $\mbC=\mbP_{\mby_z}$, $\mbn=\boldsymbol{\upnu}$, and $\mbb=\operatorname{vec}\left(\mbR_{\mby_z}\right)\odot \operatorname{vec}\left(\bGamma\right)$ can be replaced in the convergence rate presented in \eqref{St_15000} to get the similar result in the one-bit noisy scenario. Note that as can be observed in \eqref{St_15000}, a small perturbation in the linear feasibility problem $\mbC\mbx\succeq\mbb$ may slightly deviate the solution of the noisy RKA from the main solution.

\subsection{Upper Quantile-Based ORKA}
\label{QRKA}
In Section~\ref{sec_noisy_a}, the robustness of the noisy RKA in the presence of a small perturbation has been discussed. However, as can be observed in the convergence rate of the noisy RKA \eqref{St_15000}, a large perturbation in the linear feasibility $\mbC\mbx\succeq\mbb$ can lead to a significant error in the input signal recovery. In such a scenario, an accurate input signal recovery is quite a challenging task. In this section, our goal is to propose a novel RKA-based algorithm which is robust to impulsive noise.
The noisy linear inequality feasibility problem is defined as
$\mbC\mbx+\mbn\succeq \mbb$,
where $\mbn$ is the noise of our system.
If the noise does not have a significant impact on the inequality, we can rephrase the system with noise as $\mbC\mbx\succeq \mbb$ and solve it using the RKA. However, if the noise has the potential to alter the direction of the inequalities, we need to handle it differently. In this section, 
specifically, we introduce an algorithm that identifies the ``orthants" that are immune to corruption
(where noise cannot change their direction), and only incorporate them in the updates of RKA. 

The probability of orthants that are corrupted with noise is formulated as the following upper quantile
\begin{equation}
\label{Stef2}
\alpha_j = \operatorname{Pr}\left(\langle\mbc_j,\mbx\rangle-b_j\geq -n_j\right),
\end{equation}
where
$b_j$ and $n_j$ are $j$-th
elements of $\mbb$ and $\mbn$, respectively. By using this formulation, we may be able to determine a threshold for our residuals $\left\{\langle\mbc_j,\mbx\rangle-\mbb_j\right\}$ that can be used to identify the corrupted ones. The threshold is calculated based on the empirical $q$-quantile of the noise:
\begin{equation}
\label{Stef3}
\mathcal{Q}\left(\mbx\right)\triangleq q\text{-quantile}\left\{\left|\langle\mbc_j,\mbx\rangle-b_j\right|,~j\in [m]\right\}.
\end{equation}
If a residual exceeds this threshold, it suggests that the noise may not have been strong enough to change its direction. By applying the thresholding process, we ensure that a sufficient distance is maintained between $\langle\mbc_j,\mbx\rangle$ and $\mbb_j$ to prevent any noise from impacting the inequality system. This helps to preserve the integrity of the solution throughout the algorithm. Assume $\mbp_j$ is the $j$-th row of $\mbP_{\mby_z}$ randomly chosen at each iteration $i$, the proposed algorithm for the noisy one-bit sampled systems \eqref{St_10}, \emph{upper quantile-based ORKA} is written as follows: (i) update the RKA projection if $\left|\langle\mbp_j,\mbx_i\rangle-r^{(\ell)}_j\tau^{(\ell)}_j\right|\geq \mathcal{Q}\left(\mbx_i\right)$, and (ii) otherwise, set $\mbx_i=\mbx_{i+1}$.
\begin{figure*}[t]
	\centering
	\subfloat[$m=2$]
		{\includegraphics[width=0.25\textwidth]{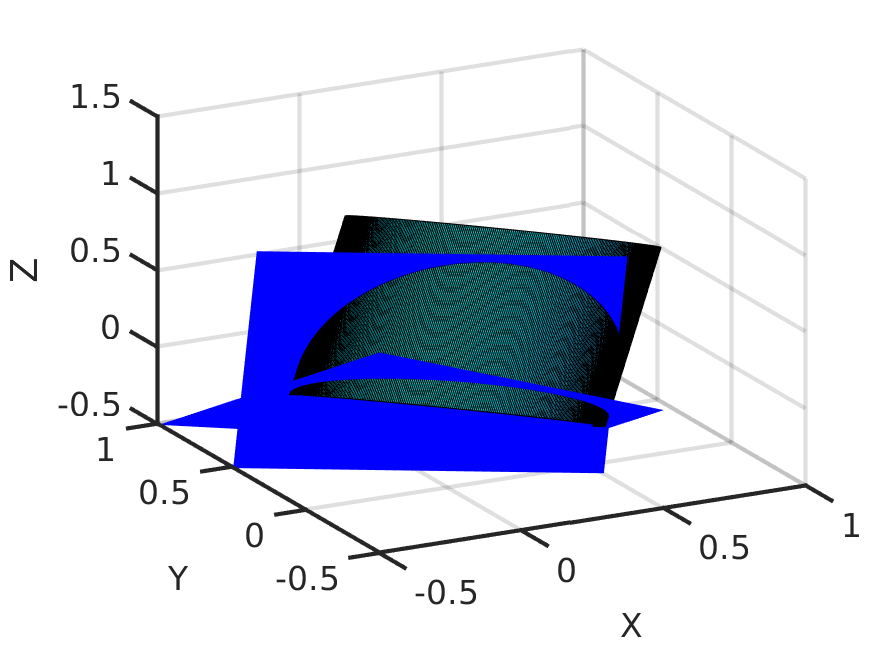}}\qquad
	\subfloat[$m=6$]
		{\includegraphics[width=0.25\textwidth]{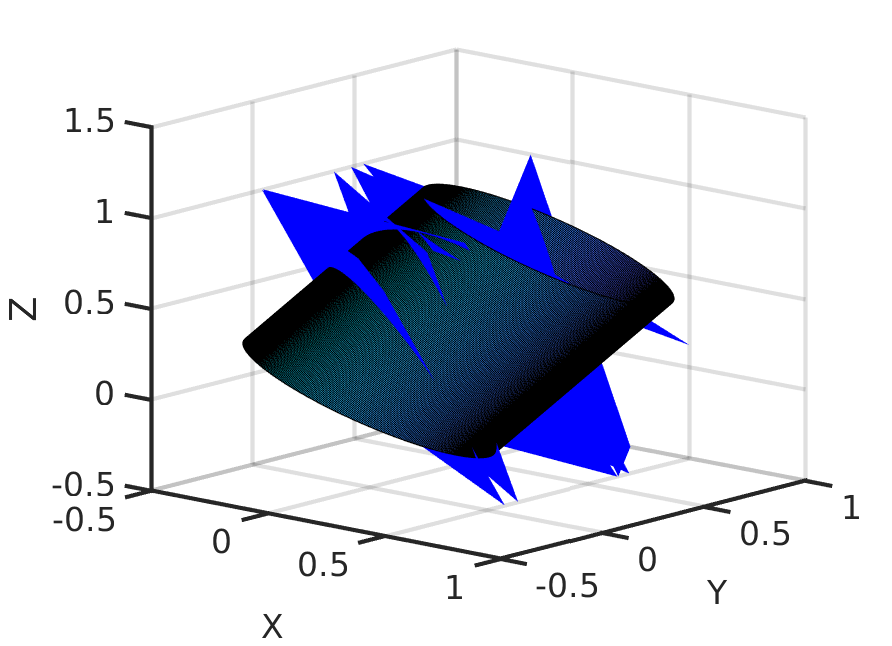}}\qquad
	\subfloat[$m=60$]
		{\includegraphics[width=0.25\textwidth]{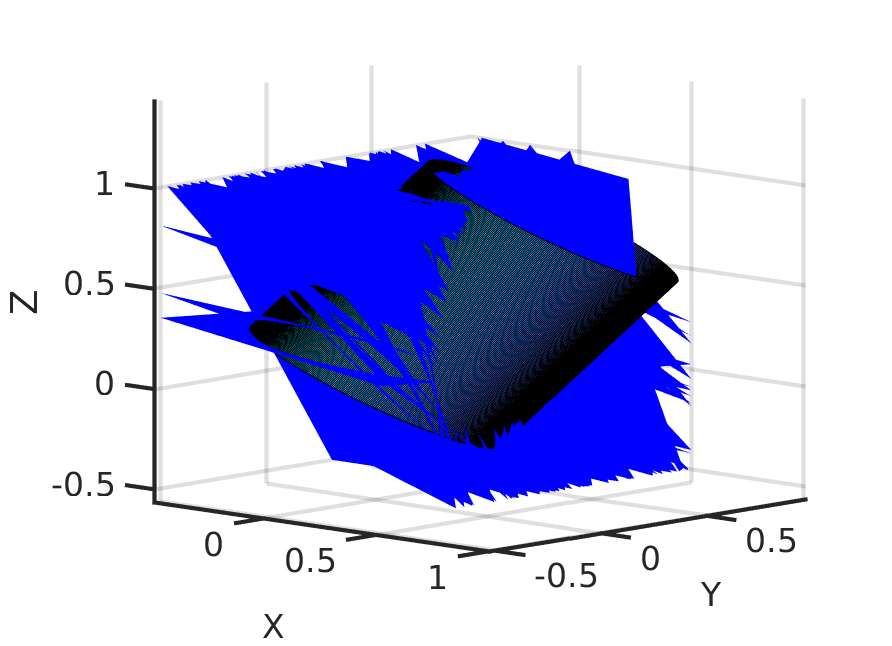}}\\
	\subfloat[$m=2$]
		{\includegraphics[width=0.25\textwidth]{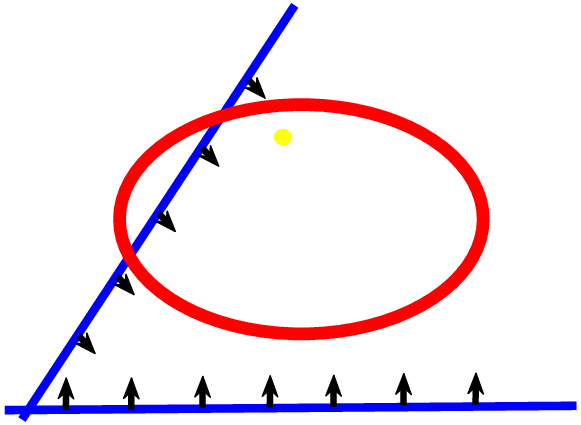}}\qquad
	\subfloat[$m=6$]
		{\includegraphics[width=0.25\textwidth]{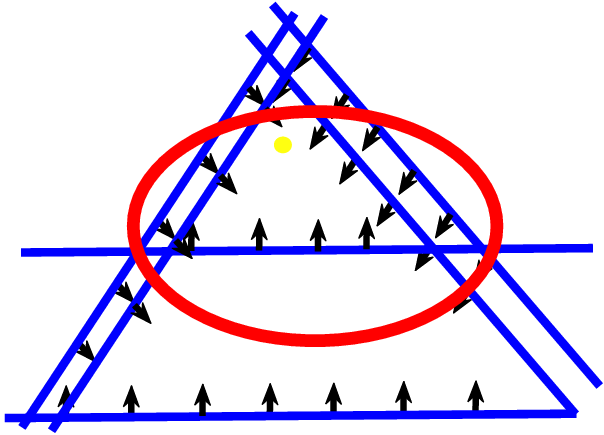}}\qquad
	\subfloat[$m=60$]
		{\includegraphics[width=0.25\textwidth]
        {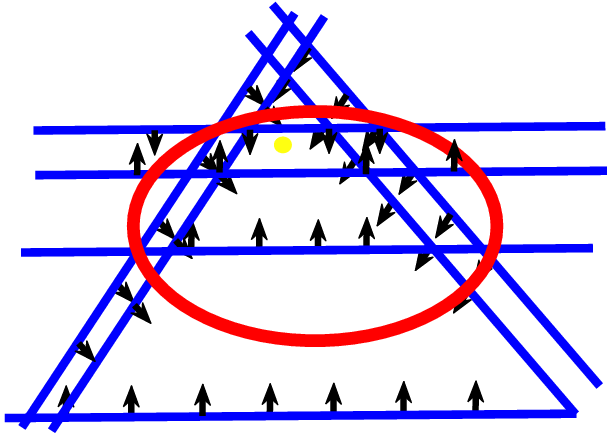}}
	\caption{Shrinkage of the one-bit polyhedron (\ref{eq:80n}) in blue, ultimately placed within the unit ball of the nuclear norm $\left\|\mbX\right\|_{\star}\leq 1$ shown with black cylindrical region and its red contours, when the number of constraints (samples) grows large. The arrows point to the half-space associated with each inequality constraint. The evolution of the feasible regime is depicted with increasing samples in three cases: (a) and (d) small sample-size regime, constraints not forming a finite-value polyhedron; (b) and (e) medium sample-size regime, constraints forming a finite-volume polyhedron, parts of which are outside the cylinder; (c) and (f) large sample-size regime, constraints forming a finite-volume polyhedron inside the nuclear norm cylinder, making its constraint redundant. The original signal representing the signal to be recovered is shown by yellow.
    \vspace{-15pt}
	}
	\label{figure_1n}
\end{figure*}
\section{Accelerated Algorithms for One-Bit Sensing Problem}
We analyze the impact of coarse quantization on two vital inverse problems: the low-rank matrix sensing problem and the compressed sensing problem, considering two scenarios: (i) sample abundance and (ii) sample restriction.
\subsection{Problem Formulation for Low-Rank Matrix Recovery}
\label{pr}
The problem of the low-rank matrix sensing is formulated as:
\begin{equation}
\label{eq:1nnnnn}
\begin{aligned}
\text{find}\quad \mbX \in \Omega_{c} \quad
\text{subject to} \quad \mathcal{A}\left(\mbX\right)=\mathbf{y},~ 
\operatorname{rank}\left(\mbX\right)\leq r,
\end{aligned}
\end{equation}
where $\mbX\in \mathbb{R}^{n_{1}\times n_{2}}$ is the matrix of unknowns, $\mathbf{y}\in \mathbb{R}^{n}$ is the measurement vector, and $\mathcal{A}$ is a linear transformation such that $\mathcal{A}:\mathbb{R}^{n_1\times n_2}\mapsto\mathbb{R^{n}}$.
In general, $\Omega_{c}$ can be chosen such as the set of semi-definite matrices, symmetric matrices, upper or lower triangle matrices, Hessenberg matrices and a specific constraint on the matrix elements $\left\|\mbX\right\|_{\infty}\leq \alpha$ or on its eigenvalues, i.e., $\lambda_{i}\leq \epsilon$ where $\left\{\lambda_{i}\right\}$ are eigenvalues of $\mbX$. The nuclear norm minimization alternative of the problem is given by \cite{cai2010singular,recht2011null}:
\begin{equation}
\label{eq:1nnnnnnn}
\begin{aligned}
\underset{\mbX \in \Omega_{c}}{\textrm{minimize}}\quad \left\|\mbX\right\|_{\star}\quad
\text{subject to} \quad \mathcal{A}\left(\mbX\right)=\mathbf{y}.
\end{aligned}
\end{equation}
In this problem, the feasible set $\mathcal{F}_{\mbX}$ is obtained as
$\mathcal{F}_{\mbX}=\left\{\mathcal{P}^{\star}_{\mbX}\cap \Omega_{c}\right\}$,
where $\mathcal{P}^{\star}_{\mbX}$ is defined as follows
\begin{equation}
\label{stephanie}
\mathcal{P}^{\star}_{\mbX} = \left\{\mbX \mid
\left\|\mbX\right\|_{\star}\leq \tau\right\},\quad
\tau\in\mathbb{R}^{+}.
\end{equation}
Next, we will apply ORKA to (\ref{eq:1nnnnnnn}) to make its costly constraints redundant by using abundant number of one-bit samples $m^{\prime}$.

In low-rank matrix sensing, the linear operator $\mathcal{A}\left(\mbX\right)$ is obtained as
\begin{equation}
\label{Stefanie_2}
\mathcal{A}\left(\mbX\right)=\frac{1}{\sqrt{n}}\left[\operatorname{Tr}\left(\mbA^{\top}_1\mbX\right)\cdots\operatorname{Tr}\left(\mbA^{\top}_n\mbX\right)\right]^{\top},
\end{equation}
where $\mbA_j\in\mathbb{R}^{n_1\times n_2}$ is the $j$-th sensing matrix. The one-bit polyhedron for the low-rank matrix sensing is given by
\begin{equation}
\label{kumar}
\mathcal{P}^{(M)}=\left\{\mbX^{\prime}\in\mathbb{R}^{n_1\times n_2} \mid r^{(\ell)}_{j}\operatorname{Tr}\left(\mbA^{\top}_j\mbX^{\prime}\right)\geq r^{(\ell)}_{j}\tau^{(\ell)}_{j}\right\}\subset \mathbb{R}^{n_1\times n_2},  
\end{equation} 
for all $j\in[n],\ell\in[m]$. Writing the update process of ORKA in matrix form yields the following representation:
\begin{equation}
\label{Stefanie_1}   
\mbX_{i+1} =\mbX_i+\frac{\left(r^{(\ell)}_j\tau^{(\ell)}_j-r^{(\ell)}_j\operatorname{Tr}\left(\mbA^{\top}_j\mbX_i\right) \right)^{+}}{\left\|\mbA_{j}\right\|^{2}_{\mathrm{F}}} \mbA_{j}.
\end{equation}
The convergence rate of this process is equivalent to that of RKA. Denote the update process in \eqref{Stefanie_1} by $\mbX_{i+1}=\mathrm{KA}_{r}(\mbX_i)$.
One can similarly utilize other variants of RKA, PrSKM and Block SKM, and generalize the update process \eqref{Stefanie_1} as $\mbX_{i+1}=\mathrm{KA}_{p}(\mbX_i)$ and $\mbX_{i+1}=\mathrm{KA}_{b}(\mbX_i)$.

A numerical investigation of (\ref{kumar})  
reveals that by increasing the number of samples, the space formed by the intersection of linear constraints can fully shrink to the desired signal $\mbX$ inside the feasible region of (\ref{stephanie}) which is shown by the cylindrical space \cite{recht2011null}---see Fig.~\ref{figure_1n} for an illustrative example of this phenomenon. As can be seen in this figure, the blue lines displaying the linear feasibility form a finite-volume space around the original signal displayed by the yellow circle inside the cylinder (the elliptical region) by growing the number of threshold sequences or one-bit samples. In (a)/(d), constraints are not enough to create a finite-volume space, whereas in (b)/(e) such constraints can create the desired finite-volume polyhedron space which, however, is not fully inside the cylinder. Lastly, in (c)/(f), the created finite-volume space shrinks to be fully inside the cylinder. The theoretical discussion regarding the required number of one-bit samples $m^{\prime}$ to accurately recover the low-rank matrix from one-bit samples will be presented in Theorem~\ref{theorem_1}.

\subsection{ORKA With Low-Rank Matrix Factorization}
\label{matrixfactor}
Although Kaczmarz algorithms use a simple and low-complexity update process, the computational cost can still be prohibitively large for a matrix $\mbX\in\mathbb{R}^{n_1\times n_2}$ when $n_1$ and $n_2$ are large values, typically requiring $\mathcal{O}\left(n_1 n_2\right)$ operations. To address this issue, we employ the well-known \emph{low-rank matrix factorization} technique. Instead of using the full matrix $\mbX$, we use its low-rank factorization $\mbX=\mbL\mbW^{\top}$, where $\mbL\in\mathbb{R}^{n_1\times r}$ and $\mbW\in\mathbb{R}^{n_2\times r}$ are low-rank factors. The advantage of this approach is that the economical representation of the low-rank matrix results in lower storage requirements, and lower per-iteration computational cost \cite{chi2019nonconvex}. It also makes the method more scalable to larger problem sizes, especially when using iterative feasibility problem solver like ORKA. In this section, we explore the integration of low-rank matrix factorization with ORKA and the alternating minimization (AltMin) algorithm discussed in \cite{escalante2011alternating}. 

The one-bit polyhedron $\mathcal{P}^{(M)}$ associated with the low-rank matrix factorization approach is written as
\begin{equation}
\label{kvm_15}
\mathcal{P}^{(M)}_0=\left\{\mbL, \mbW \mid r^{(\ell)}_{j}\operatorname{Tr}\left(\mbA^{\top}_j\mbL\mbW^{\top}\right)\geq r^{(\ell)}_{j}\tau^{(\ell)}_{j}\right\},
\end{equation}
for all $j\in[n],\ell\in [m]$. To find the solution from $\mathcal{P}^{(M)}_0$, we use the idea of AltMin algorithm. We split $\mathcal{P}^{(M)}_0$ into two linear feasibility sub-problems with respect to $\mbL$ and $\mbW$, respectively. Specifically, with respect to $\mbL$ when $\mbW$ is fixed we have:
\begin{equation}
\label{kvm_16}
\mathcal{P}^{(M)}_{\mbL}=\left\{\mbL\mid r^{(\ell)}_{j}\operatorname{Tr}\left(\mbW^{\top}\mbA^{\top}_j\mbL\right)\geq r^{(\ell)}_{j}\tau^{(\ell)}_{j}\right\},
\end{equation}
and with respect to $\mbW$ when $\mbL$ is fixed we have:
\begin{equation}
\label{kvm_17}
\mathcal{P}^{(M)}_{\mbW}=\left\{\mbW \mid r^{(\ell)}_{j}\operatorname{Tr}\left(\mbA^{\top}_j\mbL\mbW^{\top}\right)\geq r^{(\ell)}_{j}\tau^{(\ell)}_{j}\right\},
\end{equation}
for all $j\in[n],\ell\in [m]$. If either $\mbL$ or $\mbW$ are fixed, finding the solution with respect to the other variable is achieved via ORKA as indicated in Algorithm~\ref{algorithm_200}.
\begin{algorithm}[t]
\caption{ORKA (with Low-Rank Matrix Factorization)}
\label{algorithm_200}
\begin{algorithmic}[1]
\Statex \emph{Input:} One-bit samples, time-varying thresholds and the sensing matrices presented in \eqref{kvm_15}, total number of required iterations $T$, and $\operatorname{rank}(\mbX)=r$.
\Statex \emph{Output:} A solution $\bar{\mbX}\in\mathbb{R}^{n_1\times n_2}$ in the polyhedron \eqref{kvm_15}.
\Statex \emph{Note:} $\mbL_t\in\mathbb{R}^{n_1\times r}$ and $\mbW_t\in\mathbb{R}^{n_2\times r}$ denote the obtained matrices at $t$-th iteration, $\mbH_{i+1}=\mathrm{KA}_{r}(\mbH_{i};\Tilde{\mbH})$ denotes the update process of ORKA using the RKA when $\Tilde{\mbH}$ is fixed.
\State $\mbW_{t+1}\gets\mathrm{KA}_{r}(\mbW_t;\mbL_t)$ $\triangleright$ Update process for polyhedron \eqref{kvm_17}.
\State $\mbL_{t+1}\gets\mathrm{KA}_{r}(\mbL_t;\mbW_{t+1})$ $\triangleright$ Update process for polyhedron \eqref{kvm_16}.
\State Repeat steps (1) and (2) until convergence. 
\State $\mbL_{\star}\gets\mbL_T$ and $\mbW_{\star}\gets\mbW_T$.
\State \Return  $\bar{\mbX}\gets\mbL_{\star}\mbW^{\top}_{\star}$.
\end{algorithmic}
\end{algorithm}

\subsection{Curse of Dimensionality: SVP-ORKA}
\label{SVP}
While acquiring a large number of samples is not typically an issue for one-bit sampling, it is still practical to avoid using excess samples, particularly in signal processing applications where access to sufficient measurements may be limited. Adhering to the law of parsimony, ``\emph{Entia non svnt mvltiplicanda præter necessitatem}," i.e., entities should not be multiplied beyond necessity, this section focuses on developing ORKA to facilitate low-rank matrix sensing with a reduced number of one-bit samples.

Define $r$ 
as the predefined rank 
of the unknown matrix $\mbX$.
In order to obtain the solution within a reduced number of samples in the polyhedron $\mathcal{P}^{(M)}$ defined in \eqref{kumar}, we impose 
a rank constraint, $\operatorname{rank}(\mbX)\leq r$,
to shrink the entire space, as shown by the following polyhedron:
\begin{equation}
\label{Stefanie_5}
\mathcal{P}^{(M)}_1=\left\{\mbX^{\prime}\in\mathcal{P}^{(M)} \mid \operatorname{rank}\left(\mbX^{\prime}\right)\leq r\right\}\subset \mathbb{R}^{n_1\times n_2}.
\end{equation}
To tackle this problem, we apply the singular value projection (SVP) to ORKA. The operator $P_{r}$ 
calculates the $r$ largest singular values of a matrix and subsequently rewrites its SVD based on these $r$ singular values and their corresponding singular vectors \cite{chi2019nonconvex}. Similar to Section~\ref{pr}, denote $\mathrm{KA}_r(\cdot)$ as the update process of ORKA using the RKA defined in \eqref{Stefanie_1}. Through the integration of SVP into each iteration of ORKA, we can achieve the following update process:
\begin{equation}
\label{St_20}
\left\{\begin{array}{l}
\mbZ_{i+1}=\mathrm{KA}_{r}(\mbX_i),\\
\mbX_{i+1}=P_{r}\left(\mbZ_{i+1}\right).
\end{array}\right.
\end{equation}
The first step of the update process \eqref{St_20}, $\mbZ_{i+1}=\mathrm{KA}_{r}(\mbX_i)$, can be also replaced with $\mbZ_{i+1}=\mathrm{KA}_{p}(\mbX_i)$ or $\mbZ_{i+1}=\mathrm{KA}_{b}(\mbX_i)$ to be consistent with the updates of PrSKM and Block SKM, respectively.

\subsection{One-Bit CS via Sample Abundance}
\label{ocs_abundance}
The one-bit CS is solely accomplished by creating a highly-constrained one-bit polyhedron:
\begin{equation}
\label{Stefanie_500}
\mathcal{P}^{(C)}=\left\{\mbx^{\prime}\in \mathbb{R}^{d}\mid r^{(\ell)}_{j}\langle\mba_j,\mbx^{\prime}\rangle\geq r^{(\ell)}_{j}\tau^{(\ell)}_{j}\right\}\subset \mathbb{R}^{d},
\end{equation}
for all $j\in[n],\ell\in[d]$. In other words, instead of solving an optimization problem with costly constraints, the problem may be tackled by 
the following update process:
\begin{equation}
\label{kvm_23}
\mbx_{i+1}=\mathrm{KA}_{r}\left(\mbx_i\right),
\end{equation}
where $\mathrm{KA}_r(\cdot)$ denotes the RKA update process presented as $\mbx_{i+1}=\mbx_i+\frac{\left(r^{(\ell)}_j\tau^{(\ell)}_j-r^{(\ell)}_j\langle\mba_j,\mbx_i\rangle\right)^{+}}{\left\|\mba_{j}\right\|^{2}_2} \mba^{\star}_j$. The update process in \eqref{kvm_23} can be also replaced with $\mbx_{i+1}=\mathrm{KA}_{p}\left(\mbx_i\right)$ or $\mbx_{i+1}=\mathrm{KA}_{b}\left(\mbx_i\right)$ to be consistent with the updates of PrSKM and Block SKM, respectively. The theoretical discussion about the minimum number of one-bit samples $m^{\prime}$ to precisely recover the sparse vector from the polyhedron \eqref{Stefanie_500} has been presented in Section~\ref{penalt}.

\subsection{ST-ORKA}
\label{ocs_constrained}
We continue our exploration of algorithms for one-bit CS by introducing the \emph{$\ell_0$ regularized ORKA}. Building on the motivation behind our earlier proposal of the SVP-ORKA, ORKA is developed to efficiently recover
a sparse solution from the one-bit CS polyhedron using a reduced number of one-bit samples. The regularized one-bit CS polyhedron is obtained as
\begin{equation}
\label{Stefanie_50}
\mathcal{P}^{(C)}_1=\left\{\mbx^{\prime}\in \mathbb{R}^{d}\mid r^{(\ell)}_{j}\langle\mba_j,\mbx^{\prime}\rangle\geq r^{(\ell)}_{j}\tau^{(\ell)}_{j},~ \left\|\mbx^{\prime}\right\|_0\leq k\right\}\subset \mathbb{R}^{d},
\end{equation} 
where $j\in[n]$ and $\ell\in[m]$. 
Define the soft thresholding (ST) operator as $S_{\upkappa}\left(\mbx\right)= \operatorname{sgn}\left(\mbx\right)\left(|\mbx|-\mbt_1\right)^{+}$, where $\mbt_1$ is the predefined threshold. The ST-ORKA utilizes the 
ST operator $S_{\upkappa}$ 
to project each iteration of ORKA to the set $\left\{\|\mbx\|_0\leq k\right\}$
with the following update process:
\begin{equation}
\label{St_21}
\left\{\begin{array}{l}
\mbz_{i+1}=\mathrm{KA}_r\left(\mbx_i\right),\\
\mbx_{i+1} = S_{\upkappa}\left(\mbz_{i+1}\right),
\end{array}\right.
\end{equation}
where $\mathrm{KA}_r(\cdot)$ denotes the RKA updates. The first step of the update process \eqref{St_21} can be also replaced with $\mbz_{i+1}=\mathrm{KA}_p\left(\mbx_i\right)$ or $\mbz_{i+1}=\mathrm{KA}_b\left(\mbx_i\right)$ to be consistent with the updates of PrSKM and Block SKM, respectively. 
\subsection{Judicious Sampling With Adaptive Thresholding for ORKA}
\label{ada_thresh}
By the spirit of using the iterative RKA, suitable time-varying sampling thresholds can be selected in order to enhance the recovery performance. To capture the desired signal $\mathbf{x}$ more efficiently, the right-hand side of the inequalities in (\ref{eq:80n}), i.e. $\operatorname{vec}\left(\mbR_{\mby}\right)\odot \operatorname{vec}\left(\bGamma\right)$, must be determined in a way that each associated hyperplane passes through the desired feasible region within $\mathcal{F}_{\mbx}$. Accordingly, we propose an iterative algorithm generating adaptive sampling thresholds to accurately obtain a good solution. To have a smaller area of the finite-volume space around the signal $\mathbf{x}$, one can somehow choose thresholds to reduce the distances between 
the desired point and the associated hyperplanes in \eqref{eq:80n}. To do so, we update the time-varying thresholds 
as
\begin{equation}
\label{eq:200}
\left\{\begin{array}{l}
\boldsymbol{\uptau}^{(\ell)}_{k+1}=\mbA\mbx_{k}-\frac{1}{2}\left(\mbr^{(\ell)}\odot\boldsymbol{\epsilon}^{(\ell)}_{k}\right),\\\boldsymbol{\epsilon}^{(\ell)}_{k}=\mbr^{(\ell)}\odot\left(\mbA\mbx_k-\boldsymbol{\uptau}_{k}^{(\ell)}\right),
\end{array}\right.
\end{equation}
where $\mbx_k$ and $\boldsymbol{\uptau}_{k}^{(\ell)}$ denote the $k$-th updates of $\mbx$ and $\boldsymbol{\uptau}^{(\ell)}$ in our proposed adaptive threshold design algorithm, respectively. 
The proposed sampling algorithm is summarized in Algorithm~\ref{algorithm_20}.
\begin{algorithm}[t]
\caption{Adaptive Thresholding for ORKA}
\label{algorithm_20}
\begin{algorithmic}[1]
\Statex \emph{Input:} One-bit data,
time-varying sampling thresholds,
and the sampling matrix. 
\Statex \emph{Output:} A solution $\bar{\mbx}$ in the one-bit polyhedron \eqref{eq:80n}.
\State $\mbb^{(\ell)}\gets\boldsymbol{\uptau}_{k}^{(\ell)}$.
\State $\mathcal{P}_{k}\gets\left\{\mathbf{x}_{k}\mid r_j^{(\ell)}\langle\mba_j,\mbx_k\rangle\geq r_j^{(\ell)}b_j^{(\ell)},~j\in[n],~\ell\in[m]\right\}$.
\State Obtain $\mbx_k$ in $\mathcal{P}_{k}$ by RKA, PrSKM or Block SKM.
\State $\boldsymbol{\epsilon}^{(\ell)}_{k}\gets\mbr^{(\ell)}\odot\left(\mbA\mbx_k-\mbb^{(\ell)}\right)$.
\State $\boldsymbol{\uptau}^{(\ell)}_{k+1}\gets\mbA\mbx_{k}-\frac{1}{2}\left(\mbr^{(\ell)}\odot\boldsymbol{\epsilon}^{(\ell)}_{k}\right)$.
\State Increase $k$ by one.
\State Repeat Steps ($1$)-($6$) until
$\sum_{\ell=1}^{m}\left\|\boldsymbol{\uptau}^{(\ell)}_{k+1}-\boldsymbol{\uptau}^{(\ell)}_{k}\right\|_{2}\leq\delta$.
\State \Return $\bar{\mbx}$
\end{algorithmic}
\end{algorithm}
\section{Finite Volume Property:\\ Theoretical Guarantees for One-bit Sensing}
\label{ada_thresh_prob}
It has been demonstrated that in the ditherless scenario, and if the measurements are bounded within a specific set with high probability, there exists a subset of hyperplanes such that the radius of the cell containing $\mbx$ induced by these hyperplanes is bounded \cite{lybrand2022number}. In this section, we demonstrate an alternative approach to determining an upper bound. Instead of confining measurements to a bounded space, we leverage the one-bit polyhedron derived from the dithering scheme. Our method involves approaching the average of distances between the signal and surrounding hyperplanes toward its mean. We establish that as the number of samples increases, the finite volume or cell created around the signal diminishes.
We will introduce the concept of FVP and subsequently obtain the required number of one-bit samples $m^{\prime}$ to capture a solution inside a ball centered at $\mbx$. 
We begin this section by providing the following definitions:
\begin{definition}
\label{def_2}
Define the distance between a signal $\mbx\in\mathcal{T}$ and the hyperplanes of the polyhedron \eqref{eq:80n} as 
\begin{equation}
\label{dist_1}
d^{(\ell)}_{j}=\left|\langle\mba_{j},\mathbf{x}\rangle-\tau^{(\ell)}_{j}\right|,~j\in[n], ~\ell\in[m].
\end{equation}
Then, we denote the average of such distances by
\begin{equation}
\label{a_7}
T_{\mathrm{ave}}(\mbx)=\frac{1}{mn}\sum_{\ell=1}^{m}\sum_{j=1}^{n}\left(d^{(\ell)}_{j}\right).
\end{equation} 
\end{definition}
It is essential to clarify that in our analysis, we adopt the worst-case scenario for the distance
between the desired point and the solutions. This approach considers the solution to lie precisely on the hyperplanes. However, it is crucial to note that in reality, the solution to a linear inequality
system can exist anywhere between the desired point and the hyperplane. Utilizing the average of distances, we formulate the FVP and its theorems for both arbitrary and structured sets. As the number of samples increases, and the average of distances approaches its mean, we obtain sufficient samples to construct a finite volume and locate the solution within a ball. This concept is expounded upon in our theorems. 
\begin{definition}[Consistent Reconstruction Property]
\label{def_1}
Define a signal as $\mbx\in\mathcal{T}$. Denote $\bar{\mbx}\in\mathcal{T}$ as a solution obtained by an arbitrary reconstruction algorithm addressing the one-bit feasibility problem. Then, such a reconstruction algorithm is said to be consistent when
\begin{equation}
\label{a_6}
\begin{aligned}
\operatorname{sgn}\left(\langle\mba_j,\mbx\rangle-\tau_j^{(\ell)}\right)&=\operatorname{sgn}\left(\langle\mba_j,\bar{\mbx}\rangle-\tau_j^{(\ell)}\right),
\end{aligned}
\end{equation}
for all $j\in[n],\ell\in[m]$.   
\end{definition}
Note that the sign preservation (consistency) mentioned in our work, as stated in Definition~\ref{def_1}, is equivalent to the solution provided by a linear feasibility solver, denoted as $\bar{\mbx}$, satisfying all the given inequalities, meaning $\bar{\mbx}\in\mathcal{P}_{\mbx}$. In Section~\ref{penalt}, we aim to demonstrate that having a sufficient number of samples guarantees the formation of a finite volume of intersection of hyperplanes. This, in turn, implies that the solution will be confined within a ball of radius $\rho$. We will establish the relationship between this radius and FVP distortion in Section~\ref{reco}.
\begin{definition}[Isotropic Property]
\label{def_3}
A random vector $\mba$ is said to be isotropic if its covariance matrix is the identity matrix; i.e.
\begin{equation}
\label{iso}
\mathbb{E}\left\{\mba\mba^{\mathrm{H}}\right\}=\mbI.
\end{equation}
A random matrix $\mbA$ is said to be isotropic if all its rows are isotropic.
\end{definition}
In our proofs, we establish that the sole assumption on the sampling matrix is to be isotropic. With this assumption, we derive a concentration inequality for the average of distances $T_{\mathrm{ave}}$. This inequality serves as the basis for providing an upper bound on the recovery error $\rho$.
\subsection{Finite Volume Property}
\label{penalt}
In the following theorem, we derive a concentration inequality for the average of distances in a probability distribution. This inequality reveals that the upper bound of the average of distances is contingent on the signal norm, establishing an embedding which we call it the FVP. Initially demonstrated for arbitrary sets, we subsequently extend this theorem to encompass various well-known structured sets: 
\begin{theorem}[FVP for an Arbitrary Set]
\label{theorem_0}
Consider an $n\times d$ isotropic sub-gaussian sampling matrix $\mbA$ with $\left\|\mba_j\right\|_{\psi_2}\leq K$ for all $j\in[n]$, and a signal $\mbx\in\mathcal{K}$ where $\mathcal{K}\subseteq\mathbb{R}^{d}$ is an arbitrary set. The one-bit feasibility problem with total number of one-bit samples $m^{\prime}$, is obtained by uniform dithering following $\mathcal{U}_{[-\lambda,\lambda]}$. 
If a solution $\bar{\mbx}$ satisfies the consistent reconstruction property in Definition~\ref{def_1} and $\bar{\mbx}\in\mathcal{B}_{\rho}\left(\mbx\right)$, we have the following concentration inequality with positive constants $c,\epsilon\geq 0$: 
\begin{equation}
\label{a_90}
\begin{aligned}
\operatorname{Pr}\left(\sup_{\mbx\in\mathcal{K}}\left|T_{\mathrm{ave}}\left(\mbx\right)-\frac{\lambda}{2}-\frac{\left\|\mbx\right\|_{2}^2}{2\lambda}\right|\geq\epsilon\right)\leq 2 e^{-\frac{c \epsilon^2 m^{\prime}}{2\left(K+\frac{\lambda}{\sqrt{\log(2)}}\right)^2}},
\end{aligned}
\end{equation}
where the required number of samples meets $m^{\prime}\gtrsim \frac{\gamma^2\left(\mathcal{K}\right)}{\rho^2 \epsilon^2}$. 
\end{theorem}
The proof of Theorem~\ref{theorem_0} is presented in Appendix~\ref{ealmaz}. According to the FVP theorem, the average of distances finds a quasi-isometry with a distortion parameter $\epsilon$.  
The following theorem presents the FVP for the set of low-rank matrices:
\begin{theorem}[FVP for Low-Rank Matrices]
\label{theorem_1}
Consider a low-rank matrix $\mbX\in\mathcal{K}_r$ where $\mathcal{K}_r$ is defined as
\begin{equation}
\label{boz}
\mathcal{K}_{r}=\left\{\mbX\in\mathbb{R}^{n_1\times n_2}\mid \operatorname{rank}(\mbX)\leq r,\|\mbX\|_{\mathrm{F}}\leq 1\right\}.
\end{equation}
Consider an $n\times (n_1n_2)$ sub-gaussian sampling matrix $\mbA$ with $\left\|\mba_j\right\|_{\psi_2}\leq K$ for all $j\in[n]$. The one-bit feasibility problem with total number of one-bit samples $m^{\prime}$, is obtained by uniform dithering following $\mathcal{U}_{[-\lambda,\lambda]}$. 
If a solution $\bar{\mbX}$ satisfies the consistent reconstruction property in Definition~\ref{def_1} and $\bar{\mbX}\in\mathcal{B}_{\rho}\left(\operatorname{vec}\left(\mbX\right)\right)$, we have the following concentration inequality with positive constants $c,\epsilon\geq 0$:
\begin{equation}
\label{a_9}
\begin{aligned}
\operatorname{Pr}\left(\sup_{\mbX\in\mathcal{K}_r}\left|T_{\mathrm{ave}}\left(\mbX\right)-\frac{\lambda}{2}-\frac{\left\|\mbX\right\|_{\mathrm{F}}^2}{2\lambda}\right|\geq\epsilon\right)\leq 2 e^{-\frac{c \epsilon^2 m^{\prime}}{2\left(K+\frac{\lambda}{\sqrt{\log(2)}}\right)^2}},
\end{aligned}
\end{equation}
where the required number of samples satisfies
\begin{equation}
\label{lowrank}
m^{\prime} \gtrsim \epsilon^{-2}r \left(n_1+n_2\right)\log\left(1+\frac{1}{\rho}\right).
\end{equation}
\end{theorem}
For the bounded sparse signal set, we provide the following FVP theorem: 
\begin{theorem}[FVP for Sparse Signals]
\label{theorem_3}
Consider an sparse signal $\mbx\in\mathcal{K}_s$ where $\mathcal{K}_s$ is defined as
\begin{equation}
\label{boz0}
\mathcal{K}_{s}=\left\{\mbx\in\mathbb{R}^{d}\mid \left\|\mbx\right\|_{0} = s,\|\mbx\|_{2}\leq 1\right\}.
\end{equation}
Consider an $n\times d$ isotropic sub-gaussian sampling matrix $\mbA$ with $\left\|\mba_j\right\|_{\psi_2}\leq K$ for all $j\in[n]$. The one-bit feasibility problem with total number of one-bit samples $m^{\prime}$, is obtained by uniform dithering following $\mathcal{U}_{[-\lambda,\lambda]}$. 
If a solution $\bar{\mbx}$ satisfies the consistent reconstruction property in Definition~\ref{def_1} and $\bar{\mbx}\in\mathcal{B}_{\rho}\left(\mbx\right)$, we have the following concentration inequality with positive constants $c,\epsilon\geq 0$: 
\begin{equation}
\label{a_91}
\begin{aligned}
\operatorname{Pr}\left(\sup_{\mbx\in\mathcal{K}_s}\left|T_{\mathrm{ave}}\left(\mbx\right)-\frac{\lambda}{2}-\frac{\left\|\mbx\right\|_{2}^2}{2\lambda}\right|\geq\epsilon\right)\leq 2 e^{-\frac{c \epsilon^2 m^{\prime}}{2\left(K+\frac{\lambda}{\sqrt{\log(2)}}\right)^2}},
\end{aligned}
\end{equation}
where the required number of one-bit samples meets 
\begin{equation}
\label{sparse}
m^{\prime}\gtrsim \epsilon^{-2}s\log\left(\frac{d}{s}\right)\log\left(1+\frac{1}{\rho}\right).
\end{equation}
\end{theorem}
The proof of FVP theorems for the bounded low-rank and sparse sets follows a similar outline to that of Theorem~\ref{theorem_0}, with one crucial distinction. In Theorem~\ref{theorem_0}, the Kolmogorov $\rho$-entropy is bounded using Sudakov’s minoration, which has been demonstrated to be non-tight for structured sets \cite{oymak2015near}. For structured sets such as the low-rank matrix set, the upper bound on Kolmogorov $\rho$-entropy is constrained as follows \cite[Table~1]{jacques2017time}:
\begin{equation}
\mathcal{H}\left(\mathcal{K}_r,\rho\right) \lesssim  r \left(n_1+n_2\right)\log\left(1+\frac{1}{\rho}\right),
\end{equation}
and for sparse set
\begin{equation}
\mathcal{H}\left(\mathcal{K}_s,\rho\right) \lesssim  s\log\left(\frac{d}{s}\right)\log\left(1+\frac{1}{\rho}\right).
\end{equation}
In our next result, we present the FVP theorem specific to the DCT sensing model for an arbitrary set $\mathcal{K}$.
\begin{theorem}[FVP for DCT Measurements]
\label{theorem_4}
Denote the DCT coefficients of a signal $\mbx=[x_t]\in\mathcal{K}$ by
\begin{equation}
\label{dct_1}
y_k=\sum_{t=0}^{d-1} \cos(2\pi\omega_kt)x_t,\quad k\in[n],
\end{equation}
where the frequencies $\omega_k$ are uniformly chosen  at random in $\left\{0,1/d,2/d,\cdots,(d-1)/d\right\}$. The one-bit feasibility problem with total number of one-bit samples $m^{\prime}$, is obtained by uniform dithering following $\mathcal{U}_{[-\lambda,\lambda]}$. 
If a solution $\bar{\mbx}$ satisfies the consistent reconstruction property in Definition~\ref{def_1} and $\bar{\mbx}\in\mathcal{B}_{\rho}\left(\mbx\right)$, we have the following concentration inequality with a positive constant $\epsilon\geq 0$:
\begin{equation}
\label{dct_2}
\operatorname{Pr}\left(\sup_{\mbx\in\mathcal{K}}\left|T_{\mathrm{ave}}\left(\mbx\right)-\frac{\lambda}{2}-\frac{\left\|\mbx\right\|_{2}^2}{4\lambda}\right|\geq\epsilon\right)\leq 2 e^{-\frac{\epsilon^2m^{\prime}}{4\lambda^2}},
\end{equation}
where the required number of samples meets $m^{\prime}\gtrsim \frac{\gamma^2\left(\mathcal{K}\right)}{\rho^2 \epsilon^2}$.
\end{theorem}
The proof of Theorem~\ref{theorem_4} is presented in Appendix~\ref{dct_theorem}. It is crucial to emphasize that the results outlined in Theorems~\ref{theorem_0}-\ref{theorem_4} rely on the dynamic range (DR) guarantee assumption. Specifically, this assumption requires the parameter of uniform dithering $\lambda$ satisfies
$\lambda\geq\sup_{j\in[n], \mbx\in\mathcal{K}}|\langle\mba_j,\mbx\rangle|$. 
As an example, consider $K$-sub-gaussian measurements. To meet the DR guarantee, the scale parameter is defined by $\lambda = C K \sqrt{\log(m^{\prime})}$, where $C$ represents a positive constant \cite{thrampoulidis2020generalized}. 
However, in practical scenarios, the sensing process may experience corruption, leading to non-identically distributed sensing vectors. We address this situation where, due to these conditions, one measurement may not satisfy the DR guarantee. Our aim is to explore the impact of this scenario on the FVP theorem.
The following theorem presents the FVP for such scenarios:
\begin{theorem}[FVP Without DR Guarantee]
\label{theorem_5}
Consider an $n\times d$ isotropic sub-gaussian sampling matrix $\mbA$, 
and a signal $\mbx\in\mathcal{K}$ where $\mathcal{K}\subseteq\mathbb{R}^{d}$ is an arbitrary set. The one-bit feasibility problem with total number of one-bit samples $m^{\prime}$, is obtained by uniform dithering following $\mathcal{U}_{[-\lambda,\lambda]}$. Define an index set $\mathcal{I}=[n]\setminus\{j^{\prime}\}$, where $j^{\prime}$ is an arbitrary index value satisfying the condition $\left\|\mba_j\right\|_{\psi_2}\leq K$ for all $j\in\mathcal{I}$, and $\left\|\mba_{j^{\prime}}\right\|_{\psi_2}\leq K^{\prime}$ with $K^{\prime}>K$. Assume that we have
\begin{equation}
\label{spider}
\begin{aligned}
\sup_{j\in\mathcal{I},\mbx\in\mathcal{K}}\left|\langle\mba_j,\mbx\rangle\right|\leq\lambda,\quad\inf_{\mbx\in\mathcal{K}}\langle\mba_{j^{\prime}},\mbx\rangle\geq\lambda.
\end{aligned}
\end{equation}
If a solution $\bar{\mbx}$ satisfies the consistent reconstruction property in Definition~\ref{def_1} and $\bar{\mbx}\in\mathcal{B}_{\rho}\left(\mbx\right)$, we have the following concentration inequality with positive constants $c,\epsilon\geq 0$:
\begin{equation}
\label{noDR}
\begin{aligned}
\operatorname{Pr}\left(\sup_{\mbx\in\mathcal{K}}\left|T_{\mathrm{ave}}\left(\mbx\right)-\frac{n-1}{n}\left(\frac{\lambda}{2}+\frac{\left\|\mbx\right\|_{2}^2}{2\lambda}\right)-\frac{1}{n}\mu^{\prime}\right|\geq\epsilon\right)\leq 2 e^{-\frac{c \epsilon^2 m^{\prime}}{2\left(K^{\prime}+\frac{\lambda}{\sqrt{\log(2)}}\right)^2}},
\end{aligned}
\end{equation} 
where $\mu^{\prime}=\mathbb{E}\left\{\langle\mba_{j^{\prime}},\mbx\rangle\right\}$, and the required number of samples meets $m^{\prime}\gtrsim \frac{\gamma^2\left(\mathcal{K}\right)}{\rho^2 \epsilon^2}$.
\end{theorem}
The proof of Theorem~\ref{theorem_5} is presented in Appendix~\ref{Dr_guarantee}. The main question that arises is the significance of these embeddings connecting the average of distances and the signal norm. The next subsection addresses this question by establishing a relationship between the FVP theorems and the recovery error.
\subsection{Recovery Performance}
\label{reco}
In the following proposition, we illustrate that the upper bound of the recovery error in one-bit sensing is contingent on the FVP distortion parameter $\epsilon$:
\begin{proposition}
\label{error1}   
Under assumptions of Theorem~\ref{theorem_0}, the following upper recovery bound holds for all $\mbx,\bar{\mbx}\in\mathcal{K}$ satisfying the consistent reconstruction property in Definition~\ref{def_1}:
\begin{equation}
\label{a_8}
\|\mbx-\bar{\mbx}\|_2\leq\rho=4\sqrt{\epsilon\lambda},
\end{equation}  
with a probability exceeding $1-2 e^{-\frac{c \epsilon^2 m^{\prime}}{2\left(K+\frac{\lambda}{\sqrt{\log(2)}}\right)^2}}$. 
\end{proposition}
The proof of Proposition~\ref{error1} is presented in Appendix~\ref{er_1}.
\begin{corollary}
The rate of reduction for the FVP distortion $\epsilon$ concerning $m^{\prime}$ is characterized at most by $\mathcal{O}\left(\left(m^{\prime}\right)^{-\frac{1}{3}}\right)$.
\end{corollary}
\begin{IEEEproof}
According to Theorem~\ref{theorem_0} and Proposition~\ref{error1}, for an arbitrary set, the required number of one-bit samples has to meet the following bound:
\begin{equation}
m^{\prime}\gtrsim\epsilon^{-3}\gamma^2\left(\mathcal{K}\right),
\end{equation}
which completes the proof. 
\end{IEEEproof}
According to the aforementioned proposition, \emph{reducing the FVP distortion or approaching the mean of the average of distances results in a diminished finite volume around the signal.} In simpler terms, this proposition underscores the correlation between the average of distances and the finite volume created by hyperplanes around the signal.

When a solution fails to meet the consistent reconstruction property or cannot satisfy certain inequalities in the linear feasibility problem, the Hamming distance emerges in the upper bound of the recovery error, as outlined in the following proposition:
\begin{proposition}
\label{error2}   
Under assumptions of the FVP, if a solution $\bar{\mbx}$ does not meet the consistent reconstruction property in Definition~\ref{def_1}, the following upper recovery bound holds for all $\mbx,\bar{\mbx}\in\mathcal{K}$ with a probability exceeding $1-2 e^{-\frac{c \epsilon^2 m^{\prime}}{2\left(K+\frac{\lambda}{\sqrt{\log(2)}}\right)^2}}$:
\begin{equation}
\label{a_80}
\|\mbx-\bar{\mbx}\|_2\leq\rho=4\sqrt{\epsilon\lambda}+2\sqrt{\left(1+\lambda^2\right)d_{\mathrm{H}}\left(\mbr,\bar{\mbr}\right)},
\end{equation}
where $\mbr$ and $\bar{\mbr}$ are the one-bit data obtained from the signal $\mbx$ and the solution $\bar{\mbx}$, respectively.
\end{proposition}
The proof of Proposition~\ref{error2} is presented in Appendix~\ref{er_2}. Our next result provides the upper bound on the recovery error when there is one measurement which does not satisfy the DR guarantee.
\begin{proposition}
\label{error3}
Under assumptions of Theorem~\ref{theorem_5}, the following upper recovery bound holds for all $\mbx,\bar{\mbx}\in\mathcal{K}$ satisfying the consistent reconstruction property in Definition~\ref{def_1}:
\begin{equation}
\label{upper}
\left\|\mbx-\bar{\mbx}\right\|_2\leq\rho=4\sqrt{\epsilon\lambda\left(\frac{n}{n-1}\right)},
\end{equation}
with a probability exceeding $1-2 e^{-\frac{c \epsilon^2 m^{\prime}}{2\left(K^{\prime}+\frac{\lambda}{\sqrt{\log(2)}}\right)^2}}$.
\end{proposition}
The proof of Proposition~\ref{error3} is presented in Appendix~\ref{er_3}. In the subsequent corollary, we extend the result of Proposition~\ref{error3} to encompass $L$ measurements that do not adhere to the DR guarantee.
\begin{corollary}
\label{ghayum}
Under assumptions of Theorem~\ref{theorem_5}, assume that $L$ measurements do not adhere to the DR guarantee. Then, the following upper recovery bound holds for all $\mbx,\bar{\mbx}\in\mathcal{K}$ satisfying the consistent reconstruction property in Definition~\ref{def_1}:
\begin{equation}
\label{upper_1}
\left\|\mbx-\bar{\mbx}\right\|_2\leq\rho=4\sqrt{\epsilon\lambda\left(\frac{n}{n-L}\right)},
\end{equation}
with a probability exceeding $1-2 e^{-\frac{c \epsilon^2 m^{\prime}}{2\left(K^{\prime}+\frac{\lambda}{\sqrt{\log(2)}}\right)^2}}$.
\end{corollary}
The proof of Corollary~\ref{ghayum} naturally extends from the proof of Proposition~\ref{error3} and is therefore omitted.
\section{Convergence Analysis of Proposed Algorithms}
In this section, leveraging FVP theorems, we delve into the convergence rate of the proposed randomized algorithms.
\subsection{Recovery Error Upper Bound for ORKA}
\label{error}
As the scaled condition number is the central parameter governing the recovery error of the RKA or its variants, in the following theorem we will evaluate the scaled condition number of the matrix $\mbP_{\mby}$ defined in \eqref{eq:90} to unveil the connection between the convergence bounds of the RKA (or its variants) and ORKA. Throughout the remainder of this section, we simplify notation by representing $\mbP_{\mby}$ as $\mbP$.
\begin{theorem}
\label{theorem_2}
Consider the one-bit polyhedron \eqref{eq:80n} associated with $\mbA\mbx=\mby$ with $\mbA\in\mathbb{R}^{n\times d}$. Then, the scaled condition number of the matrix $\mbP$ is equal to that of $\mbA$,
\begin{equation}
\kappa\left(\mbP\right)=\kappa\left(\mbA\right).   
\end{equation}
\end{theorem}
The proof of Theorem~\ref{theorem_2} is provided in Appendix~\ref{akk}.
Considering $\mbA=\mbI$ corresponds to the one-bit sampled signal sensing problem with $\mbP=\Tilde{\bOmega}$ as formulated in \eqref{eq:8n}. In the following corollary, we present the scaled condition number of the matrix $\Tilde{\bOmega}$ in the light of Theorem~\ref{theorem_2}:
\begin{corollary}
\label{col_1}
For $\mbA=\mbI$, corresponding to the one-bit sampled signal reconstruction problem formulated in \eqref{eq:8n}, the scaled condition number of $\Tilde{\bOmega}=\left[\begin{array}{c|c|c}
\bOmega^{(1)} &\cdots &\bOmega^{(m)}
\end{array}\right]^{\top},\quad\Tilde{\bOmega}\in \{-1,0,1\}^{m n\times n}$ is $\kappa\left(\Tilde{\bOmega}\right)=\sqrt{n}$, which is the infimum of the scaled condition number as was shown in Theorem~\ref{scaled_number}.
\end{corollary}
Based on Theorem~\ref{theorem_2} and the convergence bound of the RKA, it is concluded that ORKA converges to the feasible set of one-bit polyhedron. Therefore,
the convergence bound of
ORKA is independent of the number of samples $m^{\prime}$. 
We seek a solution $\widehat{\mbx}$ that lies within a ball space centered at the desired signal, characterized by $\rho$. In other words, we aim for $\widehat{\mbx}$ to belong to the ball space $\mathcal{B}_{\rho}\left(\mbx\right)$. To achieve that,
we are required to increase the number of samples in the polyhedron, as typically provided via one-bit sensing.
Define $\operatorname{vol}\left(\mathcal{P}_{\mbx}\right)$ as the volume space created by the intersection of hyperplanes in $\eqref{eq:80n}$. In the following proposition, we present the
convergence rate of ORKA:
\begin{proposition}[Convergence rate of ORKA]
\label{penaltyyy}
Consider the one-bit polyhedron $\mathcal{P}_{\mbx}$ obtained in \eqref{eq:80n} associated with $\mbA\mbx=\mby$.
Consider a ball $\mathcal{B}_{\rho}\left(\mbx\right)$ and $\widehat{\mbx}\in\mathcal{P}_{\mbx}$, a convergence rate for ORKA is formulated as:
\begin{equation}
\label{bound2}
\begin{aligned}
\mathbb{E}\left\{\left\|\mbx_i-\mbx\right\|_2\right\} \leq \left(1-\frac{1}{\kappa^{2}\left(\mbA\right)}\right)^{\frac{i}{2}} \left\|\mbx_0-\widehat{\mbx}\right\|_2+\rho,
\end{aligned}
\end{equation}
with a probability higher than $1-2 e^{-\frac{c \epsilon^2 m^{\prime}}{2\left(K+\frac{\lambda}{\sqrt{\log(2)}}\right)^2}}$.
\end{proposition}
\begin{IEEEproof}
As demonstrated in \cite{leventhal2010randomized}, the solution obtained from RKA lies in the space formed by the hyperplanes of the linear inequality problem with the following convergence rate:
\begin{equation}
\label{bound200}
\begin{aligned}
\mathbb{E}\left\{\left\|\mbx_i-\widehat{\mbx}\right\|_2\right\} \leq \left(1-\frac{1}{\kappa^{2}\left(\mbA\right)}\right)^{\frac{i}{2}} \left\|\mbx_0-\widehat{\mbx}\right\|_2,
\end{aligned}
\end{equation}
where $\widehat{\mbx}$ is a point inside the space created by a polyhedron $\mathcal{P}_{\mbx}$. The convergence rate \eqref{bound200} only ensures that the solution will lie within the space created by the hyperplanes, not necessarily within the ball around the desired solution. However, in order to ensure that the solution of the linear feasibility problem lies within the ball around the desired solution with radius $\rho$, it is essential to have a sufficient number of samples. Once we obtain a sufficient number of samples to have the volume created by the intersection of hyperplanes inside the ball, i.e., $\operatorname{vol}\left(\mathcal{P}_{\mbx}\right) \subseteq \operatorname{vol}\left(\mathcal{B}_{\rho}(\mbx)\right)$, we then have $\widehat{\mbx}$  lying within the ball around the desired point $\mbx$, i.e., $\widehat{\mbx} \in \mathcal{B}_{\rho}\left(\mbx\right)$. To address this discrepancy between the two scenarios, we introduce a second term that is dependent on the error between $\widehat{\mbx}$ and $\mbx$, as follows:
\begin{equation}
\begin{aligned}
\mathbb{E}\left\{\left\|\mbx_i-\mbx\right\|_{2}\right\} &= \mathbb{E}\left\{\left\|\mbx_i-\widehat{\mbx}+\widehat{\mbx}-\mbx\right\|_{2}\right\}\\&\leq \mathbb{E}\left\{\left\|\mbx_i-\widehat{\mbx}\right\|_{2}\right\} + \mathbb{E}\left\{\left\|\mbx-\widehat{\mbx}\right\|_{2}\right\},
\end{aligned}   
\end{equation}
where from \eqref{bound200} and the fact that the error between $\widehat{\mbx}$ and the original signal remains deterministic with respect to each iteration, we can write
\begin{equation}
\begin{aligned}
\mathbb{E}\left\{\left\|\mbx_i-\mbx\right\|_{2}\right\} \leq \left(1-\frac{1}{\kappa^{2}\left(\mbA\right)}\right)^{\frac{i}{2}} \left\|\mbx_{0}-\widehat{\mbx}\right\|_{2} + \left\|\mbx-\widehat{\mbx}\right\|_{2}.
\end{aligned}
\end{equation}
The convergence to $\mbx$ is ensured only when the second term, $\left\|\mbx-\widehat{\mbx}\right\|_{2}$, is bounded. As demonstrated in Theorem~\ref{theorem_0}, with a minimum probability of $1-2 e^{-\frac{c \epsilon^2 m^{\prime}}{2\left(K+\frac{\lambda}{\sqrt{\log(2)}}\right)^2}}$, we have $\left\|\mbx-\widehat{\mbx}\right\|_{2} \leq \rho$, which proves the proposition. 
\end{IEEEproof}
In the following lemma, we establish a convergence rate that is uniform for all RIP matrices in the context of ORKA: 
\begin{lemma}
\label{universal}  
Under the assumptions of Proposition~\ref{penaltyyy} and if the sampling matrix meets the $\delta$-subspace embedding property for all $\mbx\in\mathbb{R}^{d}$, the uniform convergence rate for ORKA is formulated in a probability as:
\begin{equation}
\label{bound20}
\begin{aligned}
\mathbb{E}\left\{\left\|\mbx_i-\mbx\right\|_2\right\} \leq \left(1-\frac{1}{n}\left(\frac{1-\delta}{1+\delta}\right)^2\right)^{\frac{i}{2}} \left\|\mbx_0-\widehat{\mbx}\right\|_2+\rho.
\end{aligned}
\end{equation}
\end{lemma}
\begin{IEEEproof}
We know that the scaled condition number follows $\kappa\left(\mbA\right)\leq \sqrt{n}\varrho\left(\mbA\right)$. Therefore, if the matrix has RIP, the condition number satisfies the following bound:
\begin{equation}
 \varrho\left(\mbA\right) \leq \frac{1+\delta}{1-\delta},  
\end{equation}
which completes the proposition. 
\end{IEEEproof}
In a Gaussian sampling matrix case, the uniform convergence rate of ORKA is stated as follows:
\begin{lemma}
The uniform convergence rate of ORKA for the Gaussian sampling matrix, in a probability, is given by:
\begin{equation}
\begin{aligned}
\mathbb{E}\left\{\left\|\mbx_i-\mbx\right\|_2\right\} \leq \left(1-\frac{\left(1-\delta\right)^2}{1.0049 d }\right)^{\frac{i}{2}} \left\|\mbx_0-\widehat{\mbx}\right\|_2+\rho.
\end{aligned}
\end{equation}
\end{lemma}
\begin{IEEEproof}
The Gaussian sampling matrix satisfies the $\delta$-subspace embedding property with a probability comprehensively discussed in \cite[Section~8.7]{martinsson2020randomized}, i.e., $\sigma_{\mathrm{min}}\left(\mbA\right)\geq \sqrt{n} \left(1-\delta\right)$. 
Also by concentration inequality of Chi-squared random variables with probability at least $1-e^{- c n d}$, $\left\|\mbA\right\|^2_{\mathrm{F}}\leq 1.0049 n d$, which proves the lemma.     
\end{IEEEproof}
Note that if any other randomized algorithm is utilized for one-bit sensing instead of RKA to achieve uniform reconstruction, the convergence rate will maintain the same structure as Proposition~\ref{penaltyyy}. However, there will be a difference in the first term, which will be substituted by the algorithm's convergence rate to a point inside the feasible space of hyperplanes.
\subsection{One-Bit Low-Rank Matrix Sensing}
The convergence guarantee of SVP-ORKA is concluded according to the following lemma:
\begin{lemma}
\label{Ea}
Under assumptions of Theorem~\ref{theorem_1}, the update process of SVP-ORKA presented in \eqref{St_20} converges in expectation to
$\mathcal{B}_{\rho}\left(\operatorname{vec}\left(\mbX\right)\right)$, as follows:
\begin{equation}
\label{kvm_19}
\begin{aligned}
\mathbb{E}\left\{\left\|\mbX_{i+1}-\mbX\right\|_{\mathrm{F}}\right\}\leq\left(1-\frac{1}{\kappa^{2}\left(\mbV\right)}\right)^{\frac{i}{2}} \left\|\mbX_{0}-\widehat{\mbX}\right\|_{\mathrm{F}}+\rho,
\end{aligned}   
\end{equation}
where $\widehat{\mbX}\in \mathcal{P}^{(M)}_{1}$,
and $\mbV$ is the matrix with vectorized sensing matrices $\left\{\operatorname{vec}\left(\mbA_j\right)\right\}^{n}_{j=1}$ as its rows. 
\end{lemma}
The proof of Lemma~\eqref{Ea} is studied in Appendix~\ref{svp_proof}.
\subsection{One-Bit CS}
The convergence rate of ST-ORKA is studied in the following lemma:
\begin{lemma}
\label{kvm_24}
Under assumptions of Theorem~\ref{theorem_3}, consider a ball centered at the signal $\mathcal{B}_{\rho}\left(\mbx\right)$ and $\widehat{\mbx}\in\mathcal{P}^{(C)}_{1}$, the convergence of ST-ORKA presented in \eqref{St_21} is given by
\begin{equation}
\label{kvm_25}
\begin{aligned}
\mathbb{E}\left\{\left\|\mbx_{i+1}-\mbx\right\|_2\right\}\leq\left(1-\frac{1}{\kappa^{2}\left(\mbA\right)}\right)^{\frac{i}{2}} \left\|\mbx_{0}-\widehat{\mbx}\right\|_2+\rho.
\end{aligned}
\end{equation}
\end{lemma}
The proof of Lemma~\ref{kvm_24} is studied in Appendix~\ref{sta}.
Note that to have a guaranteed convergence, one can integrate ORKA with different operators that satisfy Lipschitz continuity. This is important in various applications and provides a chance to go beyond just ST-ORKA and SVP-ORKA.
By doing so, we can achieve a convergence rate that is described in the following lemma: 
\begin{lemma}
\label{Ea1}
Assume $f(\cdot)$ is an operator such that for any $\mbx_1,\mbx_2\in\mathbb{R}^{d}$ we have
$\left\|f\left(\mbx_1\right)-f\left(\mbx_2\right)\right\|_2^2 \leq L_f \left\|\mbx_1-\mbx_2\right\|_2^2$, with $L_f$ denotes the Lipschitz continuity constant. Then, the integration of RKA with the operator $f$ in each iteration of solving the feasibility problem $\mbC\mbx\succeq\mbb$ with $\mbC\in\mathbb{R}^{n\times d}$ has the following convergence rate:
\begin{equation}
\mathbb{E}\left\{\left\|\mbx_{i+1}-\widehat{\mbx}\right\|^2_2\right\}\leq L_f\left(1-\frac{1}{\kappa^{2}\left(\mbC\right)}\right)^{i} \left\|\mbx_{0}-\widehat{\mbx}\right\|^2_2.
\end{equation}
\end{lemma}
A notable example for Lemma~\ref{Ea1} is the integration of ORKA with the hard thresholding (HT) operator $\mathcal{T}_s(\cdot)$ to reconstruct the $s$-sparse signal in the one-bit CS problem.
This operator selects the best $s$-sparse approximation of the solution at each iteration. The algorithm is named HT-ORKA with the update process presented as
\begin{equation}
\label{kvm_30}
\left\{\begin{array}{l}
\mbz_{i+1}=\mathrm{KA}_r\left(\mbx_i\right),\\
\mbx_{i+1} = \mathcal{T}_s\left(\mbz_{i+1}\right).
\end{array}\right.
\end{equation}
The convergence rate of HT-ORKA can be determined using the following lemma:
\begin{lemma}
\label{ht_orka}
Under assumptions of Theorem~\ref{theorem_3}, assume $\widehat{\mbx}\in\mathcal{P}^{(C)}_{1}$, the convergence of HT-ORKA presented in \eqref{kvm_30} is given by
\begin{equation}
\begin{aligned}
\mathbb{E}\left\{\left\|\mbx_{i+1}-\mbx\right\|_2\right\}\leq 2\left(1-\frac{1}{\kappa^{2}\left(\mbA\right)}\right)^{\frac{i}{2}} \left\|\mbx_{0}-\widehat{\mbx}\right\|_2+\rho. 
\end{aligned}
\end{equation}
\end{lemma}
The proof of Lemma~\ref{ht_orka} is provided in Appendix~\ref{htt_orka}.

\section{Numerical Investigations}
In this section, we conduct numerical evaluations to assess the performance of our proposed algorithms in two distinct scenarios: (i) sample abundance and (ii) sample restriction. All presented results are averaged over 1000 experiments.

\subsection{Sample abundance}
In this particular scenario, we examine two examples: one-bit low-rank matrix sensing and one-bit CS. For both cases, we employ the Block SKM algorithm to recover the desired signal and then evaluate its performance against the adaptive threshold strategy introduced in \cite[Section~\rom{6}]{eamaz2022phase}. The only distinction is that, in this study, we did not update the one-bit data to account for an exceedingly low-complexity hardware implementation.
In all experiments, we have taken into account the presence of Gaussian additive noise, also known as \emph{prequantization error}, with a standard deviation of $0.1$.
\\
\textbf{One-bit low-rank matrix sensing.} We considered a set of sampling matrices $\{\mbA_j\}_{j=1}^{1800}$, where each entry is independently drawn from a standard normal distribution. We have generated the desired matrix $\mbX\in\mathbb{R}^{30\times 30}$ such that $\operatorname{rank}(\mbX)=2$. The number of time-varying sampling threshold sequences were set to $m\in\{1,10,20,30\}$. Accordingly, we have generated sequences of time-varying sampling thresholds as $\left\{\boldsymbol{\uptau}^{(\ell)}\sim \mathcal{U}_{[-\beta_{\mby},\beta_{\mby}]}\right\}_{\ell=1}^{m}$, where $\beta_{\mby}$ denotes the dynamic range of the high-resolution measurements $\mby$. Figure \ref{figure_2}(a) illustrates a comparison between the recovery performance of Block SKM using random thresholds and 
adaptive thresholds. It is evident that the utilization of adaptive thresholds enhances the recovery performance compared to random thresholds.
\\
\textbf{One-bit CS.} We have generated a sensing matrix $\mbA\in\mathbb{R}^{500\times 100}$ in which each element follows a standard normal distribution. The desired signal $\mbx\in\mathbb{R}^{100}$ was assumed to have the level of sparsity $s=10$. The settings for time-varying sampling thresholds were considered to be the same as the one-bit low-rank matrix sensing case. Fig.~\ref{figure_2}(b) displays the recovery performance of Block SKM using random thresholds in comparison with adaptive thresholds. Consistent with previous observations, the utilization of adaptive thresholds improves the recovery performance.
\begin{figure}[t]
	\centering
	\subfloat[]
		{\includegraphics[width=0.47\columnwidth]{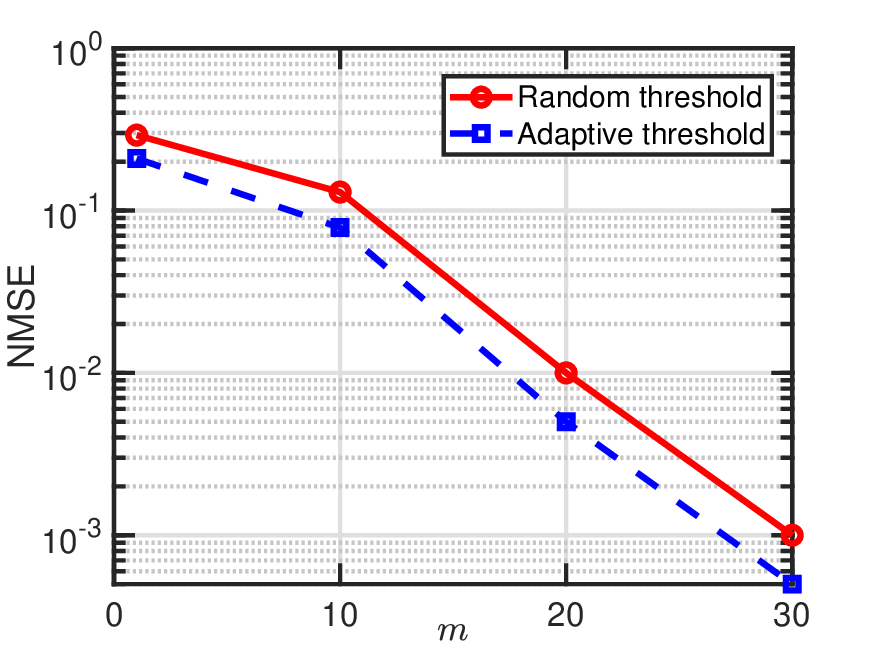}}
    \subfloat[]
		{\includegraphics[width=0.47\columnwidth]{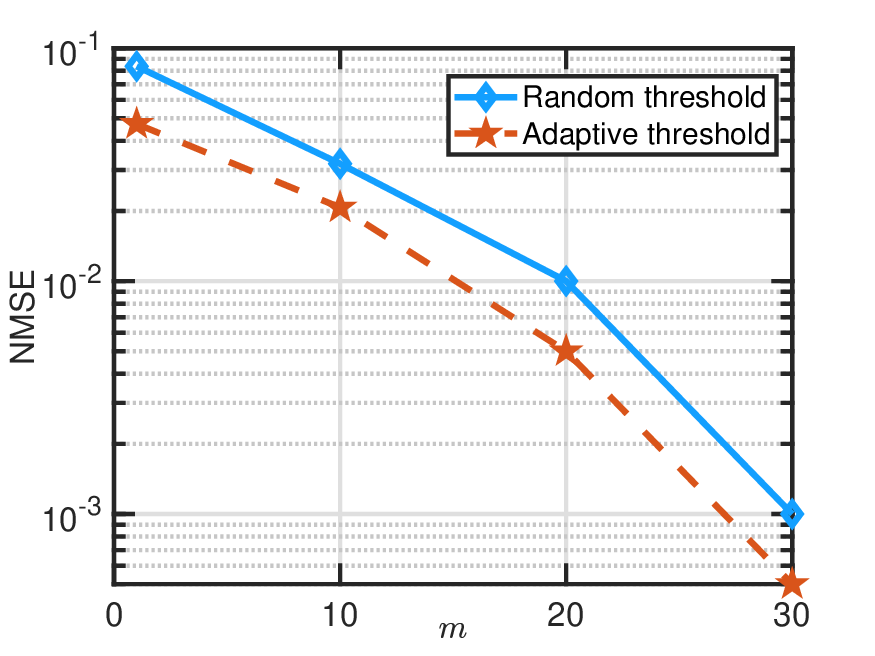}}
	\caption{Comparison between the recovery performance of Block SKM using random thresholds and adaptive thresholds in the sample abundance scenario for (a) one-bit low-rank matrix sensing, and (ii) one-bit CS.
  \vspace{-15pt}
 }
	\label{figure_2}
\end{figure}
\subsection{Sample Scarcity}
Similar to sample abundance, herein we investigate the performance of our proposed methods for one-bit low-rank matrix sensing and one-bit CS when we have a limited number of samples. Note that in all experiments, the high-resolution measurements were contaminated by the additive Gaussian noise with the standard deviation $0.1$ except the case related to low-rank matrix sensing by SVP-ORKA which was considered to be noiseless.
\\
\textbf{One-bit low-rank matrix sensing.} We generated a collection of sampling matrices $\{\mbA_j\}_{j=1}^{n}$, where each entry is independently sampled from a standard normal distribution. The desired matrix $\mbX\in\mathbb{R}^{30\times 30}$ was generated with $\operatorname{rank}(\mbX)=2$. Define the oversampling factor as $\lambda=\frac{n}{n_1r}=\frac{n}{60}$. In our experiments, we have set 
$\log(\lambda)\in\{3,4,5,6\}$. The number of time-varying sampling threshold sequences was fixed at $m=1$. The generation of time-varying sampling thresholds followed the same procedure as in the previous cases. Fig.~\ref{figure_3}(a) compares the recovery performance of SVP-ORKA with hard singular value thresholding (HSVT) algorithm \cite{foucart2019recovering} in the noiseless scenario. As can be observed, SVP-ORKA outperforms HSVT over different values of the oversampling factor. In another experimental setting aimed at assessing the performance of Algorithm~\ref{algorithm_200} (ORKA with low-rank matrix factorization), we generated the desired matrix $\mbX\in\mathbb{R}^{30\times 30}$ with $\operatorname{rank}(\mbX)=1$. The remaining parameter settings were identical to the previous example. Note that in this case we define the oversampling factor as $\beta=\frac{n}{n_1^2r}=\frac{n}{900}$. In our simulations, we have set $\beta=\{5,10,15,20\}$. As can be seen in Fig.~\ref{figure_3}(b), the recovery performance of Algorithm~\ref{algorithm_200} enhances as the value of the oversampling factor $\beta$ grows large.
\begin{figure}[t]
	\centering
	\subfloat[]
		{\includegraphics[width=0.47\columnwidth]{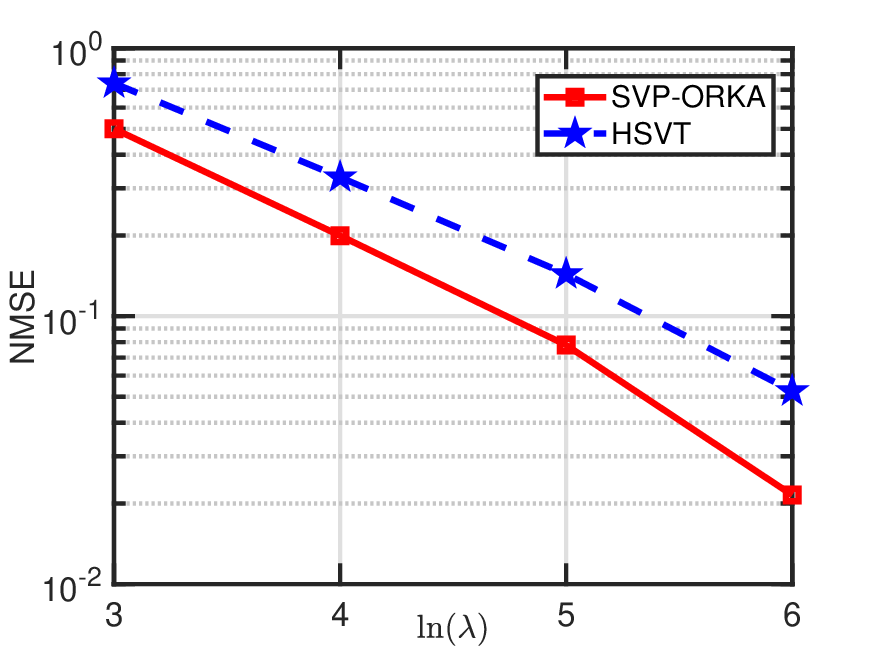}}
    \subfloat[]
		{\includegraphics[width=0.47\columnwidth]{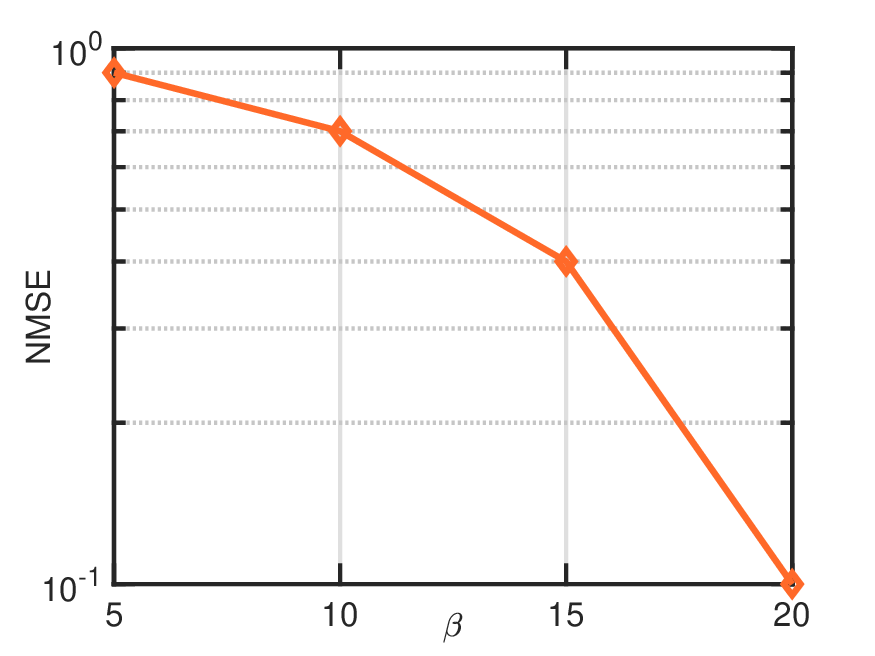}}
	\caption{(a) Comparison between the recovery performance of SVP-ORKA and HSVT algorithm over different values of oversampling factor $\lambda$. (b) Recovery performance of Algorithm~\ref{algorithm_200} (ORKA with low-rank matrix factorization) over different values of sampling factor $\beta$.
  \vspace{-15pt}
 }
	\label{figure_3}
\end{figure}
\begin{figure}[t]
	\centering
	\subfloat[]
		{\includegraphics[width=0.47\columnwidth]{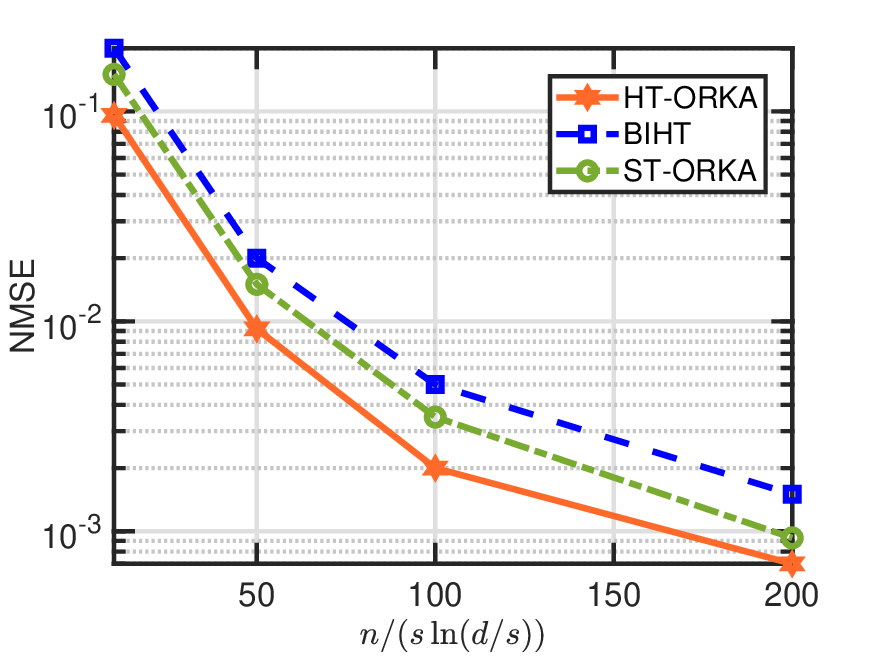}}
    \subfloat[]
		{\includegraphics[width=0.47\columnwidth]{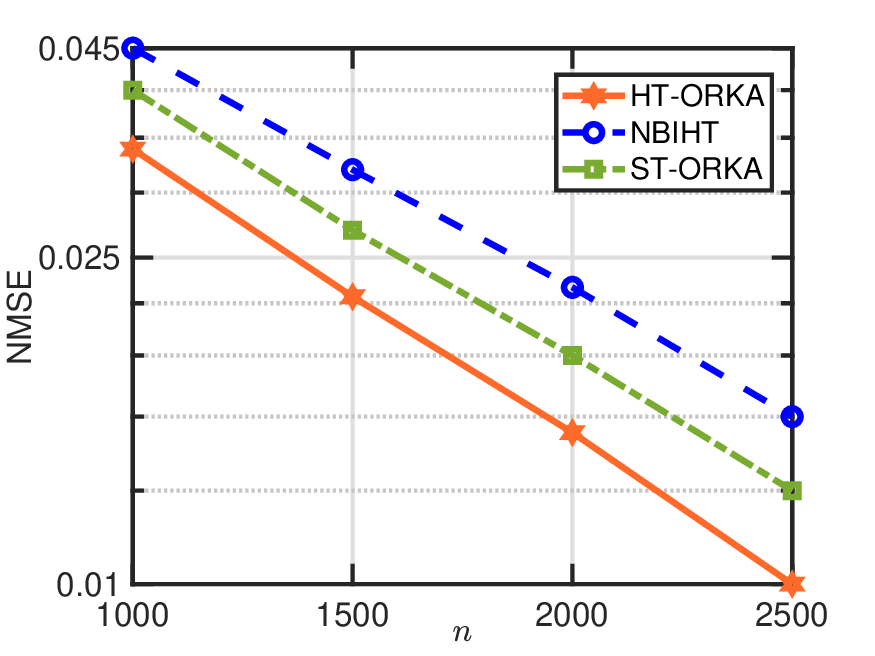}}
	\caption{Comparison between the recovery performance of HT-ORKA, ST-ORKA and (a) BIHT with random thresholds, and (b) NBIHT in ditherless scenario.
  \vspace{-15pt}
 }
	\label{figure_4}
\end{figure}
\\
\textbf{One-bit CS.} Each element of the sensing matrix $\mbA\in\mathbb{R}^{n\times 100}$ was independently drawn from a standard normal distribution. The desired signal $\mbx\in\mathbb{R}^{100}$ was assumed to have a sparsity level of $s=15$. Define the oversampling factor as $n/(s\log(d/s))$. In our experiments, we have set the oversampling factor to the values $\{10,50,100,200\}$. The number of time-varying sampling threshold sequences was fixed at $m=1$, and the generation of time-varying sampling thresholds followed the same procedure as in the previous cases. In Fig.~\ref{figure_4}(a), we compare the recovery performance of HT-ORKA, ST-ORKA, and BIHT algorithm with random time-varying sampling thresholds \cite{baraniuk2017exponential}. It is evident that HT-ORKA outperforms both ST-ORKA and BIHT in recovering the $s$-sparse signal in the one-bit CS problem. In the next example, we examine the effectiveness of our proposed algorithms, HT-ORKA and ST-ORKA, in recovering a sparse signal in a ditherless scenario. The sensing matrix $\mbA\in\mathbb{R}^{n\times 256}$ was generated in the same manner as in the previous example. The number of high-resolution samples was considered to be $n\in\{1000,1500,2000,2500\}$. The desired signal $\mbx\in\mathbb{R}^{256}$ was assumed to have a sparsity level of $s=25$. In Fig.~\ref{figure_4}(b), we compare the recovery performance of HT-ORKA, ST-ORKA and the NBIHT algorithm \cite{friedlander2021nbiht}. Once again, similar to the previous example, HT-ORKA exhibits superior recovery performance compared to ST-ORKA and the NBIHT method.

\section{Discussion}
In this paper, we have established the theoretical guarantees for uniform reconstruction in dithered one-bit sensing. Our approach involves transforming the one-bit signal reconstruction problem into a linear feasibility problem. We then introduced the FVP theorem to analyze the possibility of creating a finite volume formed by the hyperplanes around the original signal.
The FVP theorem allows us to determine the required number of samples to find a solution inside a cell centered at the signal. What sets this theorem apart from others, such as the random hyperplane tessellations theorem, is that it approaches one-bit sensing from the perspective of a linear feasibility problem. It investigates the number of samples needed to capture the original signal within the space of hyperplanes.
An intriguing aspect of the FVP theorem is its capability to provide guarantees not only for restricted sampling matrices considered in previous efforts but deterministic matrices like DCT as well. The only restriction in the FVP theorem is that the sampling matrix must be isotropic. 
Additionally, utilizing the FVP, we established theoretical guarantees for dithered one-bit sensing even when the scalar parameter cannot dominate the DR of measurements.

This work represents a first effort in the literature, as it explores the performance of \emph{randomized algorithms} in one-bit sensing. We introduced two novel variations of the Kaczmarz algorithm, PrSKM and Block SKM, which served as the foundation for our proposed algorithm, ORKA.
In our investigation of ORKA, we analyzed its upper recovery bound and demonstrated that it decays concerning the number of measurements. These findings contribute valuable insights into the potential of randomized algorithms in one-bit sensing applications. 
To the best of our knowledge, we are the first to derive the convergence rate of RKA for a noisy linear \emph{inequality} system. This novel finding highlights the robustness of the algorithm, even in the presence of noise.

We have introduced an improved update process for designing the thresholding process. Unlike the sigma-delta thresholding design discussed in references \cite{baraniuk2017exponential,foucart2019recovering}, our approach does not require updating the one-bit data, eliminating the need for feedback in the sampling scheme. Through numerical demonstrations, we showed that this adaptive thresholding process enhances signal reconstruction performance from one-bit data.
Through numerical experiments, we experimentally demonstrated that the proposed algorithm exhibits superior reconstruction performance in one-bit CS compared to the NBIHT for the ditherless scenario and BIHT adapted with dithering for the dithered scenario. Furthermore, our results show that the proposed SVP-ORKA outperforms the HSVT algorithm in terms of recovery accuracy. 

In addition to the common situation of sample abundance in dithered one-bit quantization, we also addressed scenarios with sample restrictions. For these scenarios, we further developed the proposed randomized algorithms into storage-friendly approaches, such as random sketching and low-rank matrix factorization. This extension allows for efficient handling of limited samples, broadening the applicability of the algorithms to a wider range of practical settings. The convergence rate of ORKA with low-rank matrix factorization remains an open problem. This is because it employs the idea of cyclic algorithms or AltMin, whose convergence is still an ongoing topic of research in the literature. 


\appendices

\section{Proof of Theorem~\ref{stprskm}}
\label{ak}
Since $\mbN\in\mathbb{R}^{m\times s}$ is a Gaussian matrix with sketch size $s=\mathcal{O}\left(n/\epsilon^{2}\right)$ and $\epsilon=\frac{1}{2}$, the $\epsilon$-subspace embedding property holds with high probability for all $\mbx\in\mathbb{R}^{n}$ \cite[Theorem~8.4]{martinsson2020randomized},
\begin{equation}
\label{stprskm2}
\left(1-\epsilon\right)\left\|\mbC\mbx\right\|_{2}^{2}\leq\|\mbN^{\top}\mbC\mbx\|_{2}^{2}\leq (1+\epsilon)\|\mbC\mbx\|_{2}^{2}.
\end{equation}
We can use this result to show that the preconditioned system $\mbC\mbR_{s}^{-1}$ preserves the length of vectors in the range of $\mbC$. To see how, substitute $\mbx$ in \eqref{stprskm2} with $\mbR_{s}^{-1}\mbz$. By the subspace embedding property,
\begin{equation}
\label{stprskm3}
\frac{1}{1+\epsilon}\|\mbN^{\top}\mbC\mbR_{s}^{-1}\mbz\|_{2}^{2}\leq\|\mbC\mbR_{s}^{-1}\mbz\|_{2}^{2}\leq\frac{1}{1-\epsilon}\|\mbN^{\top}\mbC\mbR_{s}^{-1}\mbz\|_{2}^{2}.
\end{equation}
Now using $\|\mbN^{\top}\mbC\mbR_{s}^{-1}\mbz\|_{2}^{2}=\|\mbQ_{s}\mbz\|_{2}^{2}=\|\mbz\|_{2}^{2}$, the relation \eqref{stprskm3} becomes
\begin{equation}
\label{stprskm4}
\frac{1}{1+\epsilon}\|\mbz\|_{2}^{2}\leq\|\mbC\mbR_{s}^{-1}\mbz\|_{2}^{2}\leq\frac{1}{1-\epsilon}\|\mbz\|_{2}^{2},
\end{equation}
which bounds the condition number $\varrho\left(\mbC\mbR_{s}^{-1}\right)\leq\sqrt{\frac{1+\epsilon}{1-\epsilon}}$. Setting $\epsilon=\frac{1}{2}$, we have $\varrho\left(\mbC\mbR_{s}^{-1}\right)\leq\sqrt{3}$ and based on \eqref{gorang}, $\kappa\left(\mbC\mbR_{s}^{-1}\right)\leq \sqrt{3n}$ . Therefore, according to \eqref{eq:15}
and the result that we have obtained here, one can conclude the convergence bound \eqref{stprskm1}.

\section{Convergence analysis of Block SKM with the sparse Gaussian sketch \eqref{St_3}}
\label{A}
The Block SKM algorithm can be considered to be a special case of the more general \emph{sketch-and-project} method with a sparse block sketch matrix as defined in \cite{derezinski2022sharp}. 
\begin{equation}
\label{neg_2}
\mbx_{i+1}= \underset{\mbx}{\textrm{argmin}}~\left\|\mbx-\mbx_{i}\right\|^{2}_{2}\quad \textrm{subject to}\quad \mbS^{\top}\mbB\mbx\succeq\mbS^{\top}\mbb,
\end{equation}
where $\mbS\in\mathbb{R}^{m n\times n}$ is the sketch matrix choosing a block uniformly at random from the main matrix $\mbB$ similar to step $1$ of Algorithm~\ref{algorithm_2}. The second step of sketch-and-project method follows the Motzkin sampling where the index $j_{\star}$ is chosen in $i$-th iteration as 
\begin{equation}
j_{\star}=\underset{j}{\textrm{argmax}}\left\{ \left(\left(\mbS^{\top}\mbb\right)_{j}-\left(\mbS^{\top}\mbB\right)_{j}\mbx_{i}\right)^{+}\right\},\end{equation}
with $(\cdot)_{j}$ denoting the $j$-th row of the matrix/vector argument (this step is similar to steps $2-4$ of Algorithm~\ref{algorithm_2} with $k^{\prime}=1$).
In the Block SKM algorithm, the sketch matrix is given by
\begin{equation}
\label{St_2}
\mbS=\left[\begin{array}{c|c|c}
\mathbf{0}_{n\times p} &\mbI_{n} &\mathbf{0}_{n\times (m n-n-p)}
\end{array}\right]^{\top}, ~\mbS\in\left\{0,1\right\}^{m n\times n},
\end{equation}
where 
$p=n\alpha,~\alpha\in\left\{1,\cdots,m-1\right\}$. Note that the literature does not offer any theoretical guarantees for the convergence of the Block SKM with the identity matrix~\cite{rebrova2021block}. To derive our theoretical guarantees for the Block SKM algorithm,
we change the sketch matrix to the sparse \emph{Gaussian} sketch matrix as follows:
\begin{equation}
\label{St_3}
\mbS=\left[\begin{array}{c|c|c}
\mathbf{0}_{n\times p} &\mbG &\mathbf{0}_{n\times (m n-n-p)}
\end{array}\right]^{\top}, ~ \mbS\in\mathbb{R}^{m n\times n},
\end{equation}
where $\mbG$ is a $n\times n$ Gaussian matrix, whose entries are i.i.d. following the distribution $\mathcal{N}\left(0,1\right)$. In this framework, we are able to provide some theoretical guarantees by taking advantage of the favorable properties of Gaussian random variables.
Assume that $\mathcal{P}_{\mbx}$ denotes a nonempty solution set of $\mbB\mbx\succeq\mbb$ and $\widehat{\mbx}\in\mathcal{P}_{\mbx}$.
Considering the sparse Gaussian sketch \eqref{St_3} which satisfies the set of inequalities $\mbS^{\top}\mbB\mbx\succeq\mbS^{\top}\mbb$, we have,
\par\noindent\small
\begin{equation}
\label{St_200}
\begin{aligned}
\left\|\mbx_{i+1}-\widehat{\mbx}\right\|_2^2&=\left\|\left(\mbx_i-\widehat{\mbx}\right)+\frac{\left(\left(\mbS^{\top} \mbb\right)_{j_{\star}}-\left(\mbS^{\top}\mbB\right)_{j_{\star}}\mbx_{i}\right)^{+} \left(\mbS^{\top}\mbB\right)^{\mathrm{H}}_{j_{\star}}}{\left\|\left(\mbS^{\top}\mbB\right)_{j_{\star}}\right\|^2_2}\right\|^2_2\\
&=\left\|\mbx_i-\widehat{\mbx}\right\|^2_2+\frac{\left(\left(\left(\mbS^{\top} \mbb\right)_{j_{\star}}-\left(\mbS^{\top}\mbB\right)_{j_{\star}}\mbx_{i}\right)^{+}\right)^2}{\left\|\left(\mbS^{\top}\mbB\right)_{j_{\star}}\right\|^2_2}+\frac{2\left(\left(\mbS^{\top} \mbb\right)_{j_{\star}}-\left(\mbS^{\top}\mbB\right)_{j_{\star}}\mbx_{i}\right)^{+}\left(\mbS^{\top}\mbB\right)_{j_{\star}}}{\left\|\left(\mbS^{\top}\mbB\right)_{j_{\star}}\right\|^2_2}\left(\mbx_i-\widehat{\mbx}\right).
\end{aligned}
\end{equation} \normalsize
We can observe that $\left(\mbS^{\top}\mbB\right)_{j_{\star}}\left(\mbx_i-\widehat{\mbx}\right)\leq\left(\mbS^{\top}\mbB\right)_{j_{\star}}\mbx_i-\left(\mbS^{\top} \mbb\right)_{j_{\star}}$, therefore, the left-hand side of \eqref{St_200} is less than or equal to,
\begin{equation}
\label{obi_one}
\begin{aligned}
&\left\|\mbx_i-\widehat{\mbx}\right\|^2_2+\frac{\left(\left(\left(\mbS^{\top} \mbb\right)_{j_{\star}}-\left(\mbS^{\top}\mbB\right)_{j_{\star}}\mbx_{i}\right)^{+}\right)^2}{\left\|\left(\mbS^{\top}\mbB\right)_{j_{\star}}\right\|^2_2}+\frac{2\left(\left(\mbS^{\top} \mbb\right)_{j_{\star}}-\left(\mbS^{\top}\mbB\right)_{j_{\star}}\mbx_{i}\right)^{+}\left(\left(\mbS^{\top}\mbB\right)_{j_{\star}}\mbx_{i}-\left(\mbS^{\top} \mbb\right)_{j_{\star}}\right)}{\left\|\left(\mbS^{\top}\mbB\right)_{j_{\star}}\right\|^2_2}=\\&\left\|\mbx_i-\widehat{\mbx}\right\|^2_2-\frac{\left(\left(\left(\mbS^{\top} \mbb\right)_{j_{\star}}-\left(\mbS^{\top}\mbB\right)_{j_{\star}}\mbx_{i}\right)^{+}\right)^2}{\left\|\left(\mbS^{\top}\mbB\right)_{j_{\star}}\right\|^2_2}.
\end{aligned}
\end{equation}
Based on the definition of $j_{i}^{\star}$, one can rewrite \eqref{obi_one} as,
\begin{equation}
\label{St_100}
\begin{aligned}
\left\|\mbx_{i+1}-\widehat{\mbx}\right\|_2^2\leq\left\|\mbx_i-\widehat{\mbx}\right\|_2^2-\frac{\left\|\left(\mbS^{\top} \mbb-\mbS^{\top}\mbB\mbx_i\right)^{+} \right\|_{\infty}^2}{\left\|\left(\mbS^{\top}\mbB\right)_{j_{\star}}\right\|_2^2}.
\end{aligned}
\end{equation}
By taking the expectation over the error, we have
\begin{equation}
\label{St_110}
\begin{aligned}
\mathbb{E}_{\mbS}\left\{\left\|\mbx_{i+1}-\widehat{\mbx}\right\|_2^2\right\}\leq\left\|\mbx_i-\widehat{\mbx}\right\|_2^2-\mathbb{E}_{\mbS} \left\{\frac{\left\|\left(\mbS^{\top} \mbb-\mbS^{\top}\mbB\mbx_i\right)^{+} \right\|_{\infty}^2}{\left\|\left(\mbS^{\top}\mbB\right)_{j_{\star}}\right\|_2^2}\right\}.
\end{aligned}
\end{equation} 
In addition, we have that
\begin{equation}
\label{St_1020}
\begin{aligned}
\mathbb{E}_{\mbS}\left\{\left\|\left(\mbS^{\top} \mbB\right)_{j_{\star}}\right\|_2^2\right\}=\sum_{k=1}^{d} \mathbb{E}_{\mbS}\left\{\left(\sum_{l=1}^{m n} \mbS_{j_{i}^{\star}l}\mbB_{l k}\right)^2\right\},
\end{aligned}
\end{equation}
or equivalently, in terms of $\mbG$ in \eqref{St_3},
\begin{equation}
\label{St_120}
\begin{aligned}
\sum_{k=1}^{d} \mathbb{E}_{\mbG}\left\{\left(\sum_{l=1}^{n} \mbG^{\top}_{j_{i}^{\star}l} \mbB_{l k}\right)^2\right\}=
\sum_{k=1}^{d} \sum_{l=1}^{n} \mathbb{E}_{\mbG}\left\{\left(\mbG^{\top}_{j_{i}^{\star} l}\right)^{2}\right\} \mbB^{2}_{l k},
\end{aligned}
\end{equation}
with $\mathbb{E}_{\mbG}\left\{\left(\mbG^{\top}_{j_{i}^{\star} l}\right)^{2}\right\}=1$, which helps to simplify \eqref{St_120} as
\begin{equation}
\label{St_1200}
\begin{aligned}
\sum_{k=1}^{d}\sum_{l=1}^{n} \mbB^{2}_{l k }=\|\widehat{\mbB}\|_{\mathrm{F}}^2,
\end{aligned}
\end{equation}
where $\widehat{\mbB}$ is the $n\times d$ submatrix of $\mbB$ (one of the candidates of $\mbB_{j}$ in Algorithm~\ref{algorithm_2}). Due to the fact that the second term in the right-hand side of \eqref{St_110} is an expectation over the convex function $f(x,y)=x^{2}/y$, we can apply Jensen's inequality as follows: 
\begin{equation}
\label{St_1300}
\begin{aligned}
\mathbb{E}_{\mbS}\left\{ \frac{\left\|\left(\mbS^{\top} \mbb-\mbS^{\top}\mbB\mbx_i\right)^{+} \right\|_{\infty}^2}{\left\|\left(\mbS^{\top} \mbB\right)_{j_{\star}}\right\|_2^2}\right\}\geq\frac{\left(\mathbb{E}_{\mbS}\left\{\left\|\left(\mbS^{\top} \mbb-\mbS^{\top}\mbB\mbx_i\right)^{+} \right\|_{\infty}\right\}\right)^2}{\mathbb{E}_{\mbS}\left\{\left\|\left(\mbS^{\top} \mbB\right)_{j_{\star}}\right\|_2^2\right\}}.
\end{aligned}
\end{equation} 
Consider the following lemma regarding the estimate 
of the maximum of independent normal random variables:
\begin{lemma}\cite[Section~2.5.2]{vershynin2018high}
\label{ghoorak}
Let $X_1,\cdots,X_n$ be independent $\mathcal{N}(0,1)$ random variables. Then we have
\begin{equation}
\label{B1_4}
\mathbb{E}\left\{\max_{i\leq n}X_i\right\}\geq c\sqrt{\log n},
\end{equation}
where $c$ is an absolute constant.
\end{lemma}
By taking advantage of Lemma~\ref{ghoorak}, we have
\begin{equation}
\label{St_1400}
\begin{aligned}
\mathbb{E}_{\mbS}\left\{\left\|\left(\mbS^{\top} \mbb-\mbS^{\top}\mbB\mbx_i\right)^{+}\right\|_{\infty}\right\}&=\mathbb{E}_{\mbS}\left\{\max_{t\in[n]}\left(\left\langle \mbs_t, \mbb-\mbB\mbx_{i}\right\rangle\right)^{+}\right\}\\&\geq
\mathbb{E}_{\mbS}\left\{\max_{t\in[n]}\left\langle \mbs_t, \left(\mbb-\mbB\mbx_{i}\right)^{+}\right\rangle\right\}\\&=
\mathbb{E}_{\mbG}\left\{\max_{t\in[n]}\left\langle \mbg_t, \left(\widehat{\mbb}-\widehat{\mbB}\mbx_{i}\right)^{+}\right\rangle\right\}\\&\geq c\left\|\left(\widehat{\mbb}-\widehat{\mbB}\mbx_{i}\right)^{+}\right\|_2 \sqrt{\log n},
\end{aligned}
\end{equation}
where $\widehat{\mbb}\in\mathbb{R}^{n}$ is a block of $\mbb$, $\mbs_{t}$ and $\mbg_{t}$ are the $t$-th columns of $\mbS$ and $\mbG$, respectively, and $c$ is a positive value.
By plugging the inequality \eqref{St_1400} into \eqref{St_110} and using the Hoffman bound \cite[Theorem~4.2]{leventhal2010randomized}, we have
\begin{equation}
\label{St_1600}
\begin{aligned}
\mathbb{E}\left\{\left\|\mbx_{i+1}-\widehat{\mbx}\right\|_2^2\right\}&\leq\left\|\mbx_i-\widehat{\mbx}\right\|_2^2-\frac{c\|(\widehat{\mbb}-\widehat{\mbB}\mbx_{i})^{+}\|_2^2 \log n}{\|\widehat{\mbB}\|_{\mathrm{F}}^2}\\
\quad &\leq\left\|\mbx_i-\widehat{\mbx}\right\|_2^2-\frac{c \sigma_{\min }^2(\widehat{\mbB}) \log n}{\|\widehat{\mbB}\|_{\mathrm{F}}^2}\left\|\mbx_i-\widehat{\mbx}\right\|_2^2\\
\quad &\leq \left(1-\frac{c \sigma_{\min }^2(\widehat{\mbB}) \log n}{\|\widehat{\mbB}\|_{\mathrm{F}}^2}\right)\left\|\mbx_i-\widehat{\mbx}\right\|_2^2,
\end{aligned}
\end{equation}
which can be recast as the following \emph{convergence rate}, after $(i+1)$ updates:
\begin{equation}
\label{St_170}
\begin{aligned}
\mathbb{E}\left\{\left\|\mbx_{i+1}-\widehat{\mbx}\right\|_2^2\right\}\leq \left(1-\frac{c \sigma_{\min }^2(\widehat{\mbB}) \log n}{\|\widehat{\mbB}\|_{\mathrm{F}}^2}\right)^{i+1}\left\|\mbx_0-\widehat{\mbx}\right\|_2^2.
\end{aligned}
\end{equation}

\section{Proof of Proposition~\ref{prop_noisyRKA}}
\label{A5}
We begin the proof by presenting the following lemma:
\begin{lemma}
\label{A5_1}
Let $\mathcal{H}_{j}=\left\{\mbx:\langle\mbc_{j},\mbx\rangle\geq b_{j}\right\}$ be the solution spaces of the unperturbed linear inequalities. Let $\bar{\mathcal{H}}_{j}=\left\{\mbx:\langle\mbc_{j},\mbx\rangle\geq b_{j}-n_{j}\right\}$ be the solution spaces of the noisy linear inequalities run by the noisy RKA. Then $\bar{\mathcal{H}}_{j}=\left\{\mbx+\alpha_{j}\mbc_{j}:\mbx\in\mathcal{H}_{j}\right\}$, where $\alpha_{j}=\frac{-n_{j}}{\left\|\mbc_{j}\right\|_{2}^{2}}$.
\end{lemma}
\begin{IEEEproof}
Assume $\mbx\in\mathcal{H}_{j}$, then we can write
\begin{equation}
\label{A5_2}
\langle\mbc_{j},\mbx+\alpha_{j}\mbc_{j}\rangle\geq b_{j}+\alpha_{j}\left\|\mbc_{j}\right\|_{2}^{2}.
\end{equation}
By plugging $\alpha_{j}=\frac{-n_{j}}{\left\|\mbc_{j}\right\|_{2}^{2}}$ in \eqref{A5_2}, we can obtain
\begin{equation}
\label{A5_3}
\langle\mbc_{j},\mbx+\alpha_{j}\mbc_{j}\rangle\geq b_{j}-n_j,
\end{equation}
which implies $\mbx+\alpha_{j}\mbc_{j}\in\bar{\mathcal{H}}_{j}$. Next let $\mbu\in\bar{\mathcal{H}}_{j}$. Set $\mbx=\mbu-\alpha_j\mbc_j$. Then we can write
\begin{equation}
\label{A5new}
\langle\mbc_j,\mbx\rangle=\langle\mbc_j,\mbu-\alpha_j\mbc_j\rangle\geq b_j-n_j-\alpha_j\|\mbc_j\|_2^2=b_j,
\end{equation}
which implies $\mbx\in\mathcal{H}_j$.
\end{IEEEproof}
Assume $\bar{\mbx}_{i}$ denotes the $i$-th iterate of the noisy RKA run with $\mbC\mbx+\mbn\succeq\mbb$. Denote $\widehat{\mbx}\in\cap_{j=1}^{m}\mathcal{H}_j$. Based on Lemma~\ref{A5_1}, we can write
\begin{equation}
\label{A5_4}
\bar{\mbx}_{i}-\widehat{\mbx}=\mbx_{i}-\widehat{\mbx}+\alpha_{j}\mbc_{j},
\end{equation}
where $\mbx_{i}\in\mathcal{H}_j$. Taking the norm-$2$ of \eqref{A5_4} leads to
\begin{equation}
\label{A5_5}
\begin{aligned}
\left\|\bar{\mbx}_{i}-\widehat{\mbx}\right\|_{2}&=\left\|\mbx_{i}-\widehat{\mbx}+\alpha_{j}\mbc_{j}\right\|_{2}\\&\leq\left\|\mbx_{i}-\widehat{\mbx}\right\|_{2}+\left\|\alpha_{j}\mbc_{j}\right\|_{2}\\&=\left\|\mbx_{i}-\widehat{\mbx}\right\|_{2}+\frac{|n_{j}|}{\left\|\mbc_{j}\right\|_{2}}.
\end{aligned}
\end{equation}
Denote the scaled condition number of the matrix $\mbC$ by $\kappa(\mbC)=\sqrt{R}$. Then by considering the convergence rate of the RKA in \eqref{eq:15} and drawing inspiration from \cite[Lemma~2.2]{needell2010randomized} in the context of linear system of inequalities, we have
\begin{equation}
\label{A5_6}
\begin{aligned}
\mathbb{E}\left\{\left\|\bar{\mbx}_{i}-\widehat{\mbx}\right\|_2\right\}\leq\left(1-\frac{1}{R}\right)^{\frac{1}{2}}\left\|\bar{\mbx}_{i-1}-\widehat{\mbx}\right\|_{2}+\frac{|n_{j}|}{\left\|\mbc_{j}\right\|_{2}}\leq\left(1-\frac{1}{R}\right)^{\frac{1}{2}}\left\|\bar{\mbx}_{i-1}-\widehat{\mbx}\right\|_{2}+\max_{j}\gamma_{j},
\end{aligned}
\end{equation}
where $\gamma_{j}$ is defined in Proposition~\ref{prop_noisyRKA}, and the expectation is conditioned upon the choice of the random selections in the first $i-1$ iterations. Then applying this recursive relation iteratively and taking full expectation, we have
\begin{equation}
\label{A5_7}
\begin{aligned}
\mathbb{E}\left\{\left\|\bar{\mbx}_{i}-\widehat{\mbx}\right\|_{2}\right\}&\leq\left(1-\frac{1}{R}\right)^{\frac{i}{2}}\left\|\mbx_{0}-\widehat{\mbx}\right\|_{2}\\&+\sum_{k=0}^{i}\left(1-\frac{1}{R}\right)^{\frac{k}{2}}\max_{j}\gamma_{j}\\&\leq\left(1-\frac{1}{R}\right)^{\frac{i}{2}}\left\|\mbx_{0}-\widehat{\mbx}\right\|_{2}+\sqrt{R}\max_{j}\gamma_{j},
\end{aligned}
\end{equation}
which completes the proof.

\section{Proof of Theorem~\ref{theorem_0}}
\label{ealmaz}
Denote $d_j^{(\ell)}$ in \eqref{dist_1} by $d_j^{(\ell)}=|f_j-\tau_{j,\ell}|$, where $\tau_{j,\ell}\overset{\mathrm{iid}}{\sim}\mathcal{U}_{[-\lambda,\lambda]}$ and $f_j=\langle\mba_j,\mbx\rangle$. Assume $\lambda$ is chosen such that
\begin{equation}
\label{z_1}
\lambda\geq\sup_{j\in[n]}|f_j|.
\end{equation}
Then, we can write 
\begin{equation}
\label{z_2}
\begin{aligned}
\mathbb{E}_{\tau}\left\{d_j^{(\ell)}\right\}&=\frac{1}{2\lambda}\int_{-\lambda}^{\lambda}\left|f_j-\tau_{j,\ell}\right|\,d\tau_{j,\ell}\\&=\frac{1}{2\lambda}\left[\int_{-\lambda}^{f_j}(f_j-\tau_{j,\ell})\,d\tau_{j,\ell}+\int_{f_j}^{\lambda}(\tau_{j,\ell}-f_j)\,d\tau_{j,\ell}\right]\\&=\frac{1}{2\lambda}\left(\lambda^2+f_j^2\right).
\end{aligned}
\end{equation} 
Since we have assumed that the sensing matrix $\mbA$ is isotropic according to Definition~\ref{def_3}, a simple calculation shows
\begin{equation}
\label{z_3}
\mathbb{E}\left\{d_j^{(\ell)}\right\}=\frac{1}{2\lambda}\left(\lambda^2+\|\mbx\|_2^2\right).
\end{equation}
In the following lemma, we present the general Hoeffding's inequality:
\begin{lemma}\cite[Theroem~2.6.2]{vershynin2018high}
\label{hoff}
Let $X_1,\cdots,X_N$ be independent, mean zero, sub-gaussian random variables. Then, for every $t\geq 0$, we have
\begin{equation}
\label{z_4}
\operatorname{Pr}\left(\left|\sum_{i=1}^{N}X_i\right|\geq t\right)\leq2e^{-\frac{ct^2}{\sum_{i=1}^{N}\|X_i\|_{\psi_2}^2}},
\end{equation}
where $c$ is a positive constant.
\end{lemma}
Since we have assumed $\tau_{j,\ell}\overset{\mathrm{iid}}{\sim}\mathcal{U}_{[-\lambda,\lambda]}$, and the sensing matrix $\mbA$ is sub-gaussian, then each random variable $d_j^{(\ell)}$ is sub-gaussian with the following bound on its sub-gaussian norm:
\begin{equation}
\label{z_5}
\begin{aligned}
\left\|d_j^{(\ell)}\right\|_{\psi_2}\leq\|\mba_j\|_{\psi_2}\|\mbx\|_2+\|\tau_{j,\ell}\|_{\psi_2}\leq K+\frac{\lambda}{\sqrt{\log(2)}}.
\end{aligned}
\end{equation}
Combining \eqref{z_3} and \eqref{z_5} with Lemma~\ref{hoff} leads to
\begin{equation}
\label{z_6}
\operatorname{Pr}\left(\left|T_{\mathrm{ave}}\left(\mbx\right)-\frac{\lambda}{2}-\frac{\left\|\mbx\right\|_{2}^2}{2\lambda}\right|\geq\epsilon\right)\leq 2 e^{-\frac{c \epsilon^2 m^{\prime}}{\left(K+\frac{\lambda}{\sqrt{\log(2)}}\right)^2}}.
\end{equation}
As we consider the supremum over all $\mbx\in\mathcal{K}$, it is necessary to multiply the resulting probability by the minimum number of $\rho$-net completely covering the set $\mathcal{K}$ whose logarithm is the same as Kolmogorov $\rho$-entropy of $\mathcal{K}$. Therefore, we have 
\begin{equation}
\label{z_7}
\begin{aligned}
\operatorname{Pr}\left(\sup_{\mbx\in\mathcal{K}}\left|T_{\mathrm{ave}}\left(\mbx\right)-\frac{\lambda}{2}-\frac{\left\|\mbx\right\|_{2}^2}{2\lambda}\right|\geq\epsilon\right)\leq2 e^{-\frac{c \epsilon^2 m^{\prime}}{\left(K+\frac{\lambda}{\sqrt{\log(2)}}\right)^2}+\mathcal{H}(\mathcal{K},\rho)}.
\end{aligned}
\end{equation} 
According to Sudakov's inequality, we can upper bound the Kolmogorov $\rho$-entropy as
\begin{equation}
\label{z_8}
\mathcal{H}(\mathcal{K},\rho)\lesssim\frac{\gamma^2\left(\mathcal{K}\right)}{\rho^2}.
\end{equation}
To achieve the failure probability at most $2 e^{-\frac{c \epsilon^2 m^{\prime}}{2\left(K+\frac{\lambda}{\sqrt{\log(2)}}\right)^2}}$, it is sufficient to write
\begin{equation}
\label{z_9}
\frac{\gamma^2\left(\mathcal{K}\right)}{\rho^2}\lesssim \frac{c \epsilon^2 m^{\prime}}{2\left(K+\frac{\lambda}{\sqrt{\log(2)}}\right)^2},
\end{equation}
which completes the proof.

\section{Proof of Theorem~\ref{theorem_4}}
\label{dct_theorem}
Rewrite the linear model \eqref{dct_1} in the matrix form as $\mby=\mbA\mbx$, where $\mbA$ denotes the DCT matrix. Assume the parameter of uniform dithering $\lambda$ is chosen similar to \eqref{z_1}. According to \eqref{z_2}, we obtain
\begin{equation}
\label{y_1}
\mathbb{E}_{\tau}\left\{d_j^{(\ell)}\right\}=\frac{\lambda}{2}+\frac{|\langle\mba_j,\mbx\rangle|^2}{2\lambda}.
\end{equation}
Based on our model \eqref{dct_1}, a simple calculation shows that $\mathbb{E}\left\{\mba_j\mba_j^{\mathrm{H}}\right\}=\frac{1}{2}\mbI$ for all $j\in[n]$, where the expected value is taken over the randomness of the frequencies $\omega_k$. This result implies that
\begin{equation}
\label{y_2}
\mathbb{E}\left\{d_j^{(\ell)}\right\}=\frac{\lambda}{2}+\frac{\|\mbx\|_2^2}{4\lambda}.
\end{equation}
In the following lemma, we present the Hoeffding's inequality for bounded random variables:
\begin{lemma}\cite[Theroem~2.2.5]{vershynin2018high}
\label{lem_2}
Let $\left\{X_i\right\}^{n}_{i=1}$ be independent, bounded random variables
satisfying $X_i \in [a_i, b_i]$, then for any $t > 0$ it holds that 
\begin{equation}
\operatorname{Pr}\left(\left|\frac{1}{n} \sum_{i=1}^n\left(X_i-\mathbb{E}\left\{ X_i\right\}\right)\right| \geq t\right)\leq 2 e^{-\frac{2 n^2 t^2}{\sum_{i=1}^n\left(b_i-a_i\right)^2}}.  
\end{equation}
\end{lemma}
Notice that for each random variable $d_j^{(\ell)}$, we have $d_j^{(\ell)}\geq 0$ and
\begin{equation}
\label{y_3}
d_j^{(\ell)}=\left|\langle\mba_j,\mbx\rangle-\tau_j^{(\ell)}\right|\leq\left|\langle\mba_j,\mbx\rangle\right|+\left|\tau_j^{(\ell)}\right|\leq2\lambda,
\end{equation}
where in the last inequality, we have utilized the assumption in \eqref{z_1}. Following Lemma~\ref{lem_2}, we can write
\begin{equation}
\label{y_4}
\operatorname{Pr}\left(\left|T_{\mathrm{ave}}\left(\mbx\right)-\frac{\lambda}{2}-\frac{\left\|\mbx\right\|_{2}^2}{4\lambda}\right|\geq\epsilon\right)\leq 2 e^{-\frac{\epsilon^2m^{\prime}}{2\lambda^2}}.
\end{equation}
Similar to the proof of Theorem~\ref{theorem_0}, by considering the supremum over all $\mbx\in\mathcal{K}$, we have
\begin{equation}
\label{y_5}
\operatorname{Pr}\left(\sup_{\mbx\in\mathcal{K}}\left|T_{\mathrm{ave}}\left(\mbx\right)-\frac{\lambda}{2}-\frac{\left\|\mbx\right\|_{2}^2}{4\lambda}\right|\geq\epsilon\right)\leq2 e^{-\frac{\epsilon^2m^{\prime}}{2\lambda^2}+\mathcal{H}(\mathcal{K},\rho)}.
\end{equation}
Based on the Sudakov's inequality, setting $\mathcal{H}(\mathcal{K},\rho)\lesssim\frac{\gamma^2\left(\mathcal{K}\right)}{\rho^2}\lesssim \frac{\epsilon^2m^{\prime}}{4\lambda^2}$ completes the proof.

\section{Proof of Theorem~\ref{theorem_5}}
\label{Dr_guarantee}
Rewrite the average of distances $T_{\mathrm{ave}}$ in \eqref{a_7} as
\begin{equation}
\label{x_2}
T_{\mathrm{ave}}(\mbx)=\frac{1}{mn}\sum_{\ell=1}^{m}\sum_{j\in\mathcal{I}}\left|\langle\mba_j,\mbx\rangle-\tau_j^{(\ell)}\right|+\frac{1}{mn}\sum_{\ell=1}^{m}\left|\langle\mba_{j^{\prime}},\mbx\rangle-\tau_{j^{\prime}}^{(\ell)}\right|.
\end{equation} 
Under the assumption in \eqref{spider}, we can utilize the result obtained in \eqref{z_3} and write
\begin{equation}
\label{x_3}
\mathbb{E}\left\{\left|\langle\mba_j,\mbx\rangle-\tau_j^{(\ell)}\right|\right\}=\frac{\lambda}{2}+\frac{\|\mbx\|_2^2}{2\lambda},
\end{equation}
for all $j\in\mathcal{I}$. For simplicity of notation, let $\tau_{j^{\prime}}^{(\ell)}=\tau_{j^{\prime},\ell}$ and $f_{j^{\prime}}=\langle\mba_{j^{\prime}},\mbx\rangle$. Define $d_{j^{\prime}}^{(\ell)}=\left|f_{j^{\prime}}-\tau_{j^{\prime},\ell}\right|$. Under the assumption in \eqref{spider}, we can write
\begin{equation}
\label{x_4}
\begin{aligned}
\mathbb{E}_{\tau}\left\{d_{j^{\prime}}^{(\ell)}\right\}&=\frac{1}{2\lambda}\int_{-\lambda}^{\lambda}\left|f_{j^{\prime}}-\tau_{j^{\prime},\ell}\right|\,d\tau_{j^{\prime},\ell}\\&=\frac{1}{2\lambda}\int_{-\lambda}^{\lambda}\left(f_{j^{\prime}}-\tau_{j^{\prime},\ell}\right)\,d\tau_{j^{\prime},\ell}=f_{j^{\prime}},
\end{aligned}
\end{equation}
where taking the expectation with respect to the randomness of $\mba_{j^{\prime}}$ leads to $\mathbb{E}\left\{d_{j^{\prime}}^{(\ell)}\right\}=\mu^{\prime}$. As a result, we have
\begin{equation}
\label{x_5}
\begin{aligned}
\mathbb{E}\left\{T_{\mathrm{ave}}(\mbx)\right\}&=\frac{|\mathcal{I}|}{n}\left(\frac{\lambda}{2}+\frac{\|\mbx\|_2^2}{2\lambda}\right)+\frac{\mu^{\prime}}{n}\\&=\frac{n-1}{n}\left(\frac{\lambda}{2}+\frac{\|\mbx\|_2^2}{2\lambda}\right)+\frac{\mu^{\prime}}{n}.
\end{aligned}
\end{equation}
The rest of the proof follows directly from the proof of Theorem~\ref{theorem_0}. The only difference lies in the failure probability, now given by $2 e^{-\frac{c \epsilon^2 m^{\prime}}{2\left(K^{\prime}+\frac{\lambda}{\sqrt{\log(2)}}\right)^2}}$ instead of $2 e^{-\frac{c \epsilon^2 m^{\prime}}{2\left(K+\frac{\lambda}{\sqrt{\log(2)}}\right)^2}}$, since we have assumed $K^{\prime}>K$.

\section{Proof of Proposition~\ref{error1}}
\label{er_1}
Define $\mbz=\frac{1}{2}\left(\mbx+\bar{\mbx}\right)$. Then, we can write
\begin{equation}
\label{w_1}
\langle\mba_j,\mbz\rangle-\tau_j^{(\ell)}=\frac{1}{2}\left(\langle\mba_j,\mbx\rangle-\tau_j^{(\ell)}+\langle\mba_j,\bar{\mbx}\rangle-\tau_j^{(\ell)}\right).
\end{equation}
Following the consistent reconstruction property in Definition~\ref{def_2}, we have
\begin{equation}
\label{w_2}
\left|\langle\mba_j,\mbz\rangle-\tau_j^{(\ell)}\right|=\frac{1}{2}\left(\left|\langle\mba_j,\mbx\rangle-\tau_j^{(\ell)}\right|+\left|\langle\mba_j,\bar{\mbx}\rangle-\tau_j^{(\ell)}\right|\right).
\end{equation}
Taking average over all $j\in[n],\ell\in[m]$ leads to
\begin{equation}
\label{w_3}
T_{\mathrm{ave}}(\mbz)=\frac{1}{2}\left(T_{\mathrm{ave}}(\mbx)+T_{\mathrm{ave}}(\bar{\mbx})\right).
\end{equation}
Based on Theorem~\ref{theorem_0}, if $m^{\prime}\gtrsim \frac{\gamma^2\left(\mathcal{K}\right)}{\rho^2 \epsilon^2}$ is met, then with a failure probability at most $2 e^{-\frac{c \epsilon^2 m^{\prime}}{2\left(K+\frac{\lambda}{\sqrt{\log(2)}}\right)^2}}$, we have
\begin{equation}
\label{w_4}
\frac{\|\mbz\|_2^2}{2\lambda}\geq T_{\mathrm{ave}}(\mbz)-\frac{\lambda}{2}-\epsilon.
\end{equation}
Combining the results of \eqref{w_3} and \eqref{w_4} leads to
\begin{equation}
\label{w_5}
\begin{aligned}
\frac{\|\mbz\|_2^2}{2\lambda}&\geq \frac{1}{2}\left(T_{\mathrm{ave}}(\mbx)+T_{\mathrm{ave}}(\bar{\mbx})\right)-\frac{\lambda}{2}-\epsilon\\&\geq\frac{1}{2}\left(\frac{\lambda}{2}+\frac{\|\mbx\|_2^2}{2\lambda}-\epsilon+\frac{\lambda}{2}+\frac{\|\bar{\mbx}\|_2^2}{2\lambda}-\epsilon\right)-\frac{\lambda}{2}-\epsilon\\&=\frac{1}{4\lambda}\left(\|\mbx\|_2^2+\|\bar{\mbx}\|_2^2\right)-2\epsilon.
\end{aligned}
\end{equation}
Based on the definition of $\mbz$, we can rewrite \eqref{w_5} in terms of $\mbx$ and $\bar{\mbx}$ as follows
\begin{equation}
\label{w_6}
\|\mbx+\bar{\mbx}\|_2^2\geq2\left(\|\mbx\|_2^2+\|\bar{\mbx}\|_2^2\right)-16\epsilon\lambda.
\end{equation}
By the parallelogram law, we conclude that
\begin{equation}
\label{w_7}
\|\mbx-\bar{\mbx}\|_2^2=2\left(\|\mbx\|_2^2+\|\bar{\mbx}\|_2^2\right)-\|\mbx+\bar{\mbx}\|_2^2,
\end{equation}
which together with \eqref{w_6} results in
\begin{equation}
\label{w_8}
\|\mbx-\bar{\mbx}\|_2\leq4\sqrt{\epsilon\lambda}.
\end{equation}

\section{Proof of Proposition~\ref{error2}}
\label{er_2}
Define $\mbz=\frac{1}{2}\left(\mbx+\bar{\mbx}\right)$. Then, we can write
\begin{equation}
\label{v_1}
\langle\mba_j,\mbz\rangle-\tau_j^{(\ell)}=\frac{1}{2}\left(\langle\mba_j,\mbx\rangle-\tau_j^{(\ell)}+\langle\mba_j,\bar{\mbx}\rangle-\tau_j^{(\ell)}\right).
\end{equation}
Define the set $\mathcal{G}$ such that
\begin{equation}
\label{v_2}
\operatorname{sgn}\left(\langle\mba_j,\mbx\rangle-\tau_j^{(\ell)}\right)\neq\operatorname{sgn}\left(\langle\mba_j,\bar{\mbx}\rangle-\tau_j^{(\ell)}\right),\quad (j,\ell)\in\mathcal{G}.
\end{equation} 
For all $(j,\ell)\in[n]\times[m]\setminus\mathcal{G}$, we have
\begin{equation}
\label{v_3}
\left|\langle\mba_j,\mbz\rangle-\tau_j^{(\ell)}\right|=\frac{1}{2}\left(\left|\langle\mba_j,\mbx\rangle-\tau_j^{(\ell)}\right|+\left|\langle\mba_j,\bar{\mbx}\rangle-\tau_j^{(\ell)}\right|\right),
\end{equation}
due to the consistent reconstruction property. It can be simply verified that for all $(j,\ell)\in\mathcal{G}$, we have
\begin{equation}
\label{v_4}
\begin{aligned}
\left|\langle\mba_j,\mbz\rangle-\tau_j^{(\ell)}\right|=\frac{1}{2}\left(\left|\langle\mba_j,\mbx\rangle-\tau_j^{(\ell)}\right|+\left|\langle\mba_j,\bar{\mbx}\rangle-\tau_j^{(\ell)}\right|\right)-\min\left(\left|\langle\mba_j,\mbx\rangle-\tau_j^{(\ell)}\right|,\left|\langle\mba_j,\bar{\mbx}\rangle-\tau_j^{(\ell)}\right|\right).
\end{aligned}
\end{equation}
Taking average over all $j\in[n],\ell\in[m]$ leads to
\begin{equation}
\label{v_5}
T_{\mathrm{ave}}(\mbz)=\frac{1}{2}\left(T_{\mathrm{ave}}(\mbx)+T_{\mathrm{ave}}(\bar{\mbx})\right)-R,
\end{equation}
where
\begin{equation}
\label{v_6}
R=\frac{1}{mn}\sum_{(j,\ell)\in\mathcal{G}}\beta_{j,\ell},
\end{equation}
with $\beta_{j,\ell}=\min\left(\left|\langle\mba_j,\mbx\rangle-\tau_j^{(\ell)}\right|,\left|\langle\mba_j,\bar{\mbx}\rangle-\tau_j^{(\ell)}\right|\right)$.
Based on Theorem~\ref{theorem_0}, if $m^{\prime}\gtrsim \frac{\gamma^2\left(\mathcal{K}\right)}{\rho^2 \epsilon^2}$ is met, then with a failure probability at most $2 e^{-\frac{c \epsilon^2 m^{\prime}}{2\left(K+\frac{\lambda}{\sqrt{\log(2)}}\right)^2}}$, we have
\begin{equation}
\label{v_7}
\frac{\|\mbz\|_2^2}{2\lambda}\geq T_{\mathrm{ave}}(\mbz)-\frac{\lambda}{2}-\epsilon,
\end{equation}
which together with \eqref{v_5} results in
\begin{equation}
\label{v_8}
\begin{aligned}
\frac{\|\mbz\|_2^2}{2\lambda}&\geq \frac{1}{2}\left(T_{\mathrm{ave}}(\mbx)+T_{\mathrm{ave}}(\bar{\mbx})\right)-R-\frac{\lambda}{2}-\epsilon\\&\geq\frac{1}{2}\left(\frac{\lambda}{2}+\frac{\|\mbx\|_2^2}{2\lambda}-\epsilon+\frac{\lambda}{2}+\frac{\|\bar{\mbx}\|_2^2}{2\lambda}-\epsilon\right)-R-\frac{\lambda}{2}-\epsilon\\&=\frac{1}{4\lambda}\left(\|\mbx\|_2^2+\|\bar{\mbx}\|_2^2\right)-R-2\epsilon.
\end{aligned}
\end{equation}
Based on the definition of $\mbz$ and by the parallelogram law, we can write
\begin{equation}
\label{v_9}
\|\mbx-\bar{\mbx}\|_2^2\leq8\lambda R+16\epsilon\lambda.
\end{equation}
According to Theorem~\ref{theorem_0}, with a failure probability at most $2 e^{-\frac{c \epsilon^2 m^{\prime}}{2\left(K+\frac{\lambda}{\sqrt{\log(2)}}\right)^2}}$, the value of $R$ is bounded as
\begin{equation}
\label{v_10}
R\leq\frac{|\mathcal{G}|}{mn}\left(\frac{\lambda}{2}+\frac{1}{2\lambda}\right),
\end{equation}
which together with \eqref{v_9} completes the proof.

\section{Proof of Proposition~\ref{error3}}
\label{er_3}
Define $\mbz=\frac{1}{2}\left(\mbx+\bar{\mbx}\right)$. Similar to the proof of Proposition~\ref{error1}, due to the consistency assumption, we can write
\begin{equation}
\label{u_1}
T_{\mathrm{ave}}(\mbz)=\frac{1}{2}\left(T_{\mathrm{ave}}(\mbx)+T_{\mathrm{ave}}(\bar{\mbx})\right).
\end{equation}
Denote $\mby=\mbA\mbx$, $\bar{\mby}=\mbA\bar{\mbx}$, and $\mby_z=\mbA\mbz$. With the assumption in \eqref{spider}, we have $\mathrm{DR}_{\mby},\mathrm{DR}_{\bar{\mby}}\geq\lambda$ which leads to $\mathrm{DR}_{\mby_z}\geq\lambda$. Based on Theorem~\ref{theorem_5}, if $m^{\prime}\gtrsim \frac{\gamma^2\left(\mathcal{K}\right)}{\rho^2 \epsilon^2}$ is met, then with a failure probability at most $2 e^{-\frac{c \epsilon^2 m^{\prime}}{2\left(K^{\prime}+\frac{\lambda}{\sqrt{\log(2)}}\right)^2}}$, we have
\begin{equation}
\label{u_2}
\frac{(n-1)}{n}\frac{\|\mbz\|_2^2}{2\lambda}\geq T_{\mathrm{ave}}(\mbz)-\frac{(n-1)}{n}\frac{\lambda}{2}-\frac{1}{n}\mu^{\prime}_{z}-\epsilon,
\end{equation}
where $\mu^{\prime}_z=\mathbb{E}\left\{\langle\mba_{j^{\prime}},\mbz\rangle\right\}$. Combining the results of \eqref{u_1} and \eqref{u_2} leads to
\begin{equation}
\label{u_3}
\begin{aligned}
\frac{(n-1)}{n}\frac{\|\mbz\|_2^2}{2\lambda}&\geq\frac{1}{2}\left(T_{\mathrm{ave}}(\mbx)+T_{\mathrm{ave}}(\bar{\mbx})\right)-\frac{(n-1)}{n}\frac{\lambda}{2}-\frac{1}{n}\mu^{\prime}_{z}-\epsilon\\&\geq\frac{1}{2}\left(\frac{(n-1)}{n}\lambda+\frac{(n-1)}{2\lambda n}\left(\|\mbx\|_2^2+\|\bar{\mbx}\|_2^2\right)+\frac{1}{n}\left(\mu^{\prime}_{x}+\mu^{\prime}_{\bar{x}}\right)-2\epsilon\right)-\frac{(n-1)}{n}\frac{\lambda}{2}-\frac{1}{n}\mu^{\prime}_{z}-\epsilon,
\end{aligned}
\end{equation} 
where $\mu^{\prime}_x=\mathbb{E}\left\{\langle\mba_{j^{\prime}},\mbx\rangle\right\}$ and $\mu^{\prime}_{\bar{x}}=\mathbb{E}\left\{\langle\mba_{j^{\prime}},\bar{\mbx}\rangle\right\}$. Based on the definition of $\mbz$, we can write $\mu^{\prime}_z=\frac{1}{2}\left(\mu^{\prime}_x+\mu^{\prime}_{\bar{x}}\right)$. Combining this result with \eqref{u_3} leads to
\begin{equation}
\label{u_4}
\|\mbx+\bar{\mbx}\|_2^2\geq2\left(\|\mbx\|_2^2+\|\bar{\mbx}\|_2^2\right)-16\epsilon\lambda\left(\frac{n}{n-1}\right),
\end{equation}
where the parallelogram law implies
\begin{equation}
\label{u_5}
\|\mbx-\bar{\mbx}\|_2\leq4\sqrt{\epsilon\lambda\left(\frac{n}{n-1}\right)}.
\end{equation}
\section{Proof of Theorem~\ref{theorem_2}}
\label{akk}
Assume that $\left\{\sigma_{i\mbP}\right\}$ and $\left\{\sigma_{i\mbA}\right\}$ denote the singular values of the matrices $\mbP$ and $\mbA$, respectively.
To obtain the singular values of $\mbP$, the matrix $\mbW = \mbP^{\top}\mbP$ is computed as
\begin{equation}
\label{proof_th}
\begin{aligned}
\mbW &= \mbP^{\top}\mbP\\&=  \left[\begin{array}{c|c|c}
\mbA^{\top}\bOmega^{(1)} &\cdots &\mbA^{\top}\bOmega^{(m)}
\end{array}\right]\left[\begin{array}{c|c|c}
\mbA^{\top}\bOmega^{(1)} &\cdots &\mbA^{\top}\bOmega^{(m)}
\end{array}\right]^{\top}\\
&=\mbA^{\top}\bOmega^{(1)}\bOmega^{(1)}\mbA+\cdots+\mbA^{\top}\bOmega^{(m)}\bOmega^{(m)}\mbA\\
&= m\mbA^{\top}\mbI\mbA=m\mbA^{\top}\mbA,
\end{aligned}
\end{equation} 
which means that the singular values of $\mbP$ are $\left\{\sigma_{i\mbP}\right\}=\sqrt{m}\left\{\sigma_{i\mbA}\right\}$. Also, the Frobenius norm of  
$\mbP$ is obtained as
\begin{equation}
\label{frob}
\begin{aligned}
\|\mbP\|^{2}_{\mathrm{F}}=\operatorname{Tr}\left(\mbP^{\top}\mbP\right)= \operatorname{Tr}\left(m\mbA^{\top}\mbA\right) = m\|\mbA\|^{2}_{\mathrm{F}}.
\end{aligned}
\end{equation}
Plugging in \eqref{proof_th} and \eqref{frob} in the definition of the scaled condition number provided in Section~\ref{sec_RKA}, one can conclude that $\kappa\left(\mbP\right)=\kappa\left(\mbA\right)$.

\section{Convergence analysis of ST-ORKA, SVP-ORKA, and HT-ORKA}
\label{MY-ST}
In this section, we 
derive the convergence rates of ST-ORKA, SVP-ORKA, and HT-ORKA. 
\subsection{Convergence Analysis for ST-ORKA}
\label{sta}
As discussed in \cite{escalante2011alternating}, if $\Omega_c$ is a closed and convex set in any Hilbert space $\mathcal{H}$, then the projection operator
$P: \mathcal{H} \rightarrow \Omega_c$ is non-expansive, i.e., for any vectors $\mbx_1, \mbx_2 \in \mathcal{H}$, $\left\|P(\mbx_1)-P(\mbx_2)\right\|_2\leq\left\|\mbx_1-\mbx_2\right\|_2$. Due to the fact that the ST operator projects the signal to $\left\|\mbx\right\|_1\leq \epsilon$ which is a closed and convex set,
we can write that for any vectors $\mbx_1,\mbx_2\in\mathcal{H}$, $\left\|S_{\upkappa}(\mbx_1)-S_{\upkappa}(\mbx_2)\right\|_2\leq\left\|\mbx_1-\mbx_2\right\|_2$.

Define $\mbe_i=\mbx_i-\widehat{\mbx}$, $t^{(\ell)}_{j}=r^{(\ell)}_j\tau^{(\ell)}_j$ and $\mbp^{(\ell)}_j=r^{(\ell)}_j\mba_j$, we derive the convergence rate of ST-ORKA as
\begin{equation}
\begin{aligned}
\label{kvm_31}
\left\|\mbe_{i+1}\right\|^2_2&=\left\|S_{\upkappa}\left(\mbz_{i+1}\right)-S_{\upkappa}\left(\widehat{\mbx}\right)\right\|^2_2\\ &\leq\left\|\mbz_{i+1}-\widehat{\mbx}\right\|^2_2\\
&=\left\|\mbe_i+\frac{\left(t^{(\ell)}_j-\langle\mbp^{(\ell)}_j,\mbx_{i}\rangle\right)^+}{\left\|\mbp^{(\ell)}_j\right\|^2_2}  \mbp^{(\ell)}_j\right\|^2_2\\
&= \left\|\mbe_i\right\|^2_2+\frac{\left(\left(t^{(\ell)}_j-\langle\mbp^{(\ell)}_j,\mbx_{i}\rangle\right)^+\right)^2}{\left\|\mbp^{(\ell)}_j\right\|^2_2}+\frac{2\left(t^{(\ell)}_j-\langle\mbp^{(\ell)}_j,\mbx_{i}\rangle\right)^+\langle\mbp^{(\ell)}_j,\mbe_i\rangle}{\left\|\mbp^{(\ell)}_j\right\|^2_2}.
\end{aligned}
\end{equation} 
Since $\langle\mbp^{(\ell)}_j,\widehat{\mbx}\rangle\geq t^{(\ell)}_j$, we have
$\langle\mbp^{(\ell)}_j,\mbe_i\rangle=\langle\mbp^{(\ell)}_j,\mbx_i-\widehat{\mbx}\rangle\leq\langle\mbp^{(\ell)}_j,\mbx_i\rangle-t^{(\ell)}_j$. Therefore, one can rewrite \eqref{kvm_31} as
\begin{equation}
\begin{aligned}
\label{kvm_32}
\left\|\mbe_{i+1}\right\|^2_2&\leq\left\|\mbe_i\right\|^2_2+\frac{\left(\left(t^{(\ell)}_j-\langle\mbp^{(\ell)}_j,\mbx_{i}\rangle\right)^+\right)^2}{\left\|\mbp^{(\ell)}_j\right\|^2_2}+\frac{2\left(t^{(\ell)}_j-\langle\mbp^{(\ell)}_j,\mbx_{i}\rangle\right)^+\langle\mbp^{(\ell)}_j,\mbe_i\rangle}{\left\|\mbp^{(\ell)}_j\right\|^2_2}\\
&\leq\left\|\mbe_i\right\|^2_2+\frac{\left(\left(t^{(\ell)}_j-\langle\mbp^{(\ell)}_j,\mbx_{i}\rangle\right)^+\right)^2}{\left\|\mbp^{(\ell)}_j\right\|^2_2}-\frac{2\left(t^{(\ell)}_j-\langle\mbp^{(\ell)}_j,\mbx_{i}\rangle\right)^+\left(t^{(\ell)}_j-\langle\mbp^{(\ell)}_j,\mbx_{i}\rangle\right)}{\left\|\mbp^{(\ell)}_j\right\|^2_2}\\
&=\left\|\mbe_i\right\|^2_2-\frac{\left(\left(t^{(\ell)}_j-\langle\mbp^{(\ell)}_j,\mbx_{i}\rangle\right)^+\right)^2}{\left\|\mbp^{(\ell)}_j\right\|^2_2}.
\end{aligned}
\end{equation}
Define $\mbt=\left[\begin{array}{c|c|c}
\mbt_1^{\top} &\cdots &\mbt_m^{\top}
\end{array}\right]^{\top}$, where $\mbt_{\ell}=\left[t_j^{(\ell)}\right]_{j=1}^{n}$ for $\ell\in[m]$. Taking the expectation from both sides of \eqref{kvm_32} results in
\begin{equation}
\begin{aligned}
\label{kvm_33}
\mathbb{E}\left\{\left\|\mbe_{i+1}\right\|^2_2\right\}&\leq \left\|\mbe_i\right\|^2_2-\sum_{j,\ell=1}^{m^{\prime}}\frac{\left\|\mbp^{(\ell)}_j\right\|^2_2}{\left\|\mbP\right\|^2_{\mathrm{F}}}\frac{\left(\left(t^{(\ell)}_j-\langle\mbp^{(\ell)}_j,\mbx_{i}\rangle\right)^+\right)^2}{\left\|\mbp^{(\ell)}_j\right\|^2_2}\\
&= \left\|\mbe_i\right\|^2_2-\frac{\left\|\left(\mbt-\mbP\mbx_{i}\right)^+\right\|^2_2}{\left\|\mbP\right\|^2_{\mathrm{F}}}.
\end{aligned}
\end{equation} 
Based on the Hoffman bound \cite[Theorem~4.2]{leventhal2010randomized}, we have $\left\|\mbe_i\right\|^2_2\leq L_1\left\|\left(\mbt-\mbP\mbx_{i}\right)^+\right\|^2_2$, where $L_1=1/\sigma^2_{\text{min}}\left(\mbP\right)$. Combining this result with \eqref{kvm_33} leads to
\begin{equation}
\begin{aligned}
\mathbb{E}\left\{\left\|\mbe_{i+1}\right\|^2_2\right\} &\leq \left(1-\frac{\sigma^2_{\text{min}}\left(\mbP\right)}{\left\|\mbP\right\|^2_{\mathrm{F}}}\right) \left\|\mbe_i\right\|^2_2\\
&= \left(1-\frac{1}{\kappa^{2}\left(\mbP\right)}\right) \left\|\mbe_i\right\|^2_2.
\end{aligned}
\end{equation}
According to Theorem~\ref{theorem_2}, we know $\kappa\left(\mbA\right)=\kappa\left(\mbP\right)$. Thus, after $i$ iterations we have
\begin{equation}
\label{Conv}
\mathbb{E}\left\{\left\|\mathbf{e}_{i+1}\right\|^2_2\right\}\leq\left(1-\frac{1}{\kappa^{2}\left(\mbA\right)}\right)^{i} \left\|\mbe_0\right\|^2_2.
\end{equation}
It is evident that the rest of the proof follows the same logic as the proof of Proposition~\ref{penaltyyy}.
Note that the convergence rate \eqref{Conv} can be derived for any operator follows $P:\mathcal{H}\rightarrow\Omega_c$. This underscores the versatility of the algorithm and its ability to accommodate various types of operators.

\subsection{Convergence Analysis for SVP-ORKA}
\label{svp_proof}
To prove the convergence of SVP-ORKA, 
consider an operator function $\mathcal{G}_f$ applied to a matrix $\mbX$ with rank $r^{\prime}$ as follows:
\begin{equation}
\label{kvm_34}
\mathcal{G}_f\left(\mbX\right)=\sum^{r^{\prime}}_{k=1} f\left(\sigma_k\right)\mbu_k\mbv^{\mathrm{H}}_k,
\end{equation}
where $\left\{\sigma_k, \mbu_k, \mbv_k\right\}^{r^{\prime}}_{k=1}$ are singular values of $\mbX$ and its corresponding singular 
vectors, and $f$ is a $L$-Lipschitz continuous projector function. 
As comprehensively discussed \cite[Theorem~4.2]{andersson2016operator}, the following relation holds for two matrices $\mbX_1$ and $\mbX_2$ belonging to the Hilbert space $\mathcal{H}$:
\begin{equation}
\label{Lip1}
\left\|\mathcal{G}_f\left(\mbX_1\right)-\mathcal{G}_f\left(\mbX_2\right)\right\|_{\mathrm{F}} \leq L \left\|\mbX_1-\mbX_2\right\|_{\mathrm{F}}.
\end{equation}
In 
SVP-ORKA, $f$ 
is an operator which only chooses the $r$-largest singular values. 
It is straightforward to verify that such $f$ satisfies \eqref{Lip1} with $L=1$. Therefore, one can conclude
\begin{equation}
\label{Lip2}
\left\|P_r\left(\mbX_1\right)-P_r\left(\mbX_2\right)\right\|_{\mathrm{F}} \leq \left\|\mbX_1-\mbX_2\right\|_{\mathrm{F}},~\forall\mbX_1, \mbX_2 \in \mathcal{H}.
\end{equation}
Since $\left\|\mbX_1-\mbX_2\right\|_{\mathrm{F}}=\left\|\operatorname{vec}\left(\mbX_1\right)-\operatorname{vec}\left(\mbX_2\right)\right\|_2$, the convergence proof of SVP-ORKA is identical to that of ST-ORKA from this point forward. It is worth noting that using the SVT operator instead of SVP can lead to the same convergence proof for finding the solution in $\mathcal{P}^{(M)}_1$. This is because in 
SVT operator, the function $f(\cdot)$ is $f(\mbx) = \left(\mbx-\boldsymbol{\uptau}\right)^{+}$, where $\boldsymbol{\uptau}$ is a predefined threshold. As discussed for the ST operator in Appendix~\ref{sta},
this function also satisfies the Lipschitz continuity.

\subsection{Convergence Analysis for HT-ORKA}
\label{htt_orka}
Define $\mbe_i=\mbx_i-\widehat{\mbx}$, $t^{(\ell)}_{j}=r^{(\ell)}_j\tau^{(\ell)}_j$, and $\mbp^{(\ell)}_j=r^{(\ell)}_j\mba_j$. Since $\mathcal{T}_{s}(\cdot)$ is an operator that selects the best $s$-sparse approximation of the solution at each iteration, we can determine the convergence rate of HT-ORKA as follows:
\begin{equation}
\label{Khoshgele}
\begin{aligned}
\left\|\mbe_{i+1}\right\|_2&=\left\|\mathcal{T}_{s}\left(\mbz_{i+1}\right)-\widehat{\mbx}\right\|_2\\ &\leq\left\|\mathcal{T}_{s}\left(\mbz_{i+1}\right)-\widehat{\mbx}\right\|_2+\left\|\mbz_{i+1}-\widehat{\mbx}\right\|_2\\
&\leq 2 \left\|\mbz_{i+1}-\widehat{\mbx}\right\|_2.
\end{aligned}
\end{equation}
As presented earlier, the term $\left\|\mbz_{i+1}-\widehat{\mbx}\right\|^2_2$ can be bounded as
\begin{equation}
\begin{aligned}
\left\|\mbz_{i+1}-\widehat{\mbx}\right\|^2_2
\leq \left\|\mbe_i\right\|^2_2-\frac{\left(\left(t^{(\ell)}_j-\langle\mbp^{(\ell)}_j,\mbx_{i}\rangle\right)^+\right)^2}{\left\|\mbp^{(\ell)}_j\right\|^2_2},
\end{aligned}
\end{equation} 
and 
\begin{equation}
\begin{aligned}
\mathbb{E}\left\{\left\|\mbz_{i+1}-\widehat{\mbx}\right\|^2_2\right\}\leq \left\|\mbe_i\right\|^2_2-\frac{\left\|\left(\mbt-\mbP\mbx_{i}\right)^+\right\|^2_2}{\left\|\mbP\right\|^2_{\mathrm{F}}}.
\end{aligned}
\end{equation} 
According to the Hoffman bound \cite[Theorem~4.2]{leventhal2010randomized} and \eqref{Khoshgele}, we have
\begin{equation}
\label{khoshgele1}
\begin{aligned}
\mathbb{E}\left\{\left\|\mathbf{e}_{i+1}\right\|_2\right\}\leq2\left(1-\frac{1}{\kappa^{2}\left(\mbA\right)}\right)^{\frac{i}{2}} \left\|\mbe_0\right\|_2.
\end{aligned}
\end{equation}
It is straightforward to verify that a similar roadmap to the proof of Proposition~\ref{penaltyyy} may be followed.

\bibliographystyle{IEEEbib}
\bibliography{strings,refs}

\end{document}